%% file: main.tex
\documentclass[11pt]{article}

\usepackage[colorlinks=true,linkcolor=blue,citecolor=blue,urlcolor=blue]{hyperref}
\usepackage{xspace}
\usepackage{graphicx}
\usepackage{waingarten}
\usepackage[linesnumbered,ruled]{algorithm2e}
\usepackage{thmtools}
\usepackage{thm-restate}
\usepackage{color}
\usepackage{setspace}

\def\FullBox{\hbox{\vrule width 8pt height 8pt depth 0pt}}
\newcommand{\QED}{\;\;\;\FullBox}
\renewenvironment{proof}{\noindent{{\textbf{Proof:}~}}} {\hfill\QED}
\providecommand{\email}[1]{\href{mailto:#1}{\nolinkurl{#1}\xspace}}
\newenvironment{proofof}[1]{\noindent{\bf Proof of {#1}:~}}{\hfill\(\QED\)}
\newcommand{\eqdef}{\stackrel{\rm def}{=}}

\def\FullBox{\hbox{\vrule width 8pt height 8pt depth 0pt}}

\makeatletter

\newcommand{\indi}{\boldsymbol{1}}

\def\cI{\mathcal{I}}
\def\cbI{\boldsymbol{\mathcal{I}}}

\def\CC{\mathbf{C}}
\def\D{\mathbf{D}}

\def\Y{\mathbf{Y}}
\def\X{\mathbf{X}}

\def\E{\mathbf{E}}
\def\S{\mathbf{S}}

\def\ZZ{\mathbf{Z}}

\newcommand{\ComputeEMD}{\textsc{ComputeEMD}}
\newcommand{\QuadTree}{\textsc{QuadTree}}
\newcommand{\MST}{\mathsf{MST}}

\definecolor{b2}{RGB}{51,153,255}
\definecolor{mygreen}{RGB}{80,180,0}
\newcommand{\cost}{\textbf{\textsc{Cost}}}
\newcommand{\pay}{\textbf{\textsc{Pay}}}

\newcommand{\EMD}{\mathsf{EMD}}

\newcommand{\pr}[1]{\text{\bf Pr}\normalfont\lbrack #1 \rbrack} 
\newcommand{\ex}[1]{\mathbf{E}\normalfont\lbrack #1 \rbrack}

\newcommand{\ignore}[1]{}

\title{New Streaming Algorithms for High Dimensional EMD and MST}

\author {
  Xi Chen\thanks{Columbia University. \email{xichen@cs.columbia.edu}. 
  Supported by NSF IIS-1838154 and NSF CCF-1703925.}
  \and 
  Rajesh Jayaram\thanks{Google Research. \email{rkjayaram@google.com}. }
  \and
  Amit Levi\thanks{Part of this work was carried out while the author was a PhD student at the University of Waterloo. \email{amit.levi@uwaterloo.ca}. }
  \and
  Erik Waingarten\thanks{Stanford University. \email{eaw@cs.columbia.edu}. This material is based upon work supported by the National Science Foundation under Award No. 2002201.}
}

\begin{document}
\maketitle

\begin{abstract}


We study streaming algorithms for two fundamental geometric problems: computing the cost of a Minimum Spanning Tree (MST) of an $n$-point set $X \subset \{1,2,\dots,\Delta\}^d$, and~computing the Earth Mover Distance (EMD) between two multi-sets $A,B \subset \{1,2,\dots,\Delta\}^d$ of size $n$. We~consider the turnstile model, where points can be added and removed. We give a one-pass streaming algorithm for MST and a two-pass streaming algorithm~for EMD, both achieving an approximation factor of $\tilde{O}(\log n)$ and using $\polylog(n,d,\Delta)$-space only. Furthermore, our algorithm for EMD can be compressed to a single pass with a small additive error. Previously, the best known sublinear-space streaming algorithms for either problem achieved an approximation of $O(\min\{ \log n , \log (\Delta d)\} \log n)$ \cite{AIK08, BDIRW20}. For MST, we also prove that any constant space streaming algorithm can only achieve an approximation of $\Omega(\log n)$, analogous to the $\Omega(\log n)$ lower bound for EMD of \cite{AIK08}.



 Our algorithms are based on an improved analysis of a recursive space partitioning method known generically as the Quadtree. Specifically, we show that the Quadtree achieves an  $\tilde{O}(\log n)$ approximation for both EMD and MST, improving on the  $O(\min\{ \log n , \log (\Delta d)\} \log n)$ approximation of \cite{AIK08, BDIRW20}. 
\end{abstract}
\thispagestyle{empty}
\newpage
\begin{spacing}{0.75}
\tableofcontents
\end{spacing}
\thispagestyle{empty}

\newpage

\pagenumbering{arabic}
\setcounter{page}{1}

\input{intro.tex}
\input{techniques.tex}

\input{inspector.tex}

\input{Algorithm.tex}
\input{EMD-sketch.tex}
\input{sketch-plan.tex}
\input{main-lemma.tex}

\input{lower-bound.tex}

\input{l-p-norms.tex}

\section*{Acknowledgments }
We would like to thank David Woodruff, Ilya Razenshteyn, and Aleksandar Nikolov for illuminating discussions relating to this work. 

\appendix

\input{appendix.tex}

\bibliographystyle{alpha}
\bibliography{waingarten,cluster}

\end{document}

%% file: intro.tex
\section{Introduction}\label{sec:intro}

We study two fundamental geometric problems in high-dimensional spaces: the Earth Mover's distance and minimum spanning tree. 
Let $(\calX, d_{\calX})$ be a metric space. Given two (multi-)sets $A,B \subset \calX$ of size $|A| = |B| = n$, the Earth Mover's distance ($\EMD$) between $A$ and $B$ is
\[ \EMD_{\calX}(A, B) = \min_{\substack{\text{matching}\\ M \subset A \times B}} \hspace{0.2cm} \sum_{(a,b) \in M} d_{\calX}(a, b). \]
Given a single multi-set $X \subset \calX$ of size $n$, the minimum spanning tree ($\MST$) of $X$ is
\[ \MST_{\calX}(A) = \min_{\substack{\text{tree $T$} \\ \text{spanning}\\ X}} \sum_{(a,b) \in T} d_{\calX}(a,b). \]
Computational aspects of $\EMD$ and $\MST$ consistently arise in multiple areas of computer science \cite{RTG00, HTF01, PC19}, such as in computer vision \cite{bonneel2011displacement, solomon2015convolutional}, image retrieval \cite{rubner2000earth}, biology \cite{needleman1970general}, document similarity \cite{KSKW15}, machine learning \cite{arjovsky2017wasserstein,mueller2015principal,flamary2018wasserstein}, among other areas.
Their centrality in both theory and practice has motivated the theoretical study of approximate and sublinear algorithms 
\cite{C02, IT03, I04b, FIS05, AIK08, ABIW09, HIM12,  SA12, mcgregor2013sketching, AS14, BI14, ANOY14,YO14, AKR15, BBDHKLM17, S17, YV18, KNP19, BDIRW20} in both low- and high-dimensional settings. 
   
   As an illustrative example, an important application for high-dimensional EMD comes from natural language processing, particularly document retrieval and classification. A document can be represented as a collection of vectors in Euclidean space by applying \textit{word embeddings} \cite{MSCCD13, PSM14} to each of its words; these embeddings have the property that semantically similar words map to geometrically close vectors. In this context, computing the EMD between the embeddings of two documents yields a natural measure of similarity, aptly termed the \emph{Word Mover's Distance} \cite{KSKW15}.
   

   In this paper, we study streaming and sketching algorithms for computing $\EMD$ and $\MST$.
   Specifically, we consider the turnstile \textit{geometric streaming model}, introduced by \cite{I04b}, where the algorithm receives the input set $X \subset \calX$ 
   via an arbitrarily ordered sequence of insertions and deletions of points $p \in \calX$. The goal is for the algorithm to approximate a fixed function of the implicit set of points $X$ in small space, without storing $X$; ideally, one would hope for space \textit{polylogarithmic} in the number of points in $|X|$. We focus on the high-dimensional Euclidean space, where $\calX = \{1,2,\dots,\Delta\}^d$, and the distance between points is given by an $\ell_p$ norm for $p \in [1,2]$.  One can always reduce from the case of $\R^d$ to $\{1,\dots,\Delta\}^{d'}$ via standard embeddings (see Appendix \ref{app:ell_p}).
   

\paragraph{Prior Work on Sketching and Streaming $\EMD$ and $\MST$.} 
We briefly survey what is known for streaming and sketching $\EMD$ and $\MST$. We emphasize that many aspects of the sketchability and streamability of $\EMD$ and $\MST$ remain open, and obtaining tight bounds for these tasks, as well as related geometric graph problems, still remains elusive.\footnote{See Open Problems $7$ and $49$ for sketching $\EMD$ in \url{https://sublinear.info/}} 

Indyk~\cite{I04b}, building on work of \cite{C02}, was the first to formulate dynamic geometric streams and give algorithms for $\EMD$ and $\MST$ which achieved an $O(d\log \Delta)$-approximation. The result for $\MST$ was improved to a $(1+\eps)$-approximation in \cite{FIS05}, however, the resulting space complexity is exponential in the dimension, making the algorithm suitable only in low-dimensional spaces. For $\EMD$ on the plane, \cite{ABIW09} gave a $O(1/\eps)$ approximation at the cost of a $\Delta^{\eps}$ dependence in the space complexity. The best lower bound on sketching $\EMD$ on the plane is due to \cite{AKR15}, where they show that one cannot have both a constant bit and constant approximation sketch. If the sketch proceeds by an embedding into $\ell_1$, \cite{NS07} show the approximation must be $\Omega(\sqrt{\log \Delta})$. Parametrizing the approximation in terms of $n$, \cite{BI14} gave embeddings of $\EMD$ on the plane into $\ell_1$ with distortion $O(\log n)$.

For the high-dimensional regime, Andoni, Indyk, and Krauthgamer \cite{AIK08} gave an algorithm for $\EMD$ (in fact, an embedding into $\ell_1$) with approximation $O(\log n \log(d\Delta))$. Furthermore, building on an $\ell_1$-embedding lower bound of \cite{KN06}, they show that any $s$-bit sketch with approximation $\alpha > 1$  must have $s \alpha = \Omega(\log n)$. For sketching, the approximation of \cite{AIK08} may be improved to $O(\log n \min \{ \log n, \log(d\Delta)\})$ by the techniques in \cite{BI14, BDIRW20}. $\MST$ has not been formally considered in the high-dimensional regime, although we note that an $O(\log n \min\{ \log n, \log(d\Delta)\})$-approximate streaming algorithm readily applies here as well. For lower bounds on streaming high-dimensional $\MST$, nothing was known, and (prior to this work) a constant-bit stream achieving a constant approximation was possible.  

\subsection{Our Results}
In this work, we develop new algorithms and lower bounds for approximating $\EMD$ and $\MST$ in a stream. Specifically, we show that the approximation factor for these problems can be improved from $O(\log n \cdot \min\{\log n, \log(\Delta d)\})$ to $\tilde{O}(\log n)$. 
 We now state the main results of this paper. In the theorem statements which follow, we consider a fixed setting of $n, d$ and $\Delta$. The metric space consists of points in $[\Delta]^d = \{1, \dots, \Delta\}^d$ with $\ell_p$ distance for any fixed $p \in [1,2]$. We state the theorems in the \emph{random-oracle model}, i.e., any random bits stored by the algorithm do not factor into the space complexity --- we show that storing the random bits would incur at most an additive $d \cdot \polylog(n,\Delta)$ bits of space (see Section~\ref{sec:linearsketching}, where we discuss removing the random oracle assumption).\footnote{Also note that to even store a single update $p \in [\Delta]^d$, one requires $\Omega(d \log \Delta)$ bits of space. }


\begin{theorem}[$\MST$ Streaming Algorithm]\label{thm:mst-main}
There exists a turnstile streaming algorithm using at most $\polylog(n,d,\Delta)$ bits of space which, given a set $X \subset [\Delta]^d$ of size $n$, outputs $\hat{\boldeta} \in \R$ satisfying
\[ \MST_{\ell_p}(X) \leq \hat{\boldeta} \leq \tilde{O}(\log n) \cdot \MST_{\ell_p}(X) \]
with high probability.
\end{theorem}

For $\EMD$, our algorithm achieving an $\tilde{O}(\log n)$-approximation requires two passes over the data. This arises from a technical issue in the approach for $\EMD$ which is not present in $\MST$. We state the theorem in terms of two-pass streaming algorithms, and then show how to compress the two passes into one, at the cost of an additive error in the approximation. 

\begin{theorem}[$\EMD$ Two-Pass Streaming Algorithm]\label{thm:emd-main-twopass}
	Given two multi-sets $A, B \subset [\Delta]^d$ of size $n$ there exists a two-pass turnstile streaming algorithm using $\polylog(n,d,\Delta)$ bits of space which outputs $\hat{\boldeta} \in \R$ satisfying
	\[ \EMD_{\ell_p}(A,B) \leq \hat{\boldeta} \leq \tilde{O}(\log n) \cdot \EMD_{\ell_p}(A, B) \]
	with high probability.
\end{theorem}

\begin{theorem}[$\EMD$ One-Pass Streaming Algorithm]\label{thm:emd-main}
Given two multi-sets $A, B \subset [\Delta]^d$ of size $n$ and any $\eps > 0$, 
there exists a turnstile streaming algorithm using $O(1 /\epsilon)\cdot \polylog(n,d,\Delta)$ bits of space which outputs $\hat{\boldeta} \in \R$ satisfying
\[ \EMD_{\ell_p}(A,B) \leq \hat{\boldeta} \leq \tilde{O}(\log n) \cdot \EMD_{\ell_p}(A, B)  + \eps d \Delta n. \]
with high probability.
\end{theorem}

We encourage the reader to think of instances where $A$ and $B$ are size-$n$ subsets of the hypercube $\{0,1\}^d$ with $\ell_1$ distance (i.e., $\Delta = 2$ and $p = 1$). This setting captures all the complexity encountered in this work. For $\Delta > 2$ and $p \in (1,2]$, the algorithm first applies an embedding into $\{0, 1\}^d$ with $\ell_1$ (see Appendix~\ref{app:ell_p}). 

Regarding the additive error in Theorem~\ref{thm:emd-main}, while an appropriate setting of $\eps$ may absorb the additive error into relative error, we leave as an open problem whether this additive error may be removed completely in one-pass algorithms. For instance, if the points do not overlap almost always, i.e., when $|A \cap B| / |A \cup B| \leq 1 - \eps_0$, then $\EMD_{\ell_p}(A, B) \geq \eps_0 n$, and $\eps$ may be set to $\eps_0 / d \Delta$ in order to absorb the additive error into the relative error by increasing the space by a factor of $d \Delta$, and keeping a poly-logarithmic dependence on $n$. From a practical perspective, the fact that points do not overlap may be a reasonable assumption to make.

All of our streaming algorithms are \emph{linear sketches}, meaning that they  store only the matrix-vector product $\bS f$ for some randomized $\bS \in \R^{k \times n}$, where $f = f_X \in \R^{\Delta^d}$ is the indicator vector (with multiplicity) of $X$ for the case of $\MST$, and $f= f_{A,B} \in \R^{2 \cdot \Delta^d}$ is the indicator vector (with multiplicity) of $A,B$ for $\EMD$. Linear sketches are an important class of turnstile streaming algorithms, and have many well-known and studied advantages. For instance, such sketches directly resulted in algorithms for distributed computation such as the MPC model, as well as algorithms for multi-party communication. Our results, therefore, can be applied in a natural way to these models as well. 



\paragraph{Improved Analysis of the Quadtree.} 
The prior sketching algorithms are based on a hierarchical partitioning method known as the \textit{Quadtree}.\footnote{The name \textit{Quadtree} is an artifact of the study of the algorithm originally in the planar (two-dimensional) case, in which the algorithm recursively partitions the plane into quadrants. Our Quadtrees, being in high dimensions, will partition space into more than $4$ parts at a time. However, since they are the natural generalization of the planar case, it is common to refer to the generic method as Quadtree regardless of dimension. } Here, we refer to quadtrees as a generic class of methods that embed points from $\calX$ into a randomized tree by recursively partitioning the space. 
%
At a high level, the Quadtree algorithm recursively and randomly partitions the space $\calX$, which results in a rooted (randomized) tree. Each point in the set $X$ for the case of $\MST$, or $A \cup B$ for $\EMD$, is sent down to a leaf of the tree. From there, a spanning tree or a matching, is constructed in a bottom-up fashion. Each point ``walks up the tree'' and is greedily connected (in the case of $\MST$), or matched (in the case of $\EMD$) as it encounters other points. This results in a very efficient \textit{offline} (non-sketching) algorithm. The recent work of \cite{BDIRW20} study the quadtree algorithm explicitly, where they call it ``Flowtree,'' and showed it has favorable practical properties. From a theoretical point-of-view, the approximation incurred by these methods were the bottleneck in prior works for sketching and streaming $\EMD$ and $\MST$, here, we improve this analysis of \cite{AIK08, BDIRW20} from $O(\log n \min \{ \log n, \log(d \Delta) \})$ to $\tilde{O}(\log n)$. 
\begin{theorem}[Quadtree Methods (Informal)]
Given two multi-sets $A, B \subset [\Delta]^d$ of size $n$, the ``Flowtree'' algorithm of \cite{BDIRW20} outputs an $\tilde{O}(\log n)$-approximation to $\EMD_{\ell_1}(A,B)$ with probability at least $0.9$. Similarly, given a multi-set $X \subset [\Delta]^d$ of size $n$, the greedy, bottom-up spanning tree is an $\tilde{O}(\log n)$-approximation to $\MST_{\ell_1}(X)$ with probability at least $0.9$.
\end{theorem}

\paragraph{Lower bounds for MST.}
For lower bounds, \cite{AIK08} shows that any randomized $\ell$-bit streaming algorithm distinguishing $\EMD_{\ell_1}(A,B) \geq r$ and $\EMD_{\ell_1}(A,B) \leq r/\alpha$ with probability at least $2/3$ must satisfy $\alpha \ell = \Omega(d)$, where the instances used have $d = \log n$. For a qualitative comparison, estimating $\ell_1$ norm does admit such $O(1)$-approximation, $O(1)$-bit space streaming algorithms (with public randomness), implying that $\EMD$ is a harder problem. We show an analogous lower bound for $\MST$ in the streaming model.

\begin{theorem}\label{thm:mst-lb-main}
Any randomized $\ell$-bit streaming algorithm which can distinguish whether a size-$n$ set $X \subset \{0,1\}^d$ has $\MST_{\ell_1}(X) \geq nd/3$ or $\MST_{\ell_1}(X) \leq nd / \alpha$ with probability at least $2/3$ must satisfy
$\ell + \log \alpha = \Omega(\log n / \alpha).$
Moreover, this holds even in the insertion-only model, where points are only added to $X$ in the stream.
\end{theorem}
We emphasize that, prior to Theorem \ref{thm:mst-lb-main}, there were no lower bounds known for streaming $\MST$ --- not even an $\Omega(1)$ lower bound was known on the approximation of a constant-bit algorithm. We note that \cite{AIK08} actually considers the (stronger) two-party communication setting for $\EMD$, where each player receives one of the sets. The two-party communication game for $\MST$ where each player receives half of the set $X$ is insufficient, as there is simple $O(1)$-approximation, constant-bit protocol.\footnote{Intuitively, the players may compute the cost of their $\MST$ locally, and compute the distance between two arbitrary points. The sum of these quantities is a $3$-approximation to the $\MST$ of the entire set.} Therefore, out theorem will crucially involve the streaming nature of the algorithm.

%% file: techniques.tex
\def\vv{\mathbf{v}}

\subsection{Technical Overview}
\label{sec:techniques}

\subsubsection{The Main Idea: Tree Embeddings with Data-dependent Edge Weights} 

In \cite{I04b}, Indyk described an approach for streaming a variety of graph problems (including $\MST$ and $\EMD$) in discrete geometric spaces, leading to $O(d\log \Delta)$-approximations for these problems in the metric space $[\Delta]^d$ with $\ell_1$ distance. This approach, later refined in \cite{AIK08}, forms the basis of our work, so we give a very high level overview in order to highlight the new ideas. For simplicity, we describe it for $\EMD$, as the high-level picture for $\MST$ is similar.

A streaming algorithm for $\EMD$ with sets $A, B \subset [\Delta]^d$ of size $n$ may proceed in the following way:
\begin{enumerate}
\item\label{en:step1} Sample a recursive random partition of the space, broadly referred to as a quadtree, which specifies an embedding of the original space $[\Delta]^d$ into a rooted tree. For example, when $d = 2$, one may sample $\log_2 \Delta$ randomly shifted, nested square grids of side length $\Delta/2, \Delta /4 , \dots, 1$ and arrange them into a rooted tree of depth $\log_2 \Delta + 1$. Each node corresponds to a region of the space, where the root contains the entire space, and the children of a node have regions which partition the region of the parent. The points in $A$ and $B$ are assigned to leaves of this tree, according the regions where points fall, and the quadtree implicitly defines a matching $M$ between $A$ and $B$ given by the natural bottom-up greedy procedure. Having implicitly specified a matching $M$, the goal of the streaming algorithm will be to approximate the cost of $M$. 
\item\label{en:step2} In order to do so, \cite{I04b, AIK08} maintains a high-dimensional vector which implicitly encodes the matching $M$. Specifically, the vector has a coordinate for each edge of the quadtree, and the entry in each coordinate is the number of points from $A$ falling within the region of the child minus the number of points in $B$ falling within the region of the child. Furthermore, the $\ell_1$-norm of the vector, where each coordinate of an edge is scaled by some edge weight (for example, by the size of the parent region) gives an approximation of the cost of $M$. Thus, this gives an $\ell_1$-embedding for $\EMD$ over $[\Delta]^d$, and known algorithms for streaming the $\ell_1$-norm can be applied.
\end{enumerate}
With the above approach in mind, there are two steps involved in showing the approximation guarantee: (i) showing the matching $M$ in Step~\ref{en:step1} has approximately optimal cost, and (ii) showing that the appropriate scalings of coordinates reduce approximating the cost of the matching $M$ to an $\ell_1$-computation. We note that even though the above presentation is a two-step procedure, \cite{I04b, AIK08} do not present it this way. In fact, de-coupling the matching $M$ from the method to approximate the cost of $M$ is an important conceptual contribution which is made explicit in~\cite{BDIRW20}, which led us to revisit the $\EMD$ problem. 

Prior to our work, (i) proceeded by the method of tree embeddings. One assigns the edge weights to the quadtree and interprets it as a tree embedding of the metric $([\Delta]^d, \ell_1)$. By studying the distortion of this embedding, one bounds the cost of $M$. The edge weights chosen in \cite{I04b} (building on work of \cite{C02}) embed $([\Delta]^d, \ell_1)$ with distortion $O(d \log \Delta)$, which will become the approximation. Refining the approach, \cite{AIK08} show that another choice of edge weights (better suited for high-dimensional spaces) embeds subsets of $([\Delta]^d, \ell_1)$ with bounded average distortion which suffices for an $O(\log n\log(d\Delta))$ bound on the cost of $M$. 
Given the bound on $M$ with respect to a fixed tree metric, (ii) is straight-forward: since the fixed tree metric specifies the scalings of the vector, and approximating the cost of $M$ amounts to an $\ell_1$-norm computation. 

Our main contribution is two-fold. First, we show how to go beyond the distortion argument in (i) to show that the cost of $M$ is a $\tilde{\Theta}(\log n)$-approximation to $\EMD$ with probability $0.9$. To do so, we study a \emph{data-dependent} notion: instead of fixing the edge weights as in \cite{I04b, AIK08}, we allow the edge weights to depend on the input $A \cup B$. The use of data-dependent edge weights implies $M$ is actually a better quality matching than what the method of tree embeddings specified. The data-dependent edge weights are (relatively) simple: the weight of an edge $(u, v)$, where $u$ is the parent of $v$, is the average distance between a randomly sampled point of $A \cup B$ within the region of $u$ and a randomly sampled point of $A \cup B$ within the region of $v$. However, the fact these data-dependent edge weights yield an improved upper bound on the cost of $M$ constitutes the bulk of the work in Sections~\ref{sec:quadtree-cost} and~\ref{maintechlemma}.

Unfortunately, the introduction of data-dependent edge weights breaks Step~\ref{en:step2}. Now, approximating the cost of $M$ with the data-dependent weights is no longer as simple as an $\ell_1$-computation. The coordinates of the vector remain the same, however, the scaling of each coordinate depends on additional structure of the points. Importantly, data-dependent edge weights do not result in an $\ell_1$-embedding, and we cannot use known $\ell_1$-sketching algorithns. This takes us to our algorithmic contribution, where we design the sketching algorithms for Step~\ref{en:step2} with data-dependent edge weights. More generally, we introduce a two-step template for transforming data-dependent costs in the Quadtree into streaming algorithms. Conceptually, the approach generalizes the well-known $\ell_p$ sampling problem \cite{monemizadeh20101,andoni2010streaming,Jowhari:2011,JW18} to \textit{$\ell_p$-sampling with meta-data}. For $\EMD$, the high-level idea is the following: first, sample a coordinate of the vector proportional to the $\ell_1$-distribution (i.e., the $\ell_1$-sampling problem), and second, estimate the data-dependent edge weight for the coordinate sampled (the meta-data), so that we can scale the contribution of that coordinate appropriately.

\input{algo-techniques}
\input{sketching-techniques}

%% file: algo-techniques.tex

\def\avgd{\text{avg}} \def\cc{\mathbf{c}} \def\bvalue{\textbf{Value}}
\def\bsv{\boldsymbol{\mathsf{v}}}  \def\sv{\mathsf{v}}

\subsubsection{Implementing Step 1: Quadtree Matching with Data-dependent Edge Weights}
We begin by describing our improved analysis of the randomized space partitioning algorithm, \textit{Quadtree}. 
For the sake of simplicity, we focus on its analysis in the context of approximating $\EMD$; the same ideas work similarly for $\MST$. We begin by more formally introducing the Quadtree in high-dimensional spaces. In what follows, we focus on the case when the metric space is the hypercube with the Hamming distance, i.e. $A,B \subset\{0,1\}^d$ and $d(p,q) = \|p-q\|_1$ for $p,q \in \{0,1\}^d$. For the approximation, this is without loss of generality. One may embed $(\R^d, \ell_p)$ into $\{0,1\}^d$ by increasing the dimension (see Appendix~\ref{app:ell_p}), which is irrelevant since the approximation we will show is dimension-independent.

{\bf Quadtree. }The Quadtree algorithm creates a randomized tree $\bT$ with depth  $h:=\log_2 2d$ by recursively sub-dividing the hypercube $\{0,1\}^d$. Therefore, each node $u$ in $\bT$ will be associated with a subcube $S_u \subseteq \{0,1\}^d$, where the root $r$ has $S_r = \{0,1\}^d$. To create these subcubes,  each internal node $u$ of $\bT$ at depth $j<h$ is labeled with an ordered tuple 
of $2^j$ coordinates $(i_1,\ldots,i_{2^j})\in [d]$ (which are not necessarily distinct),
and has $\smash{2^{2^j}}$ children. Each of the $2^{2^j}$ children of $u$ will \textit{uniquely} correspond to one of the $2^{2^j}$ fixings of the coordinates $(i_1,\ldots,i_{2^j}) \in \{0,1\}^{2^j}$. Specifically, each child $v$ of $u$ is assigned a unique bit-string $(b_1,\dots,b_{2^j}) \in \{0,1\}^{2^j}$. The child $v$ then corresponds to the subcube $S_v\subseteq S_u$ obtained by fixing the $i_{t}$-th coordinate to $b_t$, for each $t=1, \dots,2^j$. We now describe the procedure for generating a random Quadtree $\bT $:

	\begin{flushleft}\begin{enumerate}
		\item Uniformly sampling a tuple $(i_1, \dots,i_{2^j}) \in [d]^{2^j}$ of $2^j$ coordinates independently for each node $u$ at depth $j \in \{0,1,\dots,h-2\}$ to use as its label.
		\item Setting $(1,\ldots,d)$ as the label of every node at depth $h-1$.
	\end{enumerate}\end{flushleft}

A Quadtree $\bT$ defines a map $\varphi$ from $\{0,1\}^d$ to leaves of $\bT$: $\varphi(p)=v$ if $p\in S_v$.
Given $A$ and $B$, we write $A_v$ and $B_v$ to denote $A\cap S_v$ and $B\cap S_v$ for each node $v$ in $\bT$.



{\bf Depth-greedy Matching from Quadtree.} Given a random Quadtree $\bT$, one obtains a 
  natural depth-greedy matching as follows: 
We first map all points in $C=A\cup B$ to leaves of $\bT$ using $\varphi$.
Then, we greedily match  points between $A$ and $B$ in a bottom up fashion, by walking
  each point up the tree level-by-level, and at each node one arbitrarily matches as many of 
  the unmatched points from $A$ and $B$ as possible.
Let $\bM$ be any depth-greedy matching obtained from $\bT$ in this fashion.
The goal of our improved analysis of the Quadtree for $\EMD$ is to show that 
$$
\EMD(A,B)\le \cost(\bM)\eqdef \sum_{(a,b)\in \bM} \|a-b\|_1\le \tilde{O}(\log n)\cdot \EMD(A,B)
$$
with high probability (over the randomness of $\bT$). Note that the first inequality is trivial.

{\bf Analysis of Quadtree via Tree Embeddings.}
Before presenting an overview of our new techniques,
  it will be helpful  to begin with a~recap of the analysis of \cite{AIK08}
  which can be used to show that $\cost(\bM)\le O(\log n\log d)\cdot \EMD(A,B)$.
The analysis of \cite{AIK08} starts by assigning a weight of $d/2^i$ to each edge from a node at depth $i$ to 
a node at depth $i+1$ in $\bT$. This defines a
metric embedding $\varphi:A\cup B \to \bT$ by mapping each point to a leaf of $\bT$. 
The choice of edge weights is motivated by the observation that two points $x,y\in \{0,1\}^d$
with $\|x-y\|_1=d/2^i$ are expected to have their paths diverge for the first time at depth $i$.
If this is indeed the case then 
$d_{\bT}(\varphi(x),\varphi(y))$ would capture $\|x-y\|_1$ up to a constant.


To upperbound $\cost(\bM)$, one studies the distortion of this embedding. Firstly, for any $\lambda >1$ and $x,y \in \{0,1\}^d$, it is easy to verify that distances in the tree metric do not contract much:
\[ \Prx_{\bT}\left[ d_\bT(\varphi(x),\varphi(y))<  \frac{1}{\lambda}\cdot  \|x-y\|_1 \right] \leq \left( 1 - \frac{\|x-y\|_1}{d} \right)^{1 + 2 + \dots + 2^{\left\lfloor \log_2 \left(\frac{\lambda d}{\|x-y\|_1}\right) \right\rfloor}} \leq 2^{-\Omega(\lambda)}. \] 
Thus by a union bound, for all $x,y \in A \cup B$ we have
\begin{equation}\label{worsteq}
\|x-y\|_1\le O(\log n)\cdot d_\bT(\varphi(x) ,\varphi(y))   
\end{equation}
with probability at least $1-1/\poly(n)$,
which essentially means that we can assume (\ref{worsteq}) in the worst case. 
As a result, we have 
$$
\cost(\bM)=\sum_{(x,y)\in \bM} \|x-y\|_1
\le O(\log n) \sum_{(x,y)\in \bM} d_\bT(\varphi(x),\varphi(y))
\le O(\log n) \sum_{(x,y)\in M^*} d_\bT(\varphi(x),\varphi(y)),
$$
where the last inequality holds for any matching $M^*$ between $A$ and $B$ given that 
  the depth-greedy matching is optimal under the tree metric.
Setting $M^*$ to be the optimal matching between $A$ and $B$ under the original $\ell_1$ metric,
  we finish the proof by upperbounding $d_\bT(\varphi(x),\varphi(y))$ using~$O(\log d)$ $\|x-y\|_1$.
To see this, when $\|x - y\|_1 = \Theta(d/2^j)$, the probability that paths of $x,y$ diverge at level $j-k$ is $\Theta(2^{-k})$ for each $k$, and when it does, $d_{\bT}(\varphi(x), \varphi(y)) = \|x-y\|_1 \cdot \Theta(2^k)$. Since $j \leq h = O(\log d)$,
\begin{align} 
\E\big[{d_\bT(\varphi(x) ,\varphi(y))}\big] \leq \|x-y\|_1 + \sum_{k=0}^{j} \Theta(2^{-k}) \cdot \|x-y\|_1 \cdot \Theta(2^k) = O(\log d)\cdot \|x-y\|_1. \label{eq:grow-dist}
\end{align}
Together they yield the aforementioned $O(\log n \log d)\cdot \EMD(A,B)$ upper bound for $\cost(\bM)$.\footnote{\label{usefulfootnote}The reason that this analysis can achieve approximation $O(\min\{\log n, \log d\} \log n)$, as opposed to $O(\log n \log d)$ is that with probability $1 - 1/n$, every $x,y \in A \cup B$ with $\|x-y\|_1 = \Theta(d/2^j)$ diverges at depth after $j - O(\log n)$.}

\def\bc{\mathbf{c}}  \def\bv{\mathbf{v}}

{\bf Tree Embeddings with Data-dependent Edge Weights.}
We show how to go beyond the distortion arguments of \cite{AIK08} by studying
  a tree embeddings with data-dependent edge weights. In what follows, for any vertex $u \in \bT$, let $C_u = A_u \cup B_u$ be the set of all points which map through $u$ (recall $C= A\cup B$).
The weight we assign to each edge $(u,v)$\footnote{We always use $u$ in $(u,v)$ to denote the parent
  and $v$ to denote the child.} of $\bT$ will no longer be a fixed number $d/2^i$ but
$$
\avgd_{u,v}\eqdef \Ex_{\substack{\cc\sim C_u\\ \cc'\sim C_v}} \big[\|\cc-\cc'\|_1\big],
$$ 
i.e., the average distance between a point drawn randomly from $C_u$ and a point drawn randomly from $C_v$;
when $C_v=\emptyset$ we define $\avgd_{u,v}=0$ by default.
Let $d^*_\bT$ denote the tree metric under this new set of weights.
Again, the depth-greedy matching $\bM$ we are interested in is optimal and 
  the cost of $\bM$ under the new tree embedding can be expressed as
$$
\bvalue_\bT(A,B)\eqdef \sum_{(u,v)\in E_T} \big||A_v|-|B_v|\big|\cdot \avgd_{u,v}.
$$
where $E_T$ is the set of edges of $T$.\footnote{We remark that one can define an analogous quantity for the case of $\MST$,  where given a single set $X \subset [\Delta]^d$, we set $\bvalue_{\bT}(X) = \sum_{(u,v)\in E_T} \mathbf{1}(|X_v|)\cdot \avgd_{u,v}$, where $\mathbf{1}: \R \to \{0,1\}$ is the indicator function (i.e., $\mathbf{1}(x) = 0$ if and only if $x=0$). It is this quantity that we will analyze in our results for $\MST$.}
On the one hand, $\bvalue_\bT(A,B)$ is at least $\cost(\bM)$ given that for any $x,y\in C$, we always have
  $\|x-y\|_1 \le d_\bT^*(\varphi(x),\varphi(y))$ by triangle inequality. 
On the other hand, $\bvalue_\bT(A,B)$ is at most $\sum_{(a,b)\in M^*} d_\bT^*(\varphi(a),\varphi(b))$
  for any matching $M^*$ and in particular, the optimal matching $M^*$ under the $\ell_1$ metric.
As a result, it suffices to upperbound the cost of $M^*$ under the data-dependent~tree embedding
  by $\tilde{O}(\log n)\cdot \cost(M^*)$ given that $\cost(M^*)=\EMD(A,B)$.
To this end it suffices to show that the expectation of $d_\bT^*(\varphi(a),\varphi(b))$ 
   for any $a,b\in C$ can be bounded from above by $\tilde{O}(\log n)\cdot \|a-b\|_1$.

{\bf Inspector Payment.} Fix $a,b\in C$. We introduce the following quantity as the \emph{inspector payment}
  of $(a,b)$ with respect to the Quadtree $\bT$.
(We imagine the process as first drawing the Quadtree and then an ``inspector'' who examines the tree to 
  track down $a$ and $b$, making payments accordingly.)
Formally we let $\left( \sv_{0}(x), \sv_{1}(x), \dots, \sv_{h}(x)\right)$
denote the root-to-leaf path of $x$ in a Quadtree $\bT$. Then 
\begin{align}
\pay_{\bT}(a,b)  \eqdef \sum_{i\in [h]} \ind\big\{\sv_{i }(a) \neq \sv_{i }(b) \big\} 
\cdot \left(\Ex_{\cc\sim C_{\sv_{i-1}(a)}}\big[\|a-\cc\|_1\big]+
\Ex_{\cc\sim C_{\sv_{i-1}(b)}}\big[\|b-\cc\|_1\big]\right).
 \label{eq:inspecttechniques}
\end{align}
Intuitively, this payment scheme corresponds to an inspector who tracks down $a$ and $b$ from the root of $T$, and whenever $a$ and $b$ first diverge in the tree at node $u$, pays for $a$ the average distance between $a$ and a random point drawn from
  $C_v$ for every node along the $u$-to-leaf path (including $u$); the inspector pays for $b$ similarly. 
It again follows from triangle inequality that $2\cdot \pay_\bT(a,b)$ is at least 
  $d_\bT^*(\varphi(a),\varphi(b))$. So it suffices to bound the expectation of 
  $\pay_\bT(a,b)$ by $\tilde{O}(\log n)\cdot \|a-b\|_1$.
  

\ignore{once the paths of $a$ and $b$ in $T$ diverge at some depth $j$ (i.e. $v_j(a)=v_j(b)$ but $v_{j+1}(a)\ne
v_{j+1}(b)$), the inspector 
starts paying for twice the distance between $a$ and $c_{v_i(a)}$ at every node $v_i(a)$ with $i \geq j$. (The
twist we hinted earlier is that the
inspector does not need to pay at all depths but only at and below depth $j$ where the two paths diverge.)}
Before giving a sketch of this proof, which is the most challenging part of our Quadtree analysis,
we note that
  the inspector payment (\ref{eq:inspecttechniques}) depends on the data $A$ and $B$, as well as the Quadtree $\bT$ in two ways. The first is the depth when $a$ and $b$ first diverge, captured by the indicator $\ind\{ \sv_i(a) \neq \sv_i(b)\}$. The second is the average distance between $a$ and $C_v$, which not only depends on $a$, but also on global  properties of $C=A \cup B$. At a high level, incorporating this second aspect is the main novelty, since the average distance between $a$ and $C_v$ is an \emph{average} notion of radii at $v$. 
 Therefore, if the inspector pays a large amount, then an average point in $C_v$ is far from $a$ (as opposed to the farthest point implied by worst-case radii).


{\bf Bounding Inspector Payments.}  
Consider fixed $a, b\in C$ at distance $\|a-b\|_1 = \Theta(d/2^j)$, and we  
  give some intuition behind our upper bound on the expectation of the $a$-part of the payment:
$$ \sum_{i\in [h]} \ind\big\{\sv_{i }(a) \neq \sv_{i }(b) \big\} 
\cdot \avgd_{a,i-1},\quad\text{where}\quad
\avgd_{a,i-1}\eqdef\Ex_{\cc\sim C_{\sv_{i-1}(a)}}\big[\|a-\cc\|_1\big].
$$
We will ignore the indicator random variable $\ind\big\{\sv_{i }(a) \neq \sv_{i }(b) \big\}$
  and use linearity of expectation to focus on $\E_\bT[\hspace{0.03cm}\avgd_{a,i}\hspace{0.03cm}]$.
(With the indicator random variable, we need to consider the expectation of $\avgd_{a,i}$
  conditioning on the event that $a,b$ have diverged. The conditioning will
 not heavily influence the geometric intuition, so we will ignore this for the rest of this overview). 
 
 Let $\sv_i = \sv_i(a)$. 
Similar to worst-case bounds on radii,  $\E_\bT[\hspace{0.03cm}\avgd_{a,i}\hspace{0.03cm}]$ can still be $d/2^i \cdot \Omega(\log n)$.  
As~an example, let $i_1$ be a relatively large depth and for some small $\eps \approx 10^{-6}$, consider a set $P_1$ of $n^{\epsilon}$~many points at distance $\epsilon \log n \cdot d / 2^{i_1}$ around $a$. Then, at depth $i_1$ of a random Quadtree $\bT$, a point in $P_1$ traverses down to node $\sv_{i_1}$ with non-negligible probability, roughly $1/n^{-\eps}$. If no other points lie closer to $a$ than those in $P_1$, then $\E_\bT[{\hspace{0.03cm}\avgd_{a,i_1}\hspace{0.03cm}} ]= d/2^{i_1} \cdot\Omega(\eps \log n)$, since it is likely that some points of $P_1$ make it to $\sv_{i_1}$ and significantly increase
  the average distance between $a$ and $C_{\sv_{i_1}}$. If~this happened on $a$ for every depth $i$, the inspector would be in trouble, as there are $O(\log d)$ levels and a similar argument to that of worst-case radii would mean a payment of 
  $O(\log d \log n)\cdot \|a-b\|_1$. 

However, we claim if the arrangement of $P_1$ resulted in $\E_\bT[\hspace{0.03cm}\avgd_{a,i_1}\hspace{0.03cm}]
 = d/2^{i_1} \cdot \Omega(\eps \log n)$, the same situation will be a lot more difficult to orchestrate for depth $i_2 \leq i_1 - O(\log \log n)$. In particular, at depth $i_2$, in order to have $\E_\bT[\hspace{0.03cm}\avgd_{a,i_2}\hspace{0.03cm}] = d/2^{i_2} \cdot \Omega(\eps \log n)$, there must be a set of points $P_2$ at distance $d / 2^{i_2} \cdot \Omega(\eps \log n)$ which cause $\avgd_{a,i_2}$ to be large. However, it is no longer enough to have $|P_2| = n^{\eps}$. The reason is that points of $P_1$ in $\sv_{i_2}$ will help bring down the average distance. Since points in $P_1$ are at distance $\eps \log n \cdot d/2^{i_1} \ll d/2^{i_2}$ from $a$, there will oftentimes be $\Omega(n^{\eps})$ points from $P_1$ in $\sv_{i_2}$. In order to significantly increase the average distance, $\sv_{i_2}$ must oftentimes have at least $n^{\eps}/\polylog(n)$ points from $P_2$; otherwise, $\avgd_{a,i_2}$ will be mostly the average distance between $a$ and points in $P_1$. Since any given point from $P_2$ traverses down to $\sv_{i_2}$ with probability roughly $1/n^{\eps}$, we must have $|P_2| \geq n^{2\eps} / \polylog(n)$. This argument can only proceed for at most $O(1/\eps)$ depths before $|P_{O(1/\eps)}| > 2n$, in which case we obtain a contradiction, since all points are in $A \cup B$.

\ignore{\begin{figure}
\centering
\begin{picture}(300, 150)
\put(0,0){\includegraphics[width=0.7\linewidth]{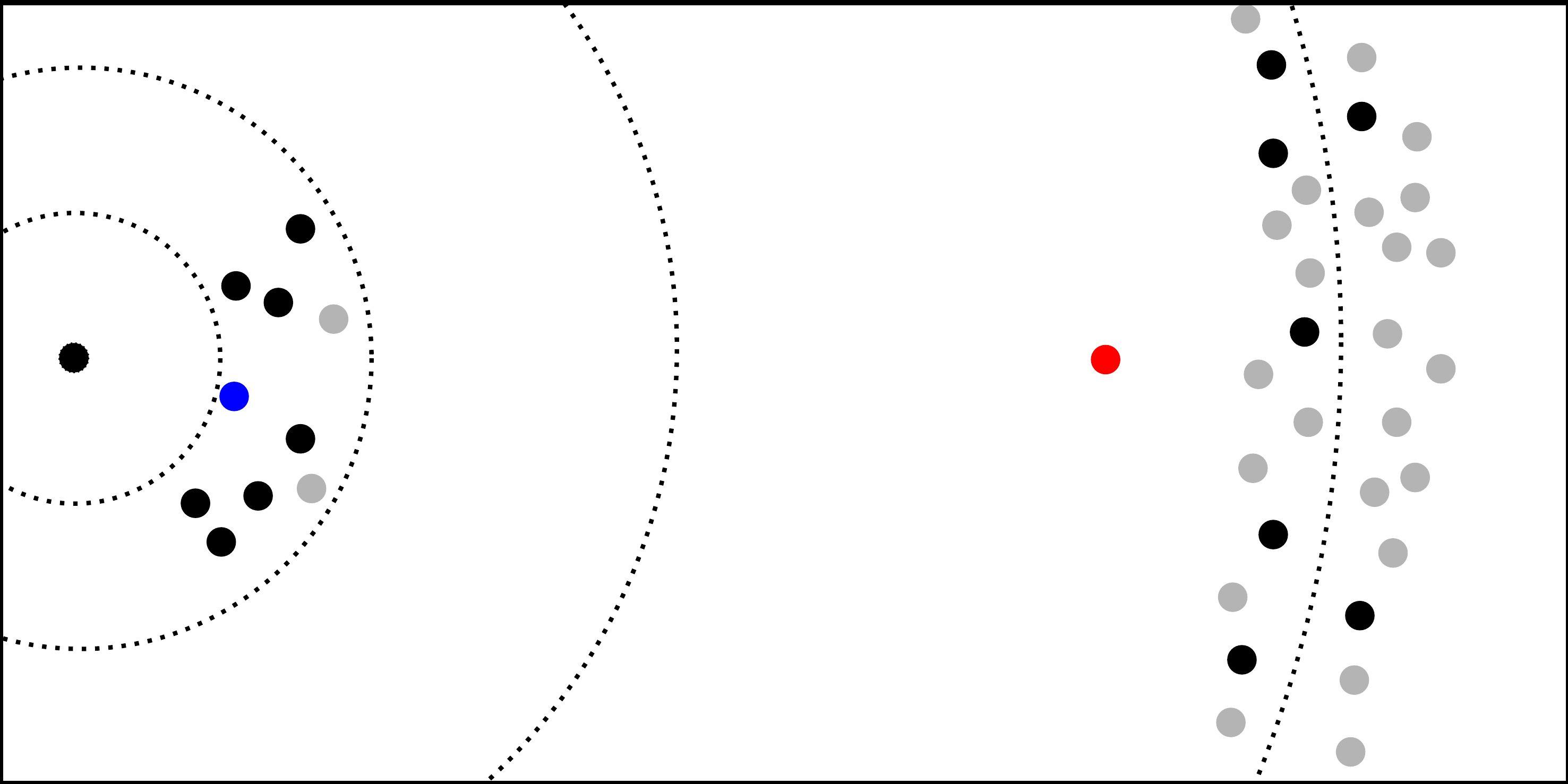}}
\end{picture}
\end{figure}}

\ignore{ corresponding to the distance $\log n$, the vertex $\bv_{\ell_1}$ will now contain nearly all the $n^\epsilon$ points from $P_1$.  Now suppose that the expectation at level $\ell_1$ is still large, say $\ex{\|x - \bc_{\ell_1}\|_1} = \Omega(\eps \log^2 s)$. Then because 
 $\bv_{\ell_1}$ now contains at least $n^\epsilon$ points which are at distance at most $\epsilon \log n$ from $a$, for the \textit{center of mass} $\bc_{\ell_1}$ to be far, there must be at least $n^\epsilon$ \textit{new} points at distance $\eps \log^2 s$ from $a$ which landed in $\bv_{\ell_1}$.
 But any such ``far point'' only lands in $\bv_{\ell_1}$ with probability $1/n^{-\epsilon}$. So for this to happen, we would have needed an even larger set $P_2$ of at least $n^{2 \epsilon}$ points to be contained in a ball of radius $\eps \log^2 s$ around $x$. 
 Then the vertex $v_{\ell_2}$, where $\ell_2 = h - 2 \log \log n$, would contain at least $n^{2\epsilon}$ points, and so on from here.}
 
Generally, in order to increase the average distance between $a$ and $C_{\sv_i}$ \emph{multiple times} as the depth $i$ goes down, the number of points around $a$ at increasing distances must grow very rapidly. More specifically, we show that if a depth $i$ is ``bad,''
 meaning that $\E_\bT[\hspace{0.03cm}\avgd_{a,i}\hspace{0.03cm}] \geq \alpha\cdot  d/2^i$
 for some $\alpha =\omega(\log \log n)$, then the number of points within a ball of radius $d/(2^i \log n)$ around $a$ and within a larger ball of radius $O(\log n\cdot  d/2^i)$ around $a$ must have increased by a factor of $\exp(\Omega(\alpha))$; this means the number of such depths $i$ is at most $((\log n) /\alpha) \cdot \poly(\log\log n)$. Combining this analysis and the fact that $a$ and $b$ must diverge in order to incur payment from the inspector, we obtain
  our upper bound that the expectation of $\pay_\bT(a,b)$ is at most $\tilde{O}(\log n)\cdot \|a-b\|_1$.
 

%% file: sketching-techniques.tex

\subsubsection{Implementing Step $2$: From Quadtree to Sketching Algorithms}

By the prior discussion, after sampling a Quadtree $\bT$, we know that the quantity $\bvalue_\bT(A,B)$ is a $\tilde{O}(\log n)$ approximation of the true cost $\EMD(A,B)$. Specifically, we have:

\begin{equation}\label{introeq}
	\EMD(A,B) \le \bvalue_\bT(A,B)\le \tilde{O}(\log n)\cdot  \EMD(A,B)
\end{equation}
 Thus, the approach of our sketching algorithm is simply to approximate $\bvalue_\bT(A,B)$. We will decompose $\bvalue_\bT(A,B)$ based on its level: $\bvalue_\bT(A,B) = \sum_{i=1}^h \bvalue_{\bT,i}(A,B)$, where 
\[\bvalue_{\bT,i}(A,B)\eqdef \sum_{\substack{(u,v)\in E_T \\ \textbf{depth}(u,v) = i }} \big||A_v|-|B_v|\big|\cdot \avgd_{u,v}\] where \textbf{depth}$(e)$ for an edge $e \in \bT$ is the depth of the child vertex in $e$. We will attempt to estimate each $\bvalue_{\bT,i}(A,B)$ independently for each $i$, so in what follows we now fix any level $i \in [h]$.

We start with some notation. For any (non-root) vertex $v \in \bT$, let $\pi(v)$ be the parent of $v$ in $\bT$. We then define the \textit{discrepancy} vector for level $i$, denoted $\Delta^i$, by $\Delta^i_v = |A_v|-|B_v|$ for every vertex $v$ at depth $i$ of the tree (i.e., $\Delta^i$ has a coordinate $\Delta_v^i$ for each vertex $v$ at depth $i$). Next, for any vector $x \in \R^N$ and any $p\geq  0$, we define the $\ell_p$ distribution $\calD_p(x)$ over the coordinates of $x$ via $\calD_p(x) = \left(\frac{|x_1|^p}{\|x\|_p^p},\frac{|x_2|^p}{\|x\|_p^p},\dots,\frac{|x_N|^p}{\|x\|_p^p}\right)$ for $p>0$, and for $p=0$ we define $\calD_0(x)$ to be the uniform distribution over the support of $x$. Now observe:\footnote{We remark that for the case of $\MST$, the relevant quantity $\bvalue_{\bT,i}(X)$ below can be written as $ \|\Delta^i\|_0 \cdot \Ex_{ v\sim  \calD_0(\Delta^i)} \left[	\avgd_{\pi(v),v} \right]$. Namely, we simply replace the $\ell_1$ norm in both the scaling and the distribution by the $\ell_0$ norm. Thus, the high-level approach to sketching $\MST$ will be similar. However, due to using the $\ell_0$ instead of the $\ell_1$ norm, an entirely different set of techniques will be required to implement each of the steps. }

\[\bvalue_{\bT,i}(A,B) =  \|\Delta^i\|_1 \cdot \Ex_{ \vv\sim  \calD_1(\Delta^i)} \left[	\avgd_{\pi(\vv),\vv} \right]\]
Thus, we can write $\bvalue_{\bT,i}(A,B)$ as the $\ell_1$ norm of $\Delta^i$, multiplied by the expected value of $\avgd_{\pi(\vv),\vv}$ taken over drawing a vertex $\vv$ in level $i$ with probability proportional to $| \Delta_{\vv}| = \big||A_\vv|-|B_\vv|\big|$. Note that the norm $\|\Delta^i\|_1$ can be easily estimated using the $\ell_1$ sketches of Indyk
\cite{indyk2006stable}. Thus, this simple manipulation motivates the following approach: \textbf{(1)} sample a vertex $\vv$ from level $i$  from the distribution $\calD_1(\Delta^i)$, \textbf{(2)} recover the value $\avgd_{\pi(\vv),\vv} $, \textbf{(3)} repeat enough times so that the empirical mean of the variables $\avgd_{\pi(\vv),\vv} $ is a good approximation of the expectation $\Ex_{ \vv\sim  \calD_1(\Delta^i)} \left[	\avgd_{\pi(\vv),\vv} \right]$. 

For the last step, we note that it will be straightforward to bound the standard deviation of the variable $	\avgd_{\pi(\vv),\vv} $ by $O(d \log n/2^i)$, which is within a $O(\log n)$ factor of the error to which we will need to estimate the expectation. Thus, if we can carry out steps \textbf{(1)} and \textbf{(2)} which sample $	\avgd_{\pi(\vv),\vv}$ from the correct distribution, we need only repeat them $\polylog(n)$ times to estimate $\bvalue_{\bT,i}(A,B) $ to sufficiently small error. 

\ignore{
Consider $\Delta_i>0$ for all $i\in [h]$ in this overview and let $\calV_i$ denote
  the distribution supported on nodes of $\bT$ at depth $i$ where each node $v$ is sampled with probability 
  $|\Delta_v|/\Delta_i$.  
Given a node $v$, we write $\pi(v)$ to denote its parent node.
Let $\alpha_i=2^i/(d\log^2 n)$. Let $\calS_i$ denote the distribution over subsets 
  of $[d]$ that includes each coordinate independently with probability $\alpha_i$.
Given $S \subset [d]$, $\chi_S \colon \{0, 1\}^d \to \{-1,1\}$ is $\chi_S(x) = (-1)^{\sum_{k \in S} x_k}$.
  
After some further transformations on $\bvalue_\bT(A,B)$, we get that with high probability on $\bT$:
\begin{equation}\label{introeq}
\EMD(A,B) \le \sum_{i\in [h]} \cbI_i\le \tilde{O}(\log n)\cdot  \EMD(A,B)
\end{equation}
where $\cbI_i$ for each $i\in [h]$ is defined as 
$$
\cbI_i\eqdef\frac{2\Delta_i}{\alpha_i}\cdot \Ex_{\vv\sim \calV_i,\bS\sim\calS_i}
\big[p_{\pi(\vv),\vv,\bS}\big] 
$$
and
  we define for each edge $(u,v)$ in $T$ and $S\subseteq [d]$ the following quantity $p_{u,v,S}$:
$$
p_{u,v,S}\eqdef  \Prx_{\cc_u\sim C_u,\cc_v\sim C_v}\big[\chi_S(\cc_u)\ne \chi_S(\cc_v)\big]\in [0,1].

$$  
}


{\bf Two-Pass Streaming Algorithms.} We first describe how the above two steps can be carried out in \textit{two-passes} over the data-stream. Perhaps unsurprisingly, our approach will be to carry out \textbf{(1)} on the first pass, obtaining a set of vertices $v$ sampled from the correct distribution $\calD_1(\Delta^i)$, and carry out \textbf{(2)} on the second pass, where we recover the actual value of $\avgd_{\pi(v),v}$ for the vertices $v$ that were sampled. 

More formally, our two-pass streaming algorithm proceeds as follows. First we draw a Quadtree $\bT$ (for which we may assume
  (\ref{introeq}) holds) and then for each $i\in [h]$,
  we estimate $\bvalue_{\bT,i}(A,B)$ as follows.
In the first pass we can estimate $\|\Delta^i\|_1$ to error $(1 \pm 1/2)$ with an $\ell_1$-sketch \cite{I06}, and we also can sample  $\vv \sim \calD_1(\Delta^i)$ via known algorithms for $\ell_1$-sampling~\cite{andoni2010streaming, Jowhari:2011, JW18}. Furthermore, once a vertex $v$ is
  fixed, we may estimate $\avgd_{\pi(\vv),\vv} $ in the second round by a point in $C_{\pi(v)}$ and in $C_{v}$ (via standard sub-sampling techniques) and approximating their distance using $\ell_1$ sketches. 
  By concurrently repeating this process $\polylog (n)$ times, we obtain our desired approximation

The remaining challenge, however, is to produce $ \bv \sim \calD_1(\Delta^i)$ \textit{and} an estimate of $\avgd_{\pi(\vv),\vv} $ \emph{simultaneously} in a single pass over the data. This task is a special case of a problem we call \emph{sampling with meta-data}, since the 
quantity $\avgd_{\pi(\vv),\vv} $
will be the \emph{meta-data} of the sample $ \bv \sim \calD_1(\Delta^i)$ needed to estimate  $\bvalue_{\bT,i}(A,B)$.

{\bf Sampling with Meta-Data and One-Pass Streaming}
\ignore{Our main technical contribution for sketching is to compress the above two-round protocol into a single linear sketch.
Specifically, our goal will be to construct the \textit{same} estimator $\smash{\hat{\cbI}_i}$ as used in the two-round sketch, where we first sample an edge $\be$ from $E_i$ with probability proportional to $|A_e^i - B_e^i|$, and then obtain an approximation $\hat{ w}_{\be}$ for $w_{\be}$. Importantly, we want to do both tasks simultaneously. Our approach involves several novel technical ideas, which may be generally applicable to sketching weighted $\ell_1$ distances.}
%
%
The key task 
of sampling with meta-data is the following: for $n, k \in \N$, 
we are given a vector $x \in \R^n$ and collection of \textit{meta-data} vectors $\lambda_1,\lambda_2,\dots,\lambda_n \in \R^{k}$, and the goal is to sample $i \in [n]$ with probability $|x_i|/\|x\|_1$ (or more generally, $|x_i|^p /\|x\|_p^p$), and output both $i$ and an 
approximation $\hat{\lambda}_i \in \R^k$ of the vector $\lambda_i$. The challenge is to solve this problem with a small-space linear sketches of $x$ and the meta-data vectors $\lambda_1, \dots, \lambda_n$. It is not hard to see that sampling with meta-data is exactly the problem we seek to solve for linear sketching of EMD.\footnote{Namely, $x$ is the vector $\Delta^i$, and the meta-data vectors $\lambda_{v}$ are $k = \polylog(n)$-dimensional $\ell_1$ \textit{sketches} of the values of $\avgd_{\pi(v),v}$. In the following discussion, for simplicity we omit the details on the $\ell_1$ sketches for $\avgd_{\pi(v),v}$, since they proceed via somewhat standard techniques, and instead assume that the meta-data is exactly given by the scalars $\lambda_\vv \approx \avgd_{\pi(v),v}$.} \ignore{Using relatively straightforward extensions of arguments in previous $\ell_p$ sampling sketches \cite{andoni2010streaming,Jowhari:2011,JW18} , one can obtain generic bounds for this problem in terms of the sums of norms of the meta-data vectors $\lambda_i$. Unfortunately, these bounds will be too weak for our purposes, and our algorithm will require significant additional techniques. }

Our algorithm builds on a powerful sketching technique known as \textit{precision sampling} \cite{andoni2010streaming,Jowhari:2011,JW18} for sampling an index $i \in [n]$ proportional to $|x_i| / \|x\|_1$ for a vector $x \in\R^n$ (or more generally, for $|x_i|^p / \|x\|_p^p$, but we focus on $p =1$). \ignore{Given a vector $x \in \R^n$, precision sampling addresses the problem of sampling a coordinate $i \in [n]$ with probability approximately $|x_i|/\|x\|_1$ (or more generally $|x_i|^p/\|x\|_p^p$, but we focus on the case of $p=1$). In precision sampling, } The idea is to produce, for each $i \in [n]$ an independent exponential random variable $\bt_i \sim \mathrm{Exp}(1)$, and 
construct a ``scaled vector'' $\bz \in \R^n$ with coordinates $\bz_i = x_i / \bt_i$. One then attempts to return the index $i_{\max} = \argmax_{i \in [n]} \bz_i$, since
\begin{align*}
\Prx_{\bt_1, \dots, \bt_n \sim \mathrm{Exp}(1)}\left[ \argmax_{i' \in [n]} \frac{|x_{i'}|}{\bt_{i'}} = i \right] &= \dfrac{|x_i|}{\|x\|_1}. 
\end{align*} 
To find the the index $i_{\max}$ with a linear sketch, we can use a ``heavy-hitters'' algorithm, such as the Count-Sketch of \cite{charikar2002finding}.\footnote{We do not explicitly use count-sketch in our one-pass algorithms, and instead apply a sketching procedure closely inspired by Count-Sketch.} 
Specifically, Count-Sketch with error $\eps\in (0,1)$ allows us to recover an estimate $\tilde{\bz}$ to $\bz$ satisfying (roughly) $\|\tilde{\bz} - \bz\|_{\infty} \leq \eps \|\bz\|_2$. Then one can show that $\argmax_{i' \in [n]} |\tilde{\bz}_{i'}|$ is close to being distributed as $|x_i| /\|x\|_1$.

In order to sample with meta-data, our sketch similarly samples independent exponential $\bt_1, \dots, \bt_n \sim \mathrm{Exp}(1)$ and applies a Count-Sketch data structure on $\bz \in \R^n$, where $\bz_i = x_i / \bt_i$, and obtains an estimate $\tilde{\bz}$ of $\bz$. In addition, we apply a Count-Sketch data structure with error $\eps$ for the vector $\bw$ with coordinates given by the values $\lambda_i / \bt_i$, namely $\bw_i = \lambda_i/ \bt_i$ (recall that we are assuming that the meta-data $\lambda_i$ are scalars for this discussion). From this we obtain an estimate $\tilde{\bw}$ of $\bw$.
The insight is the following: suppose the sample produced is $i^* \in [n]$, which means it satisfies $\tilde{\bz}_{i^*} \approx \max_{i \in [n]} |x_{i}| / \bt_{i}$. Then the value $\bt_{i^*}$ should be relatively small: in particular, one can show that we expect $\bt_{i^*}$ to be $\Theta(|x_{i^*}| / \|x\|_1)$, so that $\bz_{i^*} \approx \Theta(\|x\|_1) = \Theta(\|\Delta^i\|_1)$. When this occurs, for each $\ell \in [k]$, the guarantees of Count-Sketch 
imply that the estimate $\bt_{i^*} \cdot \tilde{\bw}_{i^*}^{\ell}$ satisfies
\begin{align*}
\left| \bt_{i^*} \cdot \tilde{\bw}_{i^*} - \lambda_{i^*}\right| = \bt_{i^*}\left| \tilde{\bw}_{i^*} - \bw_{i^*} \right| \leq \eps \bt_{i^*} \| \bw\|_2 \left( = O\left(\eps |x_{i^*}| \cdot \frac{\|\lambda\|_1}{\|x\|_1}\right) \text{ in expectation} \right) 
\end{align*}
where $\lambda \in \R^n$ is the vector with coordinates given by the meta-data $\lambda_1, \dots, \lambda_n$. In other words, if the size of $\lambda_{i^*}$ is comparable to $|x_{i^*}|$, and if the ratio $\|\lambda\|_1/\|x\|_1$ of the meta-data norms to the norm of $x$ is bounded, then $\bt_{i^*} \tilde{\bw}^{\ell}_{i^*}$ is a relatively good approximation to $\lambda_{i^*}$. 

Unfortunately, in our application, the above will not always be the case. In particular, the norm of the meta-data  $\|\lambda\|_1 $ may be much, even poly$(n)$, larger than $\|x\|_1 = \|\Delta^i\|_1$. Intuitively, the issue is that each coordinate $\lambda_v$ is a sketch of $ \avgd_{\pi(v),v}$, which is a function both of the points in $C_{\pi(v)}$ and $C_v$. Thus, the size of the sketch of $ \avgd_{\pi(v),v}$ depends on all the points in $C_{\pi(v)}$. Moreover, for \textit{every} other sibling $v'$ of $v$ (meaning that $\pi(v') = \pi(v))$, the sketch of $ \avgd_{\pi(v'),v'}$ will also have to take into account the same information from $C_{\pi(v)}$. Thus, this information is \textit{duplicated} in the sketches of the meta-data, by a number of times equal to the number of children of $\pi(v)$. This duplication, or repetition of the same information in the sketch, results in a blow-up of the norm of $\lambda$ so that $\|\lambda\|_1 = \Omega(\kappa \cdot \|\Delta^i\|_1)$, where $\kappa$ is the maximum number of non-empty children of any parent in level $i-1$. Since $\kappa$ can be $\poly(n)$, this is an non-trivial challenge. 

Our solution to this, at a high level, is to develop a \textit{two-step }precision sampling with meta-data algorithm to avoid duplication of meta-data. Instead of sampling the vertex $\bv^*$ directly, we first sample a parent $\bu^*$ from level $i-1$ with probability proportional to the $\ell_1$-norm of $\Delta^i$ restricted to coordinates corresponding to the children of $\bu^*$; namely, we sample $\bu^*$ with probability proportional to $\sum_{v : \pi(v) = \bu^*} |\Delta^i_{v}|$. Then, we use the precision sampling sketch which recovered $\bu^*$ to recover a sketch of the a randomly selected point in $C_{\bu^*}$. 
Next, once we have $\bu^*$, we apply precision sampling with meta-data once more, to sample a child $\bv^*$ of $\bu^*$ proportional to $|\Delta^i_{\bv^*}|$, and then recover a sketch of  a randomly selected point in $C_{\bv^*}$. One can then put the two sketches from $C_{\bu^*},C_{\bv^*}$ together to estimate $\avgd_{\bu^*,\bv^*}$. 

To accomplish this two-part precision sampling scheme, we must generate a second set of exponentials $\{\bt_v\}_{v}$, one for each child node $v$ at depth $i$. In order to ensure that the sample produced by the second sketch actually returns a child $\bv^*$ of $\bu^*$, and not a child of some other node, we crucially must scale the vector $\Delta^i$ by \textit{both} the child exponentials $\{\bt_v\}_{v}$ as well as the parent exponentials $\{\bt_u\}_{u}$ from the first sketch. Thus, in the second sketch we analyze the twice-scale vector $\bz$ with coordinates $\bz_v = \Delta^i_{v}/ (\bt_{\pi(v)} \bt_v)$, and attempt to find the largest coordinate of $\bz$.  Importantly, notice that this makes the scaling factors in $\bz_v$ no longer independent: two children of the same parent share one of their scaling factors. Thus, executing this plan requires a careful analysis of the behavior of norms of vectors scaled by several non-independent variables with heavy-tailed distributions. 

The advantage of this two-part scheme is that now there is no duplication of meta-data, since in the first step there is only one $\lambda_u$ for each parent $u$, and in the second step, by conditioning on the parent exponential $\bt_{\bu^*}$ being sufficiently small, we ensure that the only meta-data that contributes non-trivially to the error of the sketch are the $\lambda_v$ for children $v$ of $\bu^*$. This allows us, ultimately, to obtain our guarantees for one-pass streaming algorithms for $\EMD$. The case of $\MST$ is similar at a high-level, however implementing the two-part precision sampling scheme requires an entirely different set of sketching tools, resulting from the fact that we now need to sample a vertex $v$ from the $\ell_0$ distribution $\calD_0(\Delta^i)$.

%% file: inspector.tex

\def\bC{\mathbf{C}} \def\bv{\mathbf{v}}
\newcommand{\Compress}{\textsc{Compress}}
\def\sv{\mathsf{v}} \def\su{\mathsf{u}}
\def\avgd{\text{avg}} \def\cc{\mathbf{c}}
\def\bvalue{\textbf{Value}}

\section{Preliminaries}

Given $n\ge 1$ we write $[n]$ to denote $\{1,\ldots,n\}$. Given a vector $x \in \R^n$ and a real number $t \geq 0$, we define 
$x_{-t} \in \R^n$ to be the vector obtained by setting the largest $\lfloor t \rfloor$ coordinates of $x$ in magnitude equal to $0$ (breaking ties by using coordinates with smaller indices). For $a,b \in \R$ and $\eps \in (0,1)$, we use the notation $a = (1 \pm \eps) b$ to denote the containment of $a \in [(1-\eps)b , (1+\eps)b]$.

For convenience, we will assume without loss of generality that $d$ is always a power
  of $2$ and write $h:=\log_2 2d=\log_2 d+1$. Given a node $v$ in a rooted tree $T$, when $v$ is not the root we use $\pi(v)$ to denote the parent node of $v$ in $T$. 

Next we give a formal definition of Quadtrees used in this paper:

\begin{definition}[Quadtrees]
Fix $d\in \N$.
A quadtree is a rooted tree $T$ of depth $h:=\log_2 2d$. We say a node $v$ of $T$ is at depth
  $j$ if there are $j+1$ nodes on the root-to-$v$ path in $T$
  (so the root is~at depth $0$ and its leaves are at depth $h$). 
Each internal node $v$ of $T$ at depth $j<h$ is labelled with an ordered tuple 
  of $2^j$ coordinates $i_1,\ldots,i_{2^j}\in [d]$ (which are not necessarily distinct),
  and has $\smash{2^{2^j}}$ children, each of which we refer to as the
  $(b_1,\ldots,b_{2^j})$-child of $v$ with $b_1,\ldots,b_{2^j}\in \{0,1\}$.
Every node at depth $h-1$ is labelled with $(1,\ldots,d)$.
We write $E_T$ to denote the edge set of $T$.
Whenever we refer to an edge $(u,v)\in E_T$, $u$ is always the parent and $v$ is the child.
A random quadtree $\bT $ is drawn by (1) sampling a tuple of $2^j$ coordinates
  uniformly and 
  independently from $[d]$ for each node at depth $j<h-1$ as its label; and (2)
  use $(1,2,\ldots,d)$ as the label of every node at depth $h-1$. %
We use $\calT$ to denote this distribution of random quadtrees.
\end{definition} 

Given a quadtree $T$, each point $x\in \{0,1\}^d$ 
  induces a root-to-leaf path by starting at the root
  and repeatedly going down the tree as follows:
If the current node $v$ is at depth $j<h$ and is labelled with $(i_1,\ldots,i_{2^j})$,
  then we go down to the $\smash{(x_{i_1},\ldots,x_{i_{2^j}})}$-child of $v$.
We write
$$
\sv_{0,T}(x),\sv_{1,T}(x),\ldots,\sv_{h,T}(x)
$$
to denote this root-to-leaf path, where each $\sv_{j,T}$ is a map from
  $\{0,1\}^d$ to nodes of $T$ at depth $j$.
We usually drop $T$ from the subscript when it is clear from the context.

Alternatively we define a subcube $S_{v,T}\subseteq \{0,1\}^d$ for each $v$:
  The set of the root is $\{0,1\}^d$; If $(u,v)$ is an edge, $u$ is at depth $j$
  and is labelled with $i_1,\ldots,i_{2^j}$, and $v$ is the $(b_1,\ldots,b_{2^j})$-child
  of $u$, then 
$$
S_{v,T}=\big\{x\in S_{u,T}: (x_{i_1},\ldots,x_{i_{2^j}})=(b_1,\ldots,b_{2^j})\big\}.
$$
Note that $S_{v,T}$'s of nodes $v$ at the same depth form a partition of $\{0,1\}^n$.
The root-to-leaf path for $x\in \{0,1\}^d$ can be
  equivalently defined as the sequence of nodes $v$  
  that have $x\in S_{v,T}$.

  \begin{remark}\label{rem:aik-bdirw}
  	Both works of \emph{\cite{AIK08, BDIRW20}} use a tree structure that is very similar to the Quadtree used in this paper. In particular,
  	They consider 
  	a slightly different algorithm which~at depth $i$, samples $2^i$ coordinates from $[d]$
  	and divides into $2^{2^i}$ branches according to settings of $\{0,1\}$ to these $2^i$ coordinates (instead of each vertex independently sampling $2^i$ coordinates). For the sake of the analysis in Section \ref{sec:quadtree-cost}, there will be no difference between independently sampling coordinates for each vertex in a level, and using the same sampled coordinates for each level. Thus, our analysis apply to trees of \emph{\cite{AIK08, BDIRW20}} as well as the Quadtrees defined here.
  	
  \end{remark}

\section{Analysis of Quadtrees for $\EMD$ and $\MST$}\label{sec:quadtree-cost}

Our goal in this section is to obtain expressions based on quadtrees that
  are good approximations of $\EMD$ and $\MST$.
They will serve as the starting point of our sketches for $\EMD$ and $\MST$ later.

\subsection{Approximation of $\EMD$ using Quadtrees}\label{sec:quadtreeemd}

\def\depth{\text{depth}}

Fix $n,d\in \N$ and let $T$ be a quadtree of depth $h=\log_2 2d$.
Let $A$ and $B$ be two multisets of points from $\{0,1\}^d$ of size $n$ each.
For each node $v$ in $T$, we define
\[ A_{v, T} \eqdef \{ a \in A : \sv_{i, T}(a) = v \} \quad\text{and}\quad B_{v, T} \eqdef \{ b \in B : \sv_{i, T}(b) = v \}.\]
Equivalently we have $A_{v,T}=A\cap S_{v,T}$ and $B_{v,T}=B\cap S_{v,T}$.
Let $C_{v,T}=A_{v,T}\cup B_{v,T}$.
We give the definition of depth-greedy matchings. 
\begin{definition}\label{def222}
Let $T$ be a quadtree.
For any $a \in A$ and $b \in B$, let 
\[ \depth_{T}(a, b) \eqdef \text{depth of the least-common ancestor of leaves of $a, b$ in $T$}.\]
The class of \emph{depth-greedy matchings}, denoted by $\calM_T(A, B)$, is the set of all matchings $M \subseteq A \times B$ which maximize the sum of $\depth_T(a, b)$ over all pairs $(a, b) \in M$.
We write $$\cost(M)=\sum_{(a,b)\in M}
  \|a-b\|_1$$ to denote the cost of a matching $M$ between $A$ and $B$. Recall that $ \EMD(A,B)
  $ is defined as the minimum of $\cost(M)$ over all matchings between $A$ and $B$. 
\end{definition}

For each edge $(u,v)\in E_T$, we use $\avgd_{u,v,T}$ to denote
  the average distance between points of $C_{u,T}$ and $C_{v,T}$:
$$
\avgd_{u,v,T}\eqdef\Ex_{\substack{\cc\sim C_{u,T}\\ \cc'\sim C_{v,T}}}\big[\|\cc-\cc'\|_1\big],
$$
where both $\cc$ and $\cc'$ are drawn uniformly at random; 
  we set $\avgd_{u,v,T}$ to be $0$ by default when $C_{v,T}$ is empty.
For notational simplicity, we will suppress $T$ from the subscript 
  when it is clear from the context.
We are now ready to define the \emph{value} of $(A,B)$ 
  in a quadtree $T$:

\begin{definition}\label{def:cost}
Let $T$ be a quadtree. 
The value of $(A,B)$ in $T$ is defined as
\begin{equation}\label{eq:inspect}
{\emph{\bvalue}}_T(A,B)\eqdef \sum_{(u,v) \in E_T} \big| |A_v| - |B_v| \big| \cdot 
\emph{\avgd}_{u,v}.
\end{equation}
\end{definition}
We note that the right-hand side of (\ref{eq:inspect}) is data-dependent in two respects: 
  the discrepancy between $|A_v|$ and $|B_v|$ and 
  the average distance $\avgd_{u,v}$ between points in $C_u$ and $C_v$.

Our main lemma for $\EMD$ shows that the value of $(A,B)$ in a randomly chosen quadtree $\bT \sim \calT$ and the cost of any depth-greedy matching are all $\tilde{O}(\log n)$-approximations to $\EMD(A,B)$.

\begin{lemma}[Quadtree lemma for $\EMD$]\label{quadtreelemmaforemd}
Let $(A,B)$ be a pair of multisets of points from $\{0,1\}^d$~of size $n$ each.
Let $\bT\sim \calT$. Then with probability at least $0.99$, every $M \in \calM_{\bT}(A, B)$ satisfies we have
\begin{equation}\label{maininequalityemd}
\EMD(A,B)\le \cost(M) \leq \emph{\bvalue}_{\bT}(A,B)\le \tilde{O}(\log n)\cdot \EMD(A,B).
\end{equation}
\end{lemma}

We start with the first inequality in (\ref{maininequalityemd}).  
Indeed we will show    
  that $ \EMD(A,B)\le \bvalue_T(A,B)$  for any quadtree $T$ (Lemma \ref{lem:upper-bound}).
To this end we prove that $\cost(M)\le \bvalue_T(A,B)$ for any
  depth-greedy matching between $A$ and $B$ obtained from $T$;
  the latter by definition is at least $\EMD(A,B)$.

\begin{lemma}\label{lem:upper-bound}
Let $T$ be any quadtree. Then 
  $\cost(M)\le {\emph{\bvalue}}_T(A,B)$ for any $M\in \calM_T(A,B)$.\end{lemma}
\begin{proof}
Given an $M\in \calM_T(A,B)$ and a pair $(a,b)\in M$, we write
  $v$ and $w$ to denote the leaves of $a$ and $b$ and use 
  $v=u_1,u_2,\ldots,u_k=w$ to denote the path from $v$ to $w$ in $T$.  
By triangle inequality,
\begin{align*}
\|a-b\|_1&\le \Ex_{\cc_i\sim C_{u_i}}
\Big[\|a-\cc_1\|_1+\|\cc_1-\cc_2\|_1+\cdots+ \|\cc_{k-1}-\cc_{k}\|_1+\|\cc_k-b\|_1\Big]\\[-0.5ex]
&=\avgd_{u_1,u_2}+\cdots+\avgd_{u_{k-1},u_k},
\end{align*}
where the equation follows from the fact 
  the label of every node at depth $h-1$ is $(1,2,\ldots,d)$ and thus,
  all points at a leaf 
  must be identical.
Summing up these inequalities over all $(a,b)\in M$ gives exactly $\bvalue_T(A,B)$ on the right hand side.
For this, observe that every $M$ in $\calM_T(A,B)$ has the property that, for any edge $(u,v)$ in $T$, the number of 
  $(a,b)\in M$ such that the path between their leaves contains $(u,v)$ is exactly $||A_v|-|B_v||$.
\end{proof}


Now it suffices to upperbound $\bvalue_\bT(A,B)$ by
  $\tilde{O}(\log n)\cdot \EMD(A,B)$ with probability at least $0.9$
  for~a random quadtree $\bT\sim\calT$.
For this purpose  we let $C=A\cup B$ and define an \emph{inspector payment} for 
  any pair of points $a,b\in C$\footnote{While we will always have 
  $a\in A$ and $b\in B$ in this subsection, this more general setting allows 
  us to apply what we prove in this subsection to work on MST later.}
  based on a quadtree.
Given 
  $a,b\in C$, we let
\begin{align} \label{eq:inspect-a}
\pay_{T}(a,b)  \eqdef \sum_{i\in [h]} \ind\big\{\sv_{i }(a) \neq \sv_{i }(b) \big\} 
\cdot \Big(\avgd_{a,i-1}+\avgd_{b,i-1}\Big)
\end{align}
where
$$
\avgd_{a,i-1}\eqdef\Ex_{\cc\sim C_{\sv_{i-1}(a)}} \big[\|a-\cc\|_1\big]\quad\text{and}\quad
\avgd_{b,i-1}\eqdef\Ex_{\cc\sim C_{\sv_{i-1}(b)}} \big[\|b-\cc\|_1\big].
$$
Intuitively $\pay_T(a,b)$ pays for the average distance between $a$ (or $b$) and points in
  $C_{\sv_i(a)}$ (or $C_{\sv_i(b)}$)
  along its root-to-leaf path but the payment only starts \emph{at} the least-common
  ancestor of leaves of $a$ and $b$.
Note that 
  $\pay_T(a,b)=0$ trivially if $a=b$.

We show that for any matching $M$ between $A$ and $B$,
  the total inspector payment from $(a,b)\in M$ 
  is enough to cover $\bvalue_T(A,B)$: 
\begin{lemma}\label{lem:split-criteria}
Let $T$ be any quadtree and $M$ be any matching between $A$ and $B$. Then we have
\begin{align}
\emph{\bvalue}_T(A,B)\leq 2\sum_{(a,b)\in M} \pay_T(a,b).
\label{eq:payment-bound}
\end{align}
\end{lemma}
\begin{proof}
Using the definition of $\bvalue_T(A,B)$, it suffices to show that
$$
\sum_{(u,v) \in E_T} \big| |A_v| - |B_v| \big| \cdot \avgd_{u,v} 
\le 2\sum_{(a,b)\in M} \pay_T(a,b).
$$
By triangle inequality (and $\avgd_{a,h}=0$ because every point in $C_{\sv_h(a)}$ is identical to $a$)
\begin{align*}
2\cdot \pay_{T}(a,b)
&\ge \sum_{i\in [h]} \ind\big\{ \sv_{i }(a) \neq \sv_{i }(b) \big\}
 \cdot \Big(\avgd_{a,i-1}+\avgd_{a,i }+\avgd_{b,i-1}+\avgd_{b,i}\Big) \\[0.5ex]
 &\ge \sum_{i\in [ h ]} \ind\big\{ \sv_{i }(a) \neq \sv_{i }(b) \big\}
  \cdot \Big(\avgd_{\sv_{i-1}(a),\sv_{i }(a)}+\avgd_{\sv_{i-1}(b),\sv_{i }(b)}\Big),
\end{align*}
i.e., $2\cdot \pay_T(a,b)$ is enough to cover $\avgd_{u,v}$ for every edge $(u,v)$
  along the path between the leaf of $u$ and the leaf of $v$.
The lemma then follows from the following claim: For every edge $(u,v)$ in $T$,
  $||A_v|-|B_v||$ is at most the number of points $a\in A_v$ such that its
  matched point in $M $ is not in $B_v$ plus the number of points $b\in B_v$ such that
  its matched point in $M $ is not~in~$A_v$.
This follows from the simple fact that every $(a,b)\in M $ with $a\in A_v$ and 
  $b\in B_v$ would~get cancelled in $|A_v|-|B_v|$.
This finishes the proof of the lemma.
\end{proof}



By Lemma~\ref{lem:split-criteria} the goal now is to upperbound the total inspector payment by $\tilde{O}(\log n)\cdot \EMD(A,B)$ with probability at least $0.9$ over a randomly picked quadtree $\bT$. We consider a slight modification of 
  the payment scheme given in (\ref{eq:inspect-a}) which we define next; the purpose  is that the latter will~be easier to bound in expectation, and most often exactly equal to (\ref{eq:inspect-a}). 
  
  Specifically, given any $(a,b)$ with $a,b\in C$ and $i_0 \in [0:h-1]$, we let
\begin{align}
{\pay}^*_{i_0, T}(a,b) &\eqdef \sum_{i> i_0}^{h} \ind\big\{ \sv_{i }(a) \neq \sv_{i}(b) \big\} 
\cdot \Big(\avgd^*_{a,i-1}+\avgd^*_{b,i-1}\Big),
\label{eq:def-wtpay}
\end{align}
where
\begin{align*}
\avgd^*_{a,i}\eqdef 
\Ex_{\cc\sim C^*_{a,i}}\big[\|a-\cc\|_1\big]\quad\text{and}\quad
\avgd^*_{b,i}\eqdef 
\Ex_{\cc\sim C^*_{b ,i}}\big[\|b-\cc\|_1\big]
\end{align*}
and $C^*_{a,i}$ contains all points in $C_{\sv_i(a)}$ that is not too far away from $a$:
$$
C^*_{a,i}\eqdef \left\{ c \in C_{\sv_i(a)} : \|a - c\|_1 \leq \frac{10 d \log n}{2^i} \right\}.
$$
The set $\smash{C^*_{b,i}}$ is defined similarly.
Roughly speaking, points in $C$ that share the same node~at depth $i$ 
  are expected to have distance around $d/2^i$ (given they have agreed 
  on  $2^i-1$ random coordinates sampled so far); this is why we refer  to points in $\smash{C^*_{a,i}}$ as those that are not too far away from $a$.

The following is the crucial lemma for upperbounding the total expected
  payment according to an optimal matching $M^*$.
We delay its proof to Section \ref{maintechlemma} and first use it to prove
Lemma \ref{quadtreelemmaforemd}.
\begin{lemma}\label{lem:main-lemma}
For any $(a,b)$ with $a,b\in C$, $a\ne b$ and $i_0\in [0:h-1]$ that satisfies
\begin{equation}\label{blabla1}
i_0\le h_{a,b}\eqdef \left\lfloor \log_2\left( \frac{d}{\|a-b\|_1}\right) \right\rfloor ,
\end{equation}
  we have 
\begin{align*}
& \Ex_{\bT\sim\calT}\Big[ {\pay}^*_{i_0,\bT}(a,b)\Big] \leq \left( \tilde{O}(\log n) + O(\log \log n) \left(h_{a,b} - i_0 \right) \right) \cdot \| a - b\|_1.\end{align*}
\end{lemma}
\begin{proofof}{Lemma \ref{quadtreelemmaforemd} assuming Lemma~\ref{lem:main-lemma}}
Let $M^*$ be an optimal matching between $A$ and $B$ that achieves $\EMD(A,B)$.
Let $\bT\sim \calT$. Then we have from Lemma \ref{lem:split-criteria} that
\begin{equation}\label{hehe33}
\bvalue_\bT(A,B)\le  2\sum_{\substack{(a,b)\in M^*\\ a\ne b}} \pay_{\bT}(a,b)
\end{equation}
given that $\pay_{\bT}(a,b)=0$  when $a=b$.
Below we focus on the subset $M'$ of $M^*$ with $(a,b)\in M^*$ and $a\ne b$. 
For each  $(a,b)\in M'$, let 
\[ 0\le \ell_{a, b} \eqdef \max \big\{0, h_{a,b} - 2\lceil \log_2 n \rceil \big\}\le h_{a,b}. \]
We show that with probability at least $1-o(1)$ over the draw of $\bT$,
  every $(a, b) \in M'$ satisfies
\begin{align}
\pay_\bT(a,b) =  {\pay}^*_{\ell_{a,b},\bT}(a,b).
\label{eq:often-same}
\end{align}
Combining (\ref{hehe33}) and (\ref{eq:often-same}), we have that with probability at least $1 - o(1)$ over the draw of $\bT$,
\begin{align}\label{hehe44}
\bvalue_\bT(A,B)  
\le 2\sum_{(a,b)\in M'}  {\pay}^*_{\ell_{a,b},\bT}(a,b).
\end{align}
By applying Lemma~\ref{lem:main-lemma} to every $(a,b) \in M'$ with $i_0 = \ell_{a,b}$, as well as Markov's inequality, we have that with probability at least $0.99$ over $\bT$,
  the right hand side of (\ref{hehe44}) is at most  
\begin{align*}
 \tilde{O}(\log n) \sum_{\substack{(a,b)\in M'}} \|a - b\|_1 = \tilde{O}(\log n) \cdot\EMD(A,B).
\end{align*}
By a union bound, 
$\bvalue_\bT(A,B) \leq \tilde{O}(\log n) \cdot \EMD(A,B)$ with probability at least $.99 - o(1)\ge 0.9$. 

It suffices to define an event that implies (\ref{eq:often-same}) and then bound its probability.
The first part of the event requires that for every pair $(a,b)\in M'$,
  $\sv_{i}(a)=\sv_{i}(b)$ for every $i:1\le i\le \ell_{a,b}$.
The second part requires that for any two distinct points $x,y\in A\cup B$
  (not necessarily as a pair in $M^*$ and not even necessarily in the same set),
  we have $\sv_{i}(x)\ne \sv_{i}(y)$ for all $i$ with 
\begin{equation}\label{hehe66}
2^i\ge \frac{10d \log n}{\|x-y\|_1}.
\end{equation}
By the definition of $\smash{ {\pay}^*_{\ell_{a,b},\bT}(a,b)}$ in (\ref{eq:def-wtpay}), the first part of the event makes sure that 
  we don't miss any term in the sum;
  the second part of the event makes sure that every $C^*_{a,i}$ is 
  exactly the same as $C_{\sv_i(a)}$ so that
  $\avg_{a,i}^*=\avg_{a,i}$ (and the same holds for $b$) .
It follows that this event implies (\ref{eq:often-same}).

Finally we show that the event occurs with probability at least $1-o(1)$. 
First, for every $(a,b)\in M'$, if $\ell_{a,b}=0$ then the first part of the event trivially holds.
If $\ell_{a,b}>0$ then $\ell_{a,b}=h_{a,b}-2\lceil \log n\rceil$.
The probability of $\sv_{i }(a)\ne \sv_{i }(b)$
  for some   $i:1\le i \le \ell_{a, b}$ is at most
\begin{align*}
1 - \left( 1 - \frac{\|a-b\|_1}{d}\right)^{2^{\ell_{a,b}}-1} \leq 2^{\ell_{a,b}} \cdot \frac{\|a - b\|_1}{d} \leq \frac{1}{n^2}.
\end{align*}
Hence, by a union bound over the at most $n$ pairs $(a,b)\in M'$, the first part of the
  event holds with probability at least $1-o(1)$.
Furthermore, for any two distinct points $x, y \in A \cup B$,
  let   
\[ \ell^* =  \left\lfloor \log_2\left(\frac{10d\log n}{\|x-y\|_1}\right) \right\rfloor .\]
Then $\sv_{i }(x)= \sv_{i }(y)$ for some $i$ that satisfies (\ref{hehe66}) would
  imply  
  $\sv_{\ell^* }(x)=\sv_{\ell^*}(y)$ and $\ell^* \le \log d$ (since $\sv_{h}(x)\ne \sv_{h}(y)$   given $x\ne y$).
The event above happens with probability 
\begin{align*}
\left(1 - \frac{\|x-y\|_1}{d} \right)^{2^{\ell^*-1}} \leq \exp(-5 \log n) = \frac{1}{n^5}.
\end{align*}
Via a union bound over at most $(2n)^2$ many pairs of $x,y$, 
  we have that the second part of the event also happens with probability at least
  $1-o(1)$.
This finishes the proof of the lemma.
\end{proofof}

\subsection{Approximation of $\MST$ using Quadtrees} 
We will follow a similar strategy as we took in the previous subsection for  EMD. 
Given a quadtree $T$ of depth $h=\log_2 2d$, 
  we define similarly $\sv_0(x),\ldots,\sv_h(x)$ as the root-to-leaf path 
  of $x\in \{0,1\}^d$, and write $S_v$ for each node $v$ at depth $i$
  to denote the set of $x\in \{0,1\}^d$ with $\sv_i(x)=v$.

Let $X\subseteq \{0,1\}^d$ be a set of $n$ points.
We define $X_v$ for each node $v$ in $T$ as $X\cap S_v$,
  and write $L_i$ for each depth $i$ to denote
  the set of nodes $v$ at depth $i$ such that $X_v\ne \emptyset$ and will refer to them
  as nonempty nodes. 
  
  We give the definition of depth-greedy spanning trees. 
\begin{definition}\label{def333}
Let $T$ be a quadtree, and $X \subset \{0,1\}^d$. For any DFS walk of the quadtree $T$ starting at the root, let $\sigma \colon [n] \to X$ denote the order of points in $X$ encountered during the walk, so that $\sv_{h}(\sigma(i))$ appears before $\sv_{h}(\sigma(i+1))$ for every $i \in [n-1]$. A depth-greedy spanning tree $G$ obtained from a DFS walk is given by the edges $\{ (\sigma(i), \sigma(i+1)) \}_{i \in [n-1]}$.
The class of \emph{depth-greedy spanning trees}, denoted by $\calG_T(X)$, is the set of all spanning trees $G$ of $X$ obtained from a DFS walks down the quadtree $T$. 
For any spanning tree $G$, we write
$$
\cost(G)=\sum_{(a,b)\in E(G)} \|a-b\|_1
$$ 
to denote the cost of a tree $G$ (with $n-1$ edges) spanning points in $X$.  
Recall $\MST(X)$ is defined as the minimum of $\cost(G)$ over all spanning trees $G$ of $X$.
\end{definition}

Similar to the previous subsection, for each edge $(u,v)\in E_T$, we write
$$
\avgd_{u,v}\eqdef \Ex_{\substack{\cc\sim X_u\\ \cc'\sim X_v}}\big[ \|\cc-\cc'\|_1\big].
$$
when $X_v\ne\emptyset$, and $\avgd_{u,v}=0$ when $X_v=\emptyset$.
Recall  $\pi(v)$ denotes the parent node of $v$ in $T$.
We are now ready to define the value of $X$ in a quadtree $T$ and then
  state the main lemma:

\begin{definition}\label{lem:MST_Estimator} Let $T$ be a quadtree. The value of $X$ in $T$ is defined as 
	\[\emph{\bvalue} _T(X)\eqdef \sum_{i\in [h]}\indi\big\{|L_i|>1\big\}\cdot 
	\sum_{v\in L_i}\emph{\avgd}_{\pi(v),v}.\]
\end{definition}

The main lemma for $\MST$ shows that the value of $X$ for a random quadtree $\bT \sim \calT$ and the cost of any depth-greedy spanning tree $G \in \calG_{T}(X)$  are $\tilde{O}(\log n)$-approximations of $\MST(X)$.

\begin{lemma}[Quadtree lemma for $\MST$] \label{thm:offlineMST}Let $X\subseteq \{0,1\}^d$ be a set of size $n$, and let $\bT\sim\calT$. Then with probability at least $0.99$, for any $G \in \calG_{\bT}(X)$, we have that 
	\begin{align*}
		\frac{\MST(X)}{2} \leq \frac{\cost(G)}{2} \leq \emph{\bvalue}_{\bT}(X)\leq \tilde{O}(\log n) \cdot \MST(X).
	\end{align*}
\end{lemma}

We start with the lower bound:
\begin{lemma}\label{lem:MST-lb} Let $T$ be any quadtree and any depth-greedy spanning tree $G \in \calG_{T}(X)$. Then $\emph{\bvalue}_T(X)\ge \cost(G) / 2 \geq\MST (X)/2$.
\end{lemma}
\begin{proof}
Let $w$ be the least common ancestor of leaves $\sv_h(x)$, $x\in X$, and let 
  $T^*$ denote the subtree rooted at $w$ that consists of paths from $w$ to $\sv_h(x)$, $x\in X$.
Using $T^*$ we can equivalently write 
$$
\bvalue_T(X)=\sum_{(u,v)\in E_{T^*}} \avgd_{u,v}.
$$ 
For each node $v\in T^*$ (note that $X_v\ne \emptyset$), we define $\rho_v$ to be the center-of-mass of points in $X_v$:
$$
\rho_v\eqdef \frac{1}{|X_v|}\sum_{x\in X_v} x.
$$
By triangle inequality we have $\|\rho_u-\rho_v\|_1\le \avgd_{u,v}$ for every $(u,v)\in E_{T^*}$ and thus,  
$$
\sum_{(u,v)\in E_{T^*}} \|\rho_u-\rho_v\|_1\le \bvalue_T(X).
$$

We finish the proof by showing that any depth-greedy spanning tree $G$ of $X$ satisfies
$$\cost(G)\le 2\sum_{(u,v)\in E_{T^*}}\|\rho_u-\rho_v\|_1.$$
To this end we take a DFS walk of $T^*$ from its root $w$ and let $\sigma \colon [n] \to X$ be the order of
  points in $X$ under which
  $\sv_h(\sigma(1)),\ldots,\sv_h(\sigma(n))$ appear in the walk.
Then we set $G$ to be the spanning tree $\{ (\sigma(i), \sigma(i+1))\}_{i \in [n-1]}$.
For each $i\in [n-1]$, letting $u_1,\ldots,u_r$ be the part of DFS walk from $u_1=\sv_h(\sigma(i))$ to
  $u_r=\sv_h(\sigma(i+1))$, we have from triangle inequality that
$$
\|\sigma(i)-\sigma(i+1)\|_1=\|\rho_{u_1}-\rho_{u_r}\|_1\le \|\rho_{u_1}-\rho_{u_2}\|_1+\cdots+
\|\rho_{u_{r-1}}-\rho_{u_r}\|_1.
$$
The lemma follows from the fact that a DFS walk visits each edge twice.
\end{proof}

Now it suffices to upper bound $\bvalue_\bT(X)$ by
$\tilde{O}(\log n)\cdot \MST(X)$ with probability at least $0.9$
for~a random quadtree $\bT\sim\calT$.
For this purpose, we use the same inspector payment defined in the last subsection
  (the only change is that the set $C$ is now called $X$ which is a set and has size 
   $n$ instead of $2n$).
Recall that for any two points $x,y\in X$, we define
\begin{align}
	\pay_{T}(x,y)  \eqdef \sum_{i\in [h]} \ind\big\{\sv_{i }(x) \neq \sv_{i }(y) \big\} 
	\cdot \Big(\avgd_{x,i-1}+\avgd_{y,i-1}\Big),	
	\label{eq:inspect-a-mst}
\end{align}
where 
$$
\avgd_{x,i-1}\eqdef\Ex_{\cc\sim X_{\sv_{i-1}(x )}} \big[\|x-\cc\|_1\big]\quad\text{and}\quad
\avgd_{y,i-1}\eqdef\Ex_{\cc\sim X_{\sv_{i-1}(y)}} \big[\|y-\cc\|_1\big].
$$

Next we show that the total payment 
  from any spanning tree $G$ is enough to cover $\bvalue_T(X)$.

\begin{lemma}\label{lem:split-criteria-mst}
Let $T$ be any quadtree and $G$ be any spanning tree of $X$. Then we have
\begin{align}
	\emph{\bvalue}_T(X)\leq 2\sum_{(x,y)\in E(G)} \pay_T(x,y). 
	\label{eq:payment-bound-mst}
\end{align}
\end{lemma}
\begin{proof} Let $w$ be the least common ancestor of leaves $\sv_h(x)$, $x\in X$, and let 
	$T^*$ denote the subtree rooted at $w$ that consists of paths from $w$ to $\sv_h(x)$, $x\in X$. 
	It suffices to show that
	$$
	\sum_{(u,v) \in E_{T^*}}  \avgd_{u,v} 
	\le 2\sum_{(x,y)\in E(G)} \pay_T(x,y).	$$
By similar arguments in the proof of Lemma \ref{lem:split-criteria},
 $2\cdot\pay_{T}(x,y)$ is good enough to cover $\avgd_{u,v}$ for every edge along the path between
   the leaf of $x$ and the leaf of $y$. The lemma follows by summing over all $(x,y)\in E(G)$ and noting that the $\avgd_{u,v}$ of each $(u,v)\in E_{T^*}$ is counted at least once.
\end{proof}

To upperbound the total inspector payment from an optimal spanning tree by $\tilde{O}(\log n)\cdot \MST(X)$, we similarly consider the modified payment scheme 
  $\pay^*_{i_0,T}(x,y)$ as in (\ref{eq:def-wtpay}), replacing $C$ by $X$.
The same Lemma~\ref{lem:main-lemma} applies and we use it to prove Lemma \ref{thm:offlineMST}:


\begin{proofof}{Lemma~\ref{thm:offlineMST} assuming Lemma~\ref{lem:main-lemma}}
	The lower bound follows from Lemma~\ref{lem:MST-lb}. 	
	For the upper bound, let $G^*$ be an optimal spanning tree of $X$ and let 
	  $\bT\sim \calT$. By Lemma~\ref{lem:split-criteria-mst} we have
	\begin{align}
	\bvalue_{\bT}(X)\le2 \sum_{(x,y)\in E'(G^*)}\pay_{\bT}(x,y),\label{mst-sum}
	\end{align}
	where $E'(G^*)$ denotes the set of edges $(x,y)$ in $G^*$ with $x\ne y$.
For each $(x,y)\in E'(G^*)$, let 
	\[ 0\le \ell_{x, y} \eqdef \max \big\{0, h_{x,y} - 2\lceil \log_2 n \rceil \big\}\le h_{x,y}. \]
By similar arguments as in the proof of Lemma \ref{quadtreelemmaforemd}, we have with probability at least $1-o(1)$ over the draw of $\bT$ that
	every $(x, y) \in E'(G^*)$ satisfies
	\begin{align}
		\pay_\bT(x,y) =  {\pay}^*_{\ell_{x,y},\bT}(x,y) 
		\label{eq:often-same-mst}
	\end{align}
	Combining (\ref{mst-sum}) and (\ref{eq:often-same-mst}), we have that with probability at least $1 - o(1)$ over the draw of $\bT$,
	\begin{align}\label{hehe44-mst}
		\bvalue_\bT(X)  
		\le 2\sum_{(x,y)\in E'(G^*)}  {\pay}^*_{\ell_{x,y},\bT}(x,y) .
	\end{align}
By applying Lemma~\ref{lem:main-lemma} to every $(x,y) \in E'(G^*)$ with $i_0 = \ell_{x,y}$, as well as Markov's inequality, we have that with probability at least $0.99$ over $\bT$,
the right hand side of (\ref{hehe44-mst}) is at most  
\begin{align*}
	\tilde{O}(\log n) \sum_{\substack{(x,y)\in E'(G^*)}} \|x - y\|_1 = \tilde{O}(\log n) \cdot\MST(X).
\end{align*}
By a union bound, 
$\bvalue_\bT(X) \leq \tilde{O}(\log n) \cdot \MST(X)$ with probability at least $.99 - o(1)\ge 0.9$.
\end{proofof}

%% file: Algorithm.tex
\def\bsv{\boldsymbol{\mathsf{v}}}

\section{Proof of Lemma~\ref{lem:main-lemma}}\label{maintechlemma}
Recall $h=\log 2d$ is the depth of a quadtree.
Let $(a,b)\in M^*$ with $a\ne b$ and $i_0\in [0:h-1]$ with 
\begin{equation} \label{habdef}
i_0\le h_{a,b}\eqdef \left\lfloor \log_2\left( \frac{d}{\|a-b\|_1}\right) \right\rfloor.
\end{equation}
Our goal is to bound the expectation of the $a$-part of $\pay^*_{i_0,\bT}(a,b)$:
\begin{align}\label{mainmainbound}
 \sum_{i> i_0}^{h} \ind\big\{ \sv_{i }(a) \neq \sv_{i}(b) \big\} 
\cdot \avgd^*_{a,i-1}
\end{align}
  over the draw of 
  $\bT\sim \calT$; the bound for the $b$ part is analogous. 
In what follows, all expectations~are taken with respect to $\bT\sim \calT$
  so we will skip $\bT$ in subscripts. 
In particular, we write $\bsv_i(a)$ to denote $\sv_{i,\bT}(a)$ just to emphasize
  that it is a random variable that depends on $\bT$.
  
Given that we always have $\avgd_{a,i}^* \leq 10 d \log n  / 2^i$ for every $i$ by definition, 
  (\ref{mainmainbound}) is trivially $O(d \log n )$ and thus, the statement holds trivially
  when $\|a-b\|_1=\Omega(d)$. So we assume in the rest of the proof 
  that $\|a - b\|_1 \leq d/2$. 
To understand (\ref{mainmainbound}) for  $\bT$,
  we examine the sequence of sets $C_{\bsv_{i}(a)}$, where $\bsv_0(a),$ $\ldots,\bsv_h(a)$ is the root-to-leaf path
  of $a$ in $\bT$.
Recall from the definition of $\calT$ that
  to draw $A_{\bsv_{i}(a)}$ and $B_{\bsv_{i}(a)}$, it suffices to consider 
  independent draws of tuples $$\bI_0, \bI_1 , \dots , \bI_{h-2}$$ with
  $\smash{\bI_i \sim [d]^{2^i}}$ uniformly,
  and then use them to define $C_{\bsv_i(a)}$ as follows:
\begin{align*}
C_{ \bsv_i(a)} &= \big\{x \in C : x_j = a_j \text{ for all $j$ that appears
  in $\bI_0,\ldots,\bI_{i-1}$} \big\} 
\end{align*}
for each $i=1,\ldots,h-1$; $C_{\bsv_0(a)}=C$ since $\bsv_0(a)$ is always the root; 
$C_{\bsv_h(a)}$ contains all copies of $a$ in $C$.
Next we let
\[ \bD = \big\{ (i, \ell):\text{$i \in \{0, \dots, h-2\}$ and $\ell\in [2^i]$ such that }  j = (\bI_{i})_{\ell} \text{ satisfies } a_j \neq b_j\big\}\]
be the set of index pairs of sampled coordinates where $a$ and $b$ disagree, and let
\[ (\bi^{(s)}, \bell^{(s)}) \eqdef \left\{ \begin{array}{cc} \min \bD & \bD \neq \emptyset \\[0.45ex] (*, *) & \bD = \emptyset \end{array} \right. ,\]
where the ordering in $\min \bD$ is lexicographic.\footnote{All pairs $(i,\ell)$
  with $i < \bi^{(s)}$, or with $i=\bi^{(s)}$ and $\ell<\bell^{(s)}$,  satisfy $a_j = b_j$ for $j = (\bI_{i })_{\ell }$.} 
Note, in particular, that the node $\smash{\bsv_{\bi^{(s)} }(a) = \bsv_{\bi^{(s)} }(b)}$ is the least common ancestor of $a$ and $b$ in $\bT$ (hence, the ``s'' in $\bi^{(s)}$ and $\bell^{(s)}$ stands for ``split''). The coordinate which is the first to satisfy $a_j \neq b_j$ is specified by the random variable
\[ \bj^{(s)} \eqdef \left\{ \begin{array}{ll}  \left(\bI_{\bi^{(s)}}\right)_{\bell^{(s)}} & \text{if}\ (\bi^{(s)}, \bell^{(s)}) \neq (*,*) \\[0.5ex]
								* & \text{otherwise} \end{array} \right. , \]
where, notice that, $\bj^{(s)} = *$ only if $a$ and $b$ satisfy $\bsv_k(a) = \bsv_k(b)$ for all 
  $k \in \{0, \dots, h-1\}$. A trivial consequence of the above definitions is that
    for any $k\in \{0,\ldots,h-1\}$,
\begin{align}
\ind\big\{ \bsv_k(a) \neq \bsv_k(b) \big\} &= \sum_{i =0}^{k-1} \sum_{\ell =1}^{2^i} \sum_{\substack{j \in [d] \\ a_j \neq b_j}} \ind\left\{ (\bi^{(s)}, \bell^{(s)}, \bj^{(s)}) = (i, \ell, j) \right\}. \label{eq:divide-separation}
\end{align}
Recall the
  choice of $h_{a,b}\in [0:h-1]$ in (\ref{habdef}), we have that the expectation of (\ref{mainmainbound}) is at most 
\begin{align*}
\sum_{k>i_0}^{h_{a,b}} \Ex\Big[\ind\big\{ \bsv_k(a) \neq \bsv_k(b) \big\}\cdot 
\avgd_{a,k-1}^*\Big]
+ O(\log n) \cdot \|a - b\|_1,
\end{align*}
where we used the fact that $\avgd_{a,i}^* \leq 10d \log n /2^k$ always holds.
As a result $O(\log n)\cdot \|a-b\|_1$ as in the last term is enough to cover the sum over $k>h_{a,b}$ skipped 
  in the above expression.
Using (\ref{eq:divide-separation}), we may re-write the first summand above as
\begin{align}
&\sum_{k>i_0}^{h_{a,b}}\Ex\Big[\ind\big\{ \bsv_k(a) \neq \bsv_k(b) \big\}
\cdot \avgd_{a,k-1}^*\Big]  \label{eq:first-summand}\\
&\qquad = \sum_{k>i_0}^{h_{a, b}} \sum_{i =0}^{k-1} \sum_{\ell=1}^{2^i} \sum_{\substack{j \in [d] \\ a_j \neq b_j}} \Prx\left[(\bi^{(s)}, \bell^{(s)}, \bj^{(s)}) = (i , \ell, j) \right] \cdot \Ex\left[ \avgd_{a,k-1}^* \hspace{0.06cm} \big|\hspace{0.06cm} (\bi^{(s)}, \bell^{(s)}, \bj^{(s)})  = (i , \ell, j) \right] \nonumber. 
\end{align}
Notice that for each $i \in \{ 0, \dots, h-2\}$, $\ell \in [2^i]$, and $j \in [d]$ with $a_j \neq b_j$,
\begin{align}
\Prx\left[ (\bi^{(s)}, \bell^{(s)}, \bj^{(s)}) = (i , \ell, j)\right] &= \frac{1}{d} \left(1 - \frac{\|a-b\|_1}{d} \right)^{
2^0+\cdots+2^{i-1} + \ell - 1}, \label{eq:condition-prob}
\end{align}
which is $1/d$ when $i=0,\ell=1$.
Consider for each $k \in \{i_0 , \dots, h_{a, b}-1 \}$ and $j\in [d]$ with $a_j\ne b_j$, 
\[\hspace{0.04cm} Q_{k, j} \eqdef \max_{\substack{i \in \{0, \dots, k\} \\ \ell \in [2^{i}]}} \Ex\left[ \avgd^*_{a,k} \hspace{0.06cm} \big|\hspace{0.06cm} (\bi^{(s)}, \bell^{(s)}, \bj^{(s)}) = (i , \ell, j)\right], \]
so that invoking (\ref{eq:condition-prob}), we may upper bound (\ref{eq:first-summand}) by
\begin{align*}
\sum_{\substack{j \in [d] \\ a_j \neq b_j}} \sum_{k=i_0}^{h_{a, b}-1} Q_{k,j}\cdot \sum_{i = 0}^{k} \sum_{\ell=1}^{2^i} \frac{1}{d} \left(1 - \frac{\|a-b\|_1}{d} \right)^{2^0+\cdots+2^{i-1} + \ell -  1} \leq  \frac{2}{d}  \sum_{\substack{j \in [d] \\ a_j \neq b_j}} \sum_{k=i_0 }^{h_{a,b}-1} Q_{k,j} \cdot 2^{k}.
\end{align*}
Therefore, it remains to show that for every $j \in [d]$ with $a_j \neq b_j$,
\begin{equation}\label{finaleq}
 \sum_{k=i_0 }^{h_{a,b}-1} Q_{k,j} \cdot 2^k = O(d) \cdot \left( \tilde{O}(\log n) + (h_{a,b} - i_0) \cdot \log \log n \right). 
\end{equation}

To prove this, we start with a lemma that shows that if $Q_{k,j}$ is large, then it must be the case that there are many points
  of distance between (roughly) $d/(2^k\log n)$ and $d\log n/2^k$ from $a$ in $C$.

\begin{lemma}\label{lem:growth}
Fix $j \in [d]$ with $a_j \neq b_j$ and $k \in \{2,\ldots, h_{a,b}\}$. Let $L \subseteq C_j$ be the muti-set 
\begin{align*}
C_j\eqdef \big\{ x \in C : x_j = a_j  \big\} \quad\text{and}\quad
 L \eqdef \left\{ x \in C_j : \|a - x\|_1 \leq   \frac{1}{\log n} \cdot \frac{d}{2^k}\right\}.
\end{align*}
Suppose that for some $i = \{ 0, \dots, k\}$ and some $\ell \in [2^i]$, as well as some $\alpha \geq 800$, we have
\begin{align}
\Ex\left[ \hspace{0.04cm}\emph{\avgd}^*_{a,k} \hspace{0.06cm} \big|\hspace{0.06cm} (\bi^{(s)}, \bell^{(s)}, \bj^{(s)}) = (i , \ell, j) \right] \geq \frac{d}{2^k} \cdot \alpha.  \label{eq:haha22}
\end{align}
Then, the set
\[ H \eqdef \left\{ x \in C_j : \|a - x\|_1 \leq {10 \log n }\cdot \frac{d}{2^k} \right\} \quad\text{satisfies} \quad|H|\ge \frac{\exp\left({\alpha}/{8} \right)}{\log n}\cdot |L|  . \]
\end{lemma}

\begin{proof}
Let $\calbE$ be the event of $(\bi^{(s)}, \bell^{(s)}, \bj^{(s)}) = (i , \ell, j)$. For simplicity in nation, let $\bsv = \bsv_k(a)$ and 
$$ \bC = \bC^*_{a,k}\quad\text{where}\quad 
\bC^*_{a,k}= \left\{ c \in \bC_{\bsv_k(a)} : \|a - c\|_1 \leq \frac{10 d \log n}{2^k} \right\}.
$$
 Every $x \in L$ is of distance at most 
${d}/({2^k \log n})$ from $a$. Since $x \in \bC$ whenever $\bsv_k(x) = \bsv$, we have
\begin{align}
\Ex \Big[ |L \setminus \bC|\hspace{0.06cm} \big|\hspace{0.06cm} \calbE \Big] &= \sum_{x \in L} \Prx\Big[ \bsv_k(x)\ne \bsv \hspace{0.06cm} \big|\hspace{0.06cm} \calbE\Big] \leq |L| \cdot (2^{k}-1) \cdot \frac{d / (2^k \log n)}{d - \|a-b\|_1} \leq \frac{ 2|L|}{\log n}, \label{eq:close-stay}
\end{align} 
because there are $2^{k }-1$ coordinates sampled up to (but not including) depth $k$.
The last inequality above also used the assumption that $\|a-b\|_1\le d/2$.
Then, we have
\begin{align}
\Prx\left[ |\bC | \geq \frac{|L|}{10} \hspace{0.06cm} \Big|\hspace{0.06cm} \calbE  \right] &\geq \Prx\left[ | L \cap \bC | \geq \frac{|L|}{10} \hspace{0.06cm} \Big|\hspace{0.06cm} \calbE \right]  
 = 1 -\Prx\left[ |L \setminus \bC| \geq \frac{9|L|}{10} \hspace{0.06cm} \Big|\hspace{0.06cm} \calbE \right] 
 \geq 1 - \frac{20}{9\log n},\label{eq:close-stay-1}
\end{align}
where the last inequality follows from Markov.
Hence (noting $|\bC|\ge 1$ since it always contains $a$)
\begin{align*}
\avgd^*_{a,k} = \frac{1}{|\bC|} \sum_{x \in \bC} \|a - x\|_1 \leq \frac{d \alpha}{2^{k+1}} + \frac{1}{|\bC|} \sum_{x \in \bC} \ind\left\{ \|a - x\|_1 \geq \frac{d\alpha}{2^{k+1}} \right\}\cdot  \|a - x\|_1.
\end{align*}
Thus, we have (by splitting into two cases of $|\bC|\ge |L|/10$ and $|\bC|<|L|/10$ and applying (\ref{eq:close-stay-1}))
\begin{align*}
\Ex\Big[\hspace{0.04cm}\avgd_{a,k}^* \hspace{0.06cm} \big|\hspace{0.06cm} \calbE \Big] &\leq \frac{d\alpha}{2^{k+1}} + \Ex\left[ \dfrac{1}{|\bC|} \sum_{x \in \bC} \ind\left\{ \|a - x\|_1 \geq \frac{d\alpha}{2^{k+1}}\right\} \|a-x\|_1 \hspace{0.06cm} \Big|\hspace{0.06cm} \calbE \right] \nonumber \\[1ex]
	&\leq \frac{d\alpha}{2^{k+1}} + \frac{10}{|L|} \cdot \Ex\left[ \sum_{x \in \bC} \ind\left\{ \|a-x\|_1 \geq \frac{d \alpha}{2^{k+1}} \right\} \|a - x\|_1 \hspace{0.06cm} \Big|\hspace{0.06cm} \calbE \right] 
	 + \frac{20}{9\log n} \cdot \frac{10 d \log n }{2^k},
\end{align*}
where 
  the final term   
  used the fact that $\|a-x\|_1$ is always at most $10d\log n/2^k$ for any $x\in \bC$, by 
  the definition of the latter. 
Combining (\ref{eq:haha22}) and the inequality above (and $\alpha\ge 800$), we have
\begin{align}
\frac{\alpha}{40}\cdot \frac{d}{2^{k+1}}\cdot |L| &\leq \Ex\left[ \sum_{x \in \bC} \ind\left\{ \|a - x\|_1 \geq \frac{d\alpha}{2^{k+1}} \right\}\|a - x\|_1 \hspace{0.06cm} \Big|\hspace{0.06cm} \calbE \right].\label{eq:haha2} 
\end{align}
We finish the proof of the lemma by upperbounding the right-hand side above in terms of the size of $H$. In particular, let $H'$ be the set of points $x\in H$ with 
  $\|a-x\|_1\ge  d\alpha/2^{k+1}$.
Then 
\begin{align}
\Ex\left[ \sum_{x \in \bC} \ind\left\{ \|a - x\|_1 \geq \frac{d\alpha}{2^{k+1}} \right\}\|a - x\|_1 \hspace{0.06cm} \Big|\hspace{0.06cm} \calbE \right] &\leq \sum_{x\in H'}  \Prx\Big[x \in \bC \hspace{0.06cm} \big|\hspace{0.06cm}
\calbE \Big]\cdot  \left(\frac{10d \log n}{2^{k}} \right).  \label{eq:exp-upper-bound}
\end{align}
We consider two cases: $i=k$ (see $\calE$) and $i<k$. For the easier case of $i=k$, 
  in order for $x \in \bC$ to occur conditioned on $\calbE$, 
  we have that all $2^{k}-1$
  coordinates sampled before depth $k$
  avoid separating $x$ and $a$ conditioning on not separating $a$ and $b$.
The probability of each sample is at most
$$
\frac{d-\|a-x\|_1}{d-\|a-b\|_1}=1-\frac{\|a-x\|_1-\|a-b\|_1}{d-\|a-b\|_1}
\le 1-\frac{\|a-x\|_1-\|a-b\|_1}{d}.
$$
Therefore in this case we have
$$
\Pr\Big[x\in \bC\hspace{0.06cm} \big|\hspace{0.06cm}\calbE\Big]
\le \left(1-\frac{\|a-x\|_1-\|a-b\|_1}{d}\right)^{2^{k}-1}.
$$
Notice that by definition of $h_{a,b}$, and the fact that $k \leq h_{a,b}$,  we have
\begin{align}\label{hehe5}
\|x - a\|_1 - \|a - b\|_1 \geq \frac{d}{2^{k+1}}( \alpha - 2 )
\end{align}
which implies that $\Pr[x\in \bC\hspace{0.06cm}|\hspace{0.06cm}\calbE]$ is 
  at most $\exp(-\alpha/8)$ using $\alpha\ge 800$.

Next we deal with the case when $i<k$.
In order for $x\in \bC$ to occur conditioned on $\calbE$, it
needs~to be the case that coordinates sampled before $(i, \ell)$, of which there are $2^{i } + \ell - 2$ many, avoid separating $x$ and $a$
  conditioning on not separating $a$ and $b$; the $(i, \ell)$-th~sample is $j$
   (which better not separate $x$ and $a$; otherwise the probability is trivially $0$); and the remaining $2^{k}-1 - 2^{i } - \ell+1=2^{k}-2^i-\ell$ coordinates do not separate $x$ and $a$ (but there will be conditioning on not separating $a$ and $b$). So 
\begin{align}
\Prx\Big[ x \in \bC \hspace{0.06cm} \big|\hspace{0.06cm} \calbE \Big] &\leq \left(1 - \frac{\|a - x\|_1 - \|a - b\|_1}{d}\right)^{2^{i } + \ell - 2} \left(1 - \frac{\|a - x\|_1}{d} \right)^{2^{k} - 2^{i } - \ell}\nonumber \\ \nonumber
&\le \left(1 - \frac{\|a - x\|_1 - \|a - b\|_1}{d}\right)^{2^{k } - 2}\\
&\le \left( 1 - \frac{\alpha - 2}{2^{k+1}}\right)^{2^{k-1} - 1},
 \label{eq:split-prob}
\end{align}
which is at most $\exp(-\alpha/8)$ using $k\ge 2$ and $\alpha\ge 800$.
Hence, we can combine  (\ref{eq:haha2}) and (\ref{eq:exp-upper-bound})  to get
\begin{align*}
\frac{\alpha}{40}\cdot \frac{d}{2^{k+1}}\cdot |L| &\leq |H| \cdot \exp\left(-\frac{\alpha}{8} \right) \cdot \frac{10d \log n}{2^{k}}.
\end{align*}
Re-arranging the inequality, 
  the lemma follows using $\alpha\ge 800$.
\end{proof}

The next lemma helps upperbound the number of large $Q_{k,j}$'s. 

\begin{lemma}\label{lem:count-bound}
Fix any $j \in [d]$ with $a_j \neq b_j$. When $\alpha \geq 20 \log \log n$, the set
\begin{align*}
G_{j}(\alpha) =  \left\{k \in \{0 , \dots, h_{a, b}\}  :  Q_{k,j} \geq \frac{d}{2^k} \cdot \alpha \right\} \quad\text{satisfies}\quad \big|G_{j}(\alpha)\big| \leq O\left(\left\lceil\frac{16\log(2n)}{\alpha} \right\rceil\cdot \log\log n \right). 
\end{align*}
\end{lemma}

\begin{proof}
Assume for a contradiction that 
\[ \big|G_j(\alpha)\big| \geq \beta\cdot  \big\lceil \log_2(10\log^2 n)\big\rceil  + 2, \quad\text{where}\quad \beta \eqdef \left\lceil\frac{16\log(2n)}{\alpha} \right\rceil.\]
Then there must be $k_1, \dots, k_\beta \in \{ 2, \dots, h_{a,b}\}$ with $k_1 > k_2 > \dots > k_\beta$ and every $t \in [\beta-1]$ satisfies
\[ k_{t} - k_{t+1} \geq \big\lceil \log_2(10 \log^2 n) \big\rceil .\]
This implies that every $t\in [\beta-1]$ satisfies
\[ \frac{10 d \log n}{2^{k_{t}}} \le \frac{d}{2^{k_{t+1}} \log n}.  \]
If we consider for each $t \in [\beta]$, the following two multi-sets
\begin{align*}
L_t = \left\{ x \in C_j : \|a - x\|_1 \leq \frac{d}{2^{k_t} \log n} \right\} \quad\text{and}\quad H_t = \left\{ x \in C_j : \|a - x\|_1 \leq \frac{10 d\log n}{2^{k_t}} \right\},
\end{align*}
they satisfy
\begin{align} 
L_1 \subseteq H_1 \subseteq L_2 \subseteq H_2 \subseteq \dots \subseteq H_{\beta-1} \subseteq L_{\beta} \subseteq H_{\beta},  \label{eq:nesting}
\end{align}
but then invoking Lemma~\ref{lem:growth} (and using $20 \log\log n \geq 800$), we have that every $t \in [\beta]$ satisfy 
\[ |H_t|\ge \frac{|L_t| \exp( {\alpha}/8)}{\log n}.\]
Using $|L_1| \geq 1$ (since it contains $a$) and (\ref{eq:nesting}), we have
\begin{align*}
|H_{\beta}| &\geq \left(\dfrac{\exp( {\alpha }/{8})}{ \log n }\right)^\beta > 2n,
\end{align*}
using $\exp(\alpha \beta/ 8) \geq (2n)^2$, and $(\log n)^{\beta} <2n$, by our choices of $\beta$ and $\alpha$. This is a contradiction, as we have $H_{\beta} \subseteq A \cup B$ and thus, $|H_\beta|\le 2n$.
\end{proof}

Finally we finish the proof of (\ref{finaleq}). Let 
$\alpha_0= 20 \log \log n$.
Use Lemma~\ref{lem:count-bound} we have
\begin{align*}
\sum_{k = i_0 }^{h_{a,b}-1} Q_{k,j} \cdot 2^k &\leq (h_{a, b} - i_0 ) \cdot \alpha_0 d + \sum_{\kappa = 0}^{\lceil \log_2(10 \log n) \rceil} \big|G_j(\alpha_0 2^{\kappa}) \big| \cdot \alpha_0 2^{\kappa + 1} \cdot d \\[-0.5ex]
	&\leq O(d) \cdot \Big( (h_{a, b} - i_0) \cdot \log \log n + \log n \cdot (\log\log n)^3 \Big),
\end{align*}
where the upper limit of $\kappa \leq \lceil \log_2 (10 \log n) \rceil$ comes from the fact that $G_{j}(10 \log n )$ is trivially empty
  since $\avgd^*_{a,k}$ is always at most $10d\log n/2^k$ by definition. This finishes the proof of Lemma~\ref{lem:main-lemma}.

%% file: EMD-sketch.tex


\def\MST{\text{EMD}}
\def\sv{\mathsf{v}} \def\su{\mathsf{u}}
\def\bcalI{\boldsymbol{\mathcal{I}}}
\def\bfeta{\boldsymbol{\eta}}
\def\rr{\mathbf{r}} \def\bone{\mathbf{1}}
\def\vv{\mathbf{v}} \def\LS{\mathsf{LS}}
\def\Exp{\text{Exp}} \def\ALG{\mathsf{ALG}}
\def\cc{\mathbf{c}} \def\bDelta{\mathbf{\Delta}}
\def\bh{\mathbf{h}} \def\bsv{\boldsymbol{\sv}}
\def\uu{\mathbf{u}} \def\calU{\mathcal{U}}
\def\aa{\mathbf{a}} \def\bE{\boldsymbol{E}}
\def\bOmega{\boldsymbol{\Omega}} \def\balpha{\boldsymbol{\alpha}}
\def\bomega{\boldsymbol{\omega}} \def\bA{\mathbf{A}}
\newcommand{\calbQ}{\boldsymbol{\calQ}}



\section{Streaming Preliminaries and Sketching Tools}

\subsection{Exponential Order Statistics}\label{sec:expo}
We review some properties of the order statistics of independent non-identically distributed exponential random variables. 
Let $(\bt_1,\dots,\bt_n)$ be independent exponential random variables where $\bt_i$ has mean $1/\lambda_i$ (equivalently, $\bt_i$ has rate $\lambda_i>0$), abbreviated as $\bt_i\sim\Exp(\lambda_i)$. Recall that $\bt_i$ is given by the cumulative distribution function (cdf) $\Pr[\bt_i\le x] = 1-e^{-\lambda_i x}$. Our algorithm will require an analysis of the distribution of values $(\bt_1,\dots,\bt_n)$. We begin by noting that constant factor scalings of an exponential variable result in another exponential variable.

\begin{fact}[Scaling of exponentials]\label{fact:scale}
	Let $\bt\sim \Exp(\lambda)$ and $\alpha > 0$. Then $\alpha \bt$ is distributed as $\Exp(\lambda/\alpha)$.
\end{fact}
\begin{proof}
	The cdf of $\alpha \bt$ is given by $ \Pr[\bt < x/\alpha]=1-e^{-\lambda x/\alpha}$, which is the cdf of $\Exp(\lambda/\alpha)$.
\end{proof}

\noindent
\begin{definition}
	Let $\bt = (\bt_1,\dots,\bt_n)$ be independent exponentials. For $k=1,2,\dots,n$, we define the \textit{$k$-th anti-rank} $D_{\bt}(k)\in [n]$ of $\bt$ to be the values $D_{\bt}(k)$ such that $\bt_{D_{\bt}(1)} \leq \bt_{D_{\bt}(2)} \leq \dots \leq \bt_{D_{\bt}(n)}$. 
\end{definition}\noindent
Using the structure of the anti-rank vector, it has been observed \cite{nagaraja2006order} that there is a simple form for describing the distribution of $\bt_{D_{\bt}(k)}$ as a function of $(\lambda_1,\dots,\lambda_n)$ and the anti-rank vector.

\begin{fact}[\cite{nagaraja2006order}]\label{fact:order}
	Let $\bt=(\bt_1,\dots,\bt_n)$ be independently exponentials with 
	  $\bt_i\sim \Exp(\lambda_i)$. Then $D_{\bt}(1)=i$ with probability $\lambda_i / \sum_{j \in [n]} \lambda_j$. Furthermore, the following two sampling procedures produce the same distribution over pairs in $\R^2$: 
	\begin{enumerate}
	\item Sample $\bt=(\bt_1, \dots, \bt_n)$, where $\bt_i \sim \Exp(\lambda_i)$, and output $(\bt_{D_{\bt}(1)}, \bt_{D_{\bt}(2)} - \bt_{D_{\bt}(1)})$.
	\item Sample $\bi_1 \in [n]$ with $\Pr[\bi_1 = i] = \lambda_i / \sum_{j\in [n]} \lambda_j$, $\boldsymbol{E}_1, \boldsymbol{E}_2 \sim \Exp(1)$ independently, and output  
\[ \left(\frac{\bE_1}{\sum_{j\in[n]} \lambda_j}, \hspace{0.06cm}\frac{\bE_2}{\sum_{j\in[n]\setminus \{\bi_1\}} \lambda_j}\right). \]
	\end{enumerate} 
\end{fact}

\begin{proof}
This is a simple computation. We have that for any $r, r' \in \R_{\geq 0}$ and $i \in [n]$,
\begin{align*}
\Prx_{\bt}\left[ \begin{array}{c} D_{\bt}(1) = i, \\ \bt_{D_{\bt}(1)} \geq r, \\ \bt_{D_{\bt}(2)} - \bt_{D_{\bt}(1)} \geq r' \end{array}  \right] &= \int_{y:r}^\infty \lambda_i \exp\left(- \lambda_i y \right)\prod_{j \in [n] \setminus \{i\}} \Prx_{\bt_j \sim \Exp(\lambda_j)}\left[ \bt_j - y\geq r' \right] dy \\
	&= \lambda_i \exp\left(-r' \sum_{j \in [n] \setminus \{i\}} \lambda_j \right) \int_{y:r}^{\infty} \exp\left( - y \sum_{j=1}^n \lambda_j \right) dy\\
	&= \dfrac{\lambda_i}{\sum_{j=1}^n \lambda_i} \cdot \exp\left(-r' \sum_{j \in [n] \setminus \{i\}} \lambda_j \right) \cdot \exp\left( - r\sum_{j=1}^{n} \lambda_i \right)  \\
	&= \Prx\left[\bi_1 = i \wedge \dfrac{\bE_1}{\sum_{j\in[n]} \lambda_j} \geq r \wedge \dfrac{\bE_2}{\sum_{j\in[n]\setminus\{\bi_1\}} \lambda_j} \geq r' \right].
\end{align*}
This finishes the proof.
\end{proof}

\begin{lemma}\label{lem:sum_exponentials}
	Fix $n\in \N$ and $x \in \R^n$ to any fixed vector, and consider independent draws $\bt_1,\dots,\bt_n \sim \Exp(1)$. For any $\gamma \in (0,1/2)$, 	
	\[	\Prx_{\bt_1,\ldots,\bt_n}\left[{\sum_{i\in[n]}\frac{|x_i|}{\bt_i} \geq \frac{4\log (n/\gamma)}{\gamma} \|x\|_1 } \right]\leq 2\gamma	.\]	
\end{lemma}
\begin{proof}
	We want to compute the expectation of $\sum |x_i|/\bt_i$ and apply Markov's inequality; however, the above random variable does not have an expectation (since $1/\bt_i$ may become too large when $\bt_i$ is small). To remedy that, we effectively truncate $1/\bt_i$. Notice that for any $\alpha > 0$, the probability that any $\bt_i \sim \Exp(1)$ is less than $\alpha$ is at most $\alpha$ (simply by inspecting p.d.f of $\Exp(1)$). We apply this in two ways: (i) for any $i \in [n]$, the probability that $1/\bt_i > n/\gamma$ is at most $\gamma / n$, and (ii) for any $i\in [n]$ and $j \in \N$, the probability that $1/\bt_i \in [2^{j}, 2^{j+1}]$ is at most $1/2^{j}$. 
	
	Letting \smash{$\calbQ_i$} be the event that $1/\bt_i> n/\gamma$, and by (i), we may union bound over all $n$ to say $\calbQ_i$ is never satisfied with probability at least $1-\gamma$. 
	So now we compute the expectation and apply Markov's inequality (conditioning on $\neg (\cup_i \calbQ_i)$ so the probability that $1/\bt_i > n/\gamma$ is $0$): 
	\begin{align*}
		\Ex_{\bt_1, \dots, \bt_n}\left[ \sum_{i=1}^n \frac{|x_i|}{\bt_i} \mid \neg (\cup_i \calbQ_i)\right] &\leq \sum_{i=1}^n |x_i|  \sum_{j=0}^{\infty} 2^{j+1} \cdot\Prx_{\bt_i}\left[ \frac{1}{\bt_i} \in [2^{j}, 2^{j+1}] \mid \neg \calbQ_i \right] \\ 
		&\leq \sum_{i=1}^n |x_i| \sum_{j=0}^{\log(n/\gamma)} 2^{j+1} \cdot\Prx_{\bt_i}\left[ \frac{1}{\bt_i} \in [2^{j}, 2^{j+1}] \mid \neg \calbQ_i \right] \\
		&\leq \sum_{i=1}^n |x_i| \sum_{j=0}^{\log(n/\gamma)} 2^{j+1} \cdot \frac{1}{2^j (1-\gamma)} \leq \|x\|_1 \cdot 2\log(n/\gamma) / (1-\gamma).
	\end{align*}
	\ignore{Next we fix $i\in [n]$ and compute $\Ex\left[{\frac{1}{\bt_i}\; | \;\neg \calbQ_i}\right]$. Using the earlier fact that $\pr{\bt_i < \eps} < \eps$, for any $\eps > 0$, it follows that $\pr{2^j \leq \frac{1}{\bt_i} \leq 2^{j+1}}< \frac{1}{2^j}$ for any $j \geq 0$. Thus 
		
		\begin{align*}
			\Ex_{\bt_i}\left[\frac{1}{\bt_i}\mid\neg\calbQ_i\right] &\leq 		\Ex_{\bt_i}\left[{\min\left\{\frac{1}{\bt_i}, \frac{n}{\gamma}\right\}}\right]< \frac{n}{\gamma} \cdot \Prx_{\bt_i}\left[{ \frac{1}{\bt_i} > \frac{n}{\gamma}}\right] +\sum_{i=0}^{\log(n/\gamma)} 2^{i+1} \Prx_{\bt_i}\left[2^i \leq \frac{1}{\bt_i} \leq 2^{i+1}\right] \\
			& \leq 	1 + \sum_{i=1}^{\log(n/\gamma)} 2 \leq 4 \log \left(\frac{n}{\gamma}\right).\\
	\end{align*}}
	Then, by Markov's inequality,
	\[		\Prx_{\bt_1,\ldots,\bt_n}\left[{\sum_{i\in[n]} \frac{|x_i|}{\bt_i} \geq \frac{4\log(n/\gamma)}{\gamma} \cdot \|x\|_1 \; \Big|\;  \bigcup_i \neg \calbQ_i}\right] \leq \gamma	.\]
	The proof follows from a union bound.
\end{proof}

\begin{lemma}\label{lem:gap_exponentials} For $n\in \N$, let $x \in \R^n$ be any fixed vector. Let $\bt_1,\dots,\bt_n \sim \Exp(1)$ be i.i.d. exponentially distributed, and $\gamma > 0$ be smaller than some constant. Then,	letting $i^*=\arg\max \frac{|x_i|}{t_i}$ we have that 
	\begin{align}
	 \frac{|x_{i^*}|}{\bt_{i^*}}\ge \gamma \|x\|_1\qquad\text{and}\qquad  \frac{|x_{i^*}|}{\bt_{i^*}}\ge (1+\gamma)\max_{i\neq i^*}\left\{\frac{|x_i|}{\bt_i}\right\}, \label{Eq:gap_exponentials}	 
	 \end{align}
	
	holds with probability at least $1-4\gamma$.
\end{lemma}
\begin{proof}
	Let $\calbE$ is the first event in  (\ref{Eq:gap_exponentials}), and $\calbE'$ is the second.
	 By Fact \ref{fact:order}, the quantity $|x_{i^*}| /\bt_{i^*} $ is distributed as $\|x\|_1/\boldsymbol{E}_1$, where $\boldsymbol{E}_1$ is an exponential random variable. Using the cdf $1-e^{-x}$ of an exponential, we have that $\Prx[{\boldsymbol{E}_1 > 1/\gamma} ]< e^{-1/\gamma}$, thus $\Prx[{\calbE }] > 1-  e^{-1/\gamma} > 1-\gamma$ for $\gamma$ smaller than a constant. 
	
	We now bound $\Prx[{\calbE'}]$. Again by Fact \ref{fact:order}, we have that $\max_{i\neq i^*}\left\{\frac{|x_i|}{\bt_{i}}\right\}$ is distributed as
	\begin{align*}
		\max_{i\neq i^*}\left\{\frac{|x_i|}{\bt_{i}}\right\}& \sim		\frac{1}{\frac{\boldsymbol{E}_1}{\|x\|_1} + \frac{\boldsymbol{E}_2}{\|x\|_1 - |x_{i^*}|} } \leq \frac{\|x\|_1}{\boldsymbol{E}_1 + \boldsymbol{E}_2}.
	\end{align*}
	Thus, $ \calbE'$ will hold so long as $\boldsymbol{E}_2 \geq 3 \gamma \boldsymbol{E}_1$. The probability that this does not occur occurs is given by 
	\begin{align*}
		\Prx[{\boldsymbol{E}_2 < 3 \gamma \boldsymbol{E}_1}] & \leq \int_{0}^\infty e^{-x} \Prx\left[{\boldsymbol{E}_1 > \frac{x}{3\gamma}}\right]  dx = \int_{0}^\infty e^{-x(1+  \frac{1}{3\gamma})}  dx  = \frac{1}{(1+  \frac{1}{3\gamma})} \leq 3 \gamma.
	\end{align*}
	Thus, we have $\Prx[{\calbE'}] > 1- 3 \gamma$, and so by a union bound the event in (\ref{Eq:gap_exponentials}) holds with probability at least $1-4\gamma$, which completes the proof.
\end{proof}

Given $x \in \R^n$ and an integer $\beta \geq 1$, we write $x_{-\beta} \in \R^n$ to denote the vector given by $x$ where the $\beta$ largest coordinates in magnitude are set to $0$. When $\beta \geq 1$ is not an integer, $x_{-\beta}$ is interpreted as $x_{- \lfloor \beta \rfloor}$. 
We will need the following lemma which bounds the tail $\ell_2$ and $\ell_1$ norm 
  of a vector after its entries are scaled independently by random variables

\begin{lemma}[Generalization of Proposition $1$ of \cite{JW18}]\label{lem:sizebound}
	Fix $n \in \N$ and $c \geq 0$, and let $\calD_1,\ldots,\calD_n$ be a sequence of distributions over $\R$ satisfying 
	\[ \Prx_{\bt_i \sim \calD_i}\big[|\bt_i| \geq y\big] \leq \frac{c}{y}, \quad \text{ for all $i \in [n]$ and all $y\red{>0}$}. \]
	For any fixed vector $x \in \R^n$ and integer $\beta \geq 1$, consider the random vector $\bz \in \R^n$ given by letting
	\[ \bz_i \eqdef \bt_i \cdot x_i, \quad\text{ where $\bt_i \sim \calD_i$ independently for all $i \in [n]$}. \] 
	Then we have 
	$$\|\bz_{-\beta}\|_2 \leq \frac{12c}{\sqrt{\beta}}\cdot  \|x \|_1 \quad\text{and}\quad
	\|\bz_{-\beta}\|_1 \leq 9c \lceil \log_2 n \rceil\cdot \|x\|_1 $$ with probability at least $1-3 e^{-\beta/8}$ over the draws of $\bt_i \sim \calD_i$.
\end{lemma}
\begin{proof}
Assume without loss of generality that $\beta\ge 8$; otherwise $1-3e^{-\beta/8}<0$ and the statement~is trivial.
Define the following random sets $\bI_j$ for each $j=0,1,\ldots,\lceil \log_2 n\rceil$: 
	\[ \bI_j = \left\{i \in [n] : \frac{c \|x\|_1}{2^{j+1}} \leq |\bz_i| \leq \frac{c \|x\|_1}{2^{j}} \right\}\]
 and notice that, for any $i \in [n]$, we have
 \[ \Prx_{\bt_i \sim \calD_i}\big[i \in \bI_j\big] \leq \Prx_{\bt_i \sim \calD_i}\left[ |\bt_i| \geq \frac{c  \|x\|_1}{2^{j+1}|x_i|} \right] \leq \frac{2^{j+1} |x_i|}{\|x\|_1} \qquad\text{implying} \qquad  \Ex_{\bt_1, \dots, \bt_n}\big[|\bI_j| \big]  \leq 2^{j+1}. \] 
Let $\calbE_1$ denote the event that there exists $j \geq \lceil \log_2(\beta/8) \rceil$ such that $|\bI_j| > 4 \cdot 2^{j+1}$. Then
\begin{align*}
\Prx_{\bt_1, \dots, \bt_n}\left[ \calbE_1 \right] \leq \sum_{j=\lceil \log_2(\beta/8)\rceil}^{\lceil \log_2 n \rceil} \Prx_{\bt_1, \dots, \bt_n}\left[ |\bI_j| > 4 \cdot 2^{j+1}\right] 
 \leq \sum_{j=\lceil \log_2(\beta/8)\rceil}^{\lceil \log_2 n \rceil}\exp ( - 2^{j} ) \leq 2 \exp(-\beta / 8), 
\end{align*}
by a Chernoff bound.
%
%
On the other hand, each $i \in [n]$ satisfies $|\bz_i| \geq 4 c \|x\|_1 / \beta$ with probability at most $\beta |x_i| / (4 \|x\|_1)$ over the draw of $\bt_i \sim \calD_i$. The event $\calbE_2$ that more than $\beta$ indices  $i \in [n]$ satisfy $|\bz_i| \geq 4 c \|x\|_1 / \beta$ happens with probability at most $\exp(-\beta/8)$.
%
%
Whenever $\calbE_1$ and $\calbE_2$ do not occur, (which happens with probability at least $1-3e^{-\beta/8}$), we have 
\begin{align*}
\|\bz_{-\beta}\|_2^2 &\leq \sum_{j=\lceil\log_2(\beta/8)\rceil}^{\lceil\log_2 n\rceil}|\bI_j| \cdot \frac{c^2 \|x\|_1^2}{2^{2j}} 	 + n \cdot\frac{c^2 \|x\|_1^2}{n^2} \leq  
\frac{129}{\beta}\cdot c^2 \|x\|_1\qquad\text{and}  \\[0.6ex]
\|\bz_{-\beta}\|_1 &\leq \sum_{j= \lceil \log_2 (\beta/8)\rceil}^{\lceil \log_2 n \rceil} |\bI_j| \cdot \frac{c\|x\|_1}{2^j} + n \cdot \frac{c \|x\|_1}{n} \leq c \|x\|_1 \left( 8 \lceil \log_2 n \rceil + 1 \right)
\end{align*}
since once $\beta \geq n$, the bound trivially becomes $0$. The lemma follows.
\end{proof}

\subsection{Geometric Streaming and Linear Sketching}\label{sec:linearsketching}
All of the streaming algorithms in this paper will be \textit{linear sketches}. We begin by formalizing what a linear sketch is in the context of geometric streaming algorithms. Recall that in the geometric streaming model, we receive a stream of updates $(p_1,\sigma_1),(p_2,\sigma_2),\dots,(p_m,\sigma_m)$, where $p_i \in \{0,1\}^d$ and $\sigma_i \in \{1,-1\}$ indicates either an insertion or deletion of the point $p_i$ from the active dataset $X \subset \{0,1\}^d$. For the case of $\EMD$, we also need to specify whether a given point $p_i$ is being inserted or deleted from $A$ or $B$, in which case the stream consists of updates of the form $(p_i,\sigma_i,\theta_i)$ where $\theta_i \in \{A,B\}$ indicates which of the two sets the update applies to.

We now observe that the above geometric streaming model is just a special case of the standard turnstile streaming model \cite{babcock2002models,muthukrishnan2005data}, 
 which consists of a sequence of insertions and deletions to the coordinates of a high-dimensional vector $f \in \R^{N}$. Specifically, in the standard streaming model, the stream consists of updates $(i_1,\Delta_1),\dots,(i_m,\Delta_m)$, where $i_t \in [N]$ is a coordinate and $\Delta_t \in \{1,-1\}$. Here, the update $(i_t,\Delta_t)$ causes the change $f_{i_t}\leftarrow f_{i_t} + \Delta_t$. Thus, by simply setting $f=f_X \in \R^{2^d}$ to be the indicator vector (with multiplicity) of the multi-set $X$ for the case of $\MST$, or setting $f=f_{A,B} \in \R^{2 \cdot2^d}$ to be the indicator vector (with multiplicity) of the two sets $A$ and $B$ stacked together, then a stream of insertions and deletions of points $p$ in the geometric model coincides to insertions and deletions to the coordinates of $f$. Given this connection, we can now define a linear sketch.

 \begin{definition}\label{def:linearsketch}
 	Given a stream of updates to a vector $f \in \R^N$, a (one-pass) \textit{linear sketch} generates a random matrix $\bS \in \R^{k \times N}$, and stores only the matrix-vector product $\bS f$ and the matrix $\bS$. At the end of the stream, it answers a query based on $\bS f,\bS$. A two-pass \textit{linear sketch} generates random matrix $\bS^1$, and on the first pass stores only $\bS^1 f, \bS^1$. After the first pass over the data, it generates a random matrix $\bS^2$ (possibly depending on $\bS^1 f, \bS^1$), and after the second pass it answers a query based on $\bS^1 f,\bS^1,\bS^2 f, \bS^2$. 
 	
 	The space used by a one-pass (resp. two-pass) linear sketch is the space required to store $\bS f$ (resp. $\bS^1 f,\bS^2 f$.) 	
 	\end{definition}

For simplicity, in the entirety of Sections \ref{sec:EMDSketch} and \ref{sec:MSTSketch}, we will defer consideration of the \textit{bit complexity} required to store entries of the sketches $\bS f$, and instead focus on bounding the dimension of the sketch $\bS f$ --- for instance, in both sections we will often work with real-valued random variables (with unbounded bit complexity). Then in Appendix \ref{app:formal-models}, we will handle the issue of bit-complexity, by demonstrating that we can generating all random variables to $\polylog(n)$-bits of precision. This will result in the space of the algorithm being within a factor of $\polylog(n)$ of the dimension of the sketch. 
 
 Notice in the definition of a linear sketch, the space complexity did not depend on the space required to store the matrix $\bS$, while the output the algorithm is allowed to depend on $\bS$. This coincides with what is known as the \textit{random oracle model} of streaming, or the \textit{public coin} model of communication complexity. Formally:
 
  \begin{definition}\label{def:randoracle}
 	In the random oracle streaming model, the algorithm is given access to an arbitrarily long string of random bits which do not count against the space complexity. In particular, the space of a one-pass (resp. two-pass) linear sketching algorithm in the random oracle model is just the space required to store $\bS f$ (resp. $\bS^1 f,\bS^2 f$.) 	
 \end{definition}

Working in the random oracle model of streaming is common in the sketching literature, as nearly all lower bounds for streaming are derived from the public coin model of communication complexity (and therefore apply to the random oracle model). Moreover, in applications of sketching to distributed computation, the assumption of a random oracle is often founded. 

On the other hand, we demonstrate that the random oracle model is not required for our algorithms. In particular, we show how the assumption of a random oracle can be removed in Appendix \ref{app:formal-models}, albeit at the cost of an additive $d\cdot \polylog n$ in the space of the algorithm. Note that even to read a single update in the stream requires $\Omega(d)$-bits of space, so any algorithm requires $\Omega(d)$ \textit{working space} (see remark below). Thus, this additive $d \cdot \polylog n$ in the space resulting from derandomization is comparable, up to $\log$ factors, to the space required to even read an update. 

Since the issues of bit complexity and derandomization are handled in Appendix \ref{app:formal-models}, for the remainder of Sections \ref{sec:EMDSketch} and \ref{sec:MSTSketch} we assume the random oracle model, and analyze space in terms of the dimension of the sketch




\begin{remark}
	When discussing the space complexity of streaming algorithms, there are two separate notions: \textit{working space} and \textit{intrinsic space}. Oftentimes these two notations of space are the same (up to constant factors) for streaming algorithms, however for our purposes it will be useful to distinguish them, since a $d$-dimensional point itself requires $\Omega(d)$-space to specify. The \textit{working space} of a streaming algorithm $\mathcal{A}$ is the space required to store an update $(i_t,\Delta_t)$ in the stream and process it. The \textit{intrinsic space} is the space which the algorithm must store \textit{between} updates. The intrinsic space coincides with the size of a message in the area of communication complexity, where two parties hold a fraction of the input, and must exchange messages to approximation a function of the entire input. The intrinsic space also coincides with the definition of space given in Definition \ref{def:linearsketch}. 
In this paper, we focus on intrinsic space, which we will hereafter just refer to as the \textit{space} of the algorithm.
\end{remark}

\subsection{Count-Sketch, $\ell_1$-sketch, and $\ell_1$-sampling}\label{sec:usefulsketches}

We now introduce several useful sketches from the streaming literature. We remark that all the sketches below are already derandomized (their do not require the random oracle model), and the space complexity in bits is stated within the theorems. 

\begin{theorem}[Count-Sketch~\cite{charikar2002finding}]\label{thm:countsketch}	
Fix $n \in \N$ and $\eps \in (0,1)$.
There is a $O(\log n/\eps^2)$-bits of space linear sketch that, given any input vector  $x\in \R^n$,
  outputs $\widehat{\bx}\in \R^n$ such that 
 $
  \| \widehat{\bx} - x\|_{\infty} \leq \eps \| x_{-1/\eps^2}\|_2 
$  
with probability at least $1-1/\poly(n)$.

\ignore{
given only $Sx \in \R^k$ and $S$, one can recovery a vector $y \in \R^n$ such that with probability $1-1/\poly(n)$, we have
	\[	\|x - y\|_\infty \leq \eps \|x_{-1/\eps^2}\|_2	\]
	and for a vector $x$ and value $k$, $x_{- k}$ denotes the result of setting equal to zero the $k$ largest coordinates (in absolute value) of $x$ equal to zero. Moreover, the matrix $S$ can be stored in $O(k \log )$ bits of space. }
\end{theorem}

\begin{theorem}[$\ell_1$-sketch~\cite{I06,kane2010exact}] \label{lem:cauchy}
Fix $n \in \N$ and $\eps, \delta \in (0, 1)$, and $s = O(\log(1/\delta)/\eps^2)$. Let 
$\bOmega$ be an $s \times n$ matrix with independent Cauchy random variables $\bOmega_{i,j}\sim\calC$.
For any $x\in \R^n$, let
$$
 {\bfeta}=\frac{\median\{ (\bOmega x)_j: j\in [s]\}}{\median(|\calC|)}.
$$
Then we have $|\bfeta-\|x\|_1|\le \eps\|x\|_1$ with probability at least $1-\delta$
  over the randomness of $\bOmega$. Moreover, the matrix $\bOmega$ can be generated with limited independence so that it can be stored in $O(\log(1/\delta) \log(n) /\eps^2)$-bits of space, and so that the above guarantees still hold.
%
\ignore{	
	Fix $\eps,\delta \geq $. Let $x \in \R^n$ be any vector, and let $\CC \in \R^{k \times n}$ be a matrix of i.i.d. Cauchy random variables, where $k = \Theta(\eps^{-2} \log \delta^{-1})$. Let $y = \CC x$. Then with probability $1-\delta$, we have $\text{median}_{i \in [k]} \frac{|y_i|}{\texttt{median}(|\mathcal{D}_1|)}  = (1 \pm \eps)\|x\|_1$. 	}
\end{theorem}


\begin{theorem}[Perfect $\ell_1$-sampling \cite{JW18}]\label{thm:ell-1-sa}
For $m \in \N$, let $c > 1$ be an arbitrarily large constant and $t = O(\log^2 m)$. There exists a distribution $\calS_a^{1}(m, t)$ supported on pairs $(\bS, \Alg_{\bS}^{1})$ where $\bS$ is an $t \times m$ matrix and $\Alg_{\bS}^{1}$ is an algorithm which receives as input a vector $y \in \R^t$ and outputs a failure symbol $\bot$ with probability at most $1/3$, otherwise it returns an index $j \in [m]$. For any $x \in \R^m$ and $j \in [m]$,
\begin{align*}
\left| \Prx_{(\bS, \Alg_\bS^{1}) \sim \calS_{a}^1(m, t)}\left[ \Alg_{\bS}^{1}(\bS x) =j \mid \Alg_{\bS}^1(\bS x) \neq \bot \right] - \frac{|x_j|}{\|x\|_1} \right| \leq \frac{1}{m^c}.
\end{align*}
Lastly, the matrix $\bS$ be generated with limited independence so that it can be stored in $O(\log^2 n(\log \log n)^2 )$ bits of space, and so that the above guarantees still hold.
\end{theorem}

\section{Linear Sketches for EMD}\label{sec:EMDSketch}

We give linear sketches for EMD in this section and prove the following theorems:

\begin{theorem}\label{thm:one-pass}
 	For $d,n \in \N$, there is a  $O(1/\eps)\cdot \polylog(n,d)$-space linear sketching algorithm that, given multi-sets $A,B \subset \{0,1\}^d$ with $|A| = |B| = n$, outputs a number $\bfeta$ such that
 	\begin{equation}\label{mainmainmain}
\EMD(A,B) \le {\bfeta}\le \tilde{O}(\log n)\cdot \EMD(A,B) +\eps nd
\end{equation}
	with probability at least $2/3$. 
 \end{theorem}

\begin{theorem}\label{thm:two-pass}
 	For $d,n \in \N$, there is a 2-round $\polylog(n, d)$-space linear sketching algorithm that, given multi-sets $A,B \subset \{0,1\}^d$ with $|A| = |B| = n$, outputs a number $\bfeta$ such that
 	\[	\EMD(A,B) \leq  \bfeta \leq \tilde{O}(\log n)\cdot \EMD(A,B) 	\]
 	with probability at least $2/3$. 
 \end{theorem}

\subsection{Preparation}

Fix $d\in \N$, and fix a quadtree $T$ of depth $h=\lceil \log_2 d\rceil$.
For each $i \in [0:h]$, we let $V_i$ be the set of nodes of $T$ at depth $i$; for a node $v \in V_i$ for $i > 0$, we let $\pi(v)$ be the parent node of $v$ in $T$ (which is a node at depth $i-1$).
For each $i\in [0:h]$, $T$ induces a map $\sv_i:\{0,1\}^d\rightarrow V_i$:
  $\sv_0(a),\sv_1(a),\ldots,\sv_h(a)$ is the path of each point $a\in \{0,1\}^d$ going down the quadtree $T$.

Fix $n\in \N$. Let $A,B\subseteq \{0,1\}^d$ be the pair of input sets of size $n$ each. 
We use $A_v$ (or $B_v$) for each $v\in V_i$ 
  to denote the set of $a\in A$ (or $b\in B$) with $\sv_i(a)=v$ (or $\sv_i(b)=v$),
  with $C_v=A_v\cup B_v$.
Let 
$$
\Delta_i=\sum_{v\in V_i} \big||A_v|-|B_v|\big|.
$$
When $\Delta_i>0$ (which is the main case we will work on),
  we write $\calV_i$ to denote the distribution over $V_i$ where $v\in V_i$
  is sampled with probability 
$||A_v|-|B_v||/\Delta_i$.
For each $i\in [h]$,
  we let $\calS_i$ denote the distribution over subsets of $[d]$
  which includes each coordinate independently with probability 
  $$\alpha_i\eqdef\frac{2^i}{d\log^2 n}.$$
Given $S \subset [d]$, the character $\chi_S \colon \{0, 1\}^d \to \{-1,1\}$ is the function  $\chi_S(x) = (-1)^{\sum_{k \in S} x_k}$.

We use the quadtree lemma for EMD (Lemma \ref{quadtreelemmaforemd}) in Section \ref{sec:quadtreeemd} to get the
  following lemma that is key to our linear sketches:



\begin{lemma}\label{mainsketchestimate}
At least $90\%$ of quadtrees $T$
  (as drawn from the distribution $\calT$) satisfy 
\begin{equation}\label{quadtreeinequality}
\EMD(A,B) \le \sum_{i\in [h]} \calI_i\le \tilde{O}(\log n)\cdot  \EMD(A,B)\quad\ \text{and}\ \quad
\sum_{i\in [h]} \frac{\Delta_i d}{2^i}\le \tilde{O}(\log n)\cdot \EMD(A,B) 
\end{equation}
where $\calI_i$ for each $i\in [h]$ is defined as $\calI_i=0$ when $\Delta_i=0$ and 
$$
\calI_i\eqdef\frac{2\Delta_i}{\alpha_i}\cdot \Ex_{\vv\sim \calV_i,\bS\sim\calS_i}
\big[p_{\pi(\vv),\vv,\bS}\big]
$$
when $\Delta_i>0$, where for each $v\in V_i$ and $u=\pi(v)$,
$$
p_{u,v,S}\eqdef  \Prx_{\cc_u\sim C_u,\cc_v\sim C_v}\big[\chi_S(\cc_u)\ne \chi_S(\cc_v)\big]\in [0,1].
$$
\end{lemma}
\begin{proof}
We show that $98\%$ of quadtrees satisfy the first part of (\ref{quadtreeinequality})
  and $98\%$ of quadtrees satisfy the second part of (\ref{quadtreeinequality}) and the lemma follows.
The second part follows from an analysis similar to (\ref{eq:grow-dist}) (although (\ref{eq:grow-dist}) only gives $O(\log d)$ on the right hand~side instead of the $O(\log n)$ we need, one can improve it to $\min(\log n,\log d)$; see footnote \ref{usefulfootnote}). 
We focus on the first part of (\ref{quadtreeinequality}).

Recall from Lemma \ref{quadtreelemmaforemd} that at least $99\%$ of quadtrees $T$ satisfy
$$
 \EMD(A,B)\le  \bvalue_T(A,B)=\sum_{(u,v) \in E_T} \big| |A_v| - |B_v| \big| \cdot 
 {\avgd}_{u,v}\le \tilde{O}(\log n)\cdot \EMD(A,B).
$$
We can write $\bvalue_T(A,B)$ as the sum of $\calI_i^*$ over $i\in [h]$: $\calI_i^*=0$ when $\Delta_i=0$ and
$$
\calI_i^*\eqdef \Delta_i\cdot \Ex_{\vv\sim\calV_i}\left [\avgd_{\pi(\vv),\vv}\right ]
=\Delta_i\cdot \Ex_{\vv\sim\calV_i}\left [\Ex_{\cc\sim C_{\pi(\vv)},\cc'\sim C_{\vv}}\left[\|\cc-\cc'\|_1\right] \right ]. 
$$
when $\Delta_i>0$. Also recall from the proof of Lemma \ref{quadtreelemmaforemd} that in at least $(1-o(1))$-fraction of $T$,
  every two points $c,c'\in C_{\pi(v)}$ for a node $v$ at level $i$ have $\|a-b\|_1\le O(d\log n/2^i)$.
When this happens, we have from the following claim that $\calI^*\le \calI\le 2\calI^*$.

\begin{claim}\label{charlemma}
Let $a,b\in \{0,1\}^d$ with $\|a-b\|_1=O(d\log n/2^i)$.
Then we have
$$
\frac{\alpha_i \|a-b\|_1}{2}\le  \Pr_{\bS\sim \calS_i}\left[\chi_\bS(a)\ne \chi_\bS(b)\right]\le \alpha_i \|a-b\|_1.
$$
\end{claim}
\begin{proof}
Let $D$ be the set of indices $i\in [d]$ with $a_i\ne b_i$.
The second part follows from $$\Pr_{\bS\sim \calS_i}\big[\chi_\bS(a)\ne \chi_\bS(b)\big]
\le \Pr_{\bS\sim\calS_i}\big[\bS\cap D\ne \emptyset\big]\le \alpha_i |D|=\alpha_i\|a-b\|_1.$$
The first part follows from
$$\Pr_{\bS\sim \calS_i}\big[\chi_\bS(a)\ne \chi_\bS(b)\big]
\ge \Pr_{\bS\sim\calS_i}\big[|\bS\cap D|=1\big]
=|D| \cdot \alpha_i\cdot (1-\alpha_i)^{|D|}
\ge  \frac{\alpha_i|D|}{2}.$$
This finishes the proof of the claim.
\end{proof}

It follows that the first part of (\ref{quadtreeinequality}) holds for $99\%-o(1)$ fraction of quadtrees $T$.
\end{proof}

To prove Theorem \ref{thm:one-pass} it suffices to prove the following lemma:

\begin{lemma}\label{EMDlemma1}
Fix a quadtree $T$, $i\in [h]$, and $\eps\in (0,1)$.
Then there is a $O(1/\eps)\cdot \polylog (n,d)$-space linear sketch
  that, given any $(A,B)$ of size $n$ each, outputs a number $\bfeta_i$ such that
\begin{equation}\label{temp1}
\calI_i- \frac{\eps nd}{2^i} \le \bfeta_i\le 2\calI_i
 +\frac{\Delta_id}{2^i} 
\end{equation}
with probability at least $2/3$.
\end{lemma}
\begin{proofof}{Theorem \ref{thm:one-pass} assuming Lemma \ref{EMDlemma1}}
\hspace{-0.1cm}The linear sketch starts by sampling a quadtree $T$ from $\calT$ and we assume that
 $T$ satisfies (\ref{quadtreeinequality}), which happens with probability at least $0.9$.
For each $i\in [h]$, it repeats independently 
  the linear sketch given in Lemma \ref{EMDlemma1}
  for $O(\log d)$ times, and we use $\bfeta_i$ to denote the median of $O(\log d)$ numbers
  returned by the reporting procedure.
It follows from a Chernoff bound and a union bound that 
  with probability $1-o_d(1)$, $\bfeta_i$ satisfies (\ref{temp1}) for every $i\in [h]$.
The reporting procedure returns $\bfeta=\sum_{i\in [h]}\bfeta_i+\eps nd$.
(\ref{mainmainmain}) follows directly from (\ref{quadtreeinequality}). 
\end{proofof}

For Theorem \ref{thm:two-pass} the following simpler lemma suffices:

\begin{lemma}\label{EMDtwopasslemma1}
Fix a quadtree $T$ and  any $i\in [h]$.
There is a two-round $\polylog(n,d)$-space linear sketch
  that, given any $(A,B)$ of size $n$ each, outputs a number $\bfeta_i$ such that
\begin{equation}\label{temp222}
 \calI_i \le \bfeta_i\le 2\calI_i
 +\frac{\Delta_id}{2^i}  
\end{equation}
with probability at least $2/3$.
\end{lemma}
\begin{proofof}{of Theorem \ref{thm:two-pass} using Lemma \ref{EMDtwopasslemma1}}
The proof is exactly the same as that of Theorem \ref{thm:one-pass}, except that we don't need
  to add $\eps nd$ in the estimate and no longer need the additive error.
\end{proofof}


\subsection{Universe Reduction} 

Now we focus on the proof of Lemma \ref{EMDlemma1}.
Fix a quadtree $T$ and a depth $i\in [h]$. 
The first step is to perform a routine universe reduction.

Let $m=n^3$. Let $H $ denote the following fixed bipartite graph between 
  $U=[m]$ and $V=[m]\times [m]$: $(u,v)\in H $ if and only if $u=v_1$.
Similarly we write $\pi(v)$ to denote the $u\in U$ with $(u,v)\in H$.
A \emph{universe reduction} between depths $i-1$ and $i$ of $T$
  consists~of~two functions 
  $h_{i-1}:[V_{i-1}]\rightarrow [m]$ and $h_i:[V_i ]\rightarrow [m]$,
  which can be used to induce a map $\sv$ from $\{0,1\}^d$ to $V$: For each  point $a\in \{0,1\}^d$,    $\sv(a)=v=(v_1,v_2)$ if $h_{i-1}(\sv_{i-1}(a))=v_1$
  and $h_i(\sv_{i}(a))=v_2$.
We also use $\sv$ to map points to $U$ using $\su:\{0,1\}^d\rightarrow U$: For each $a\in \{0,1\}^d$, 
  $\su(a)=\sv(a)_1$.
 
Given  the input pair of points $(A,B)$, we similarly define $A_v,B_v$ and $C_v$ for each node $v\in V$ of $H$ using $\sv$ and define $A_u,B_u$ and $C_u$ for each node $u\in U$ using $\su$ (e.g., $C_u$ is the set of $a\in A\cup B$ such that 
  $\su(a)=u$).
Let $\Delta=\sum_{v\in V}||A_v|-|B_v||$.
Similarly, whenever $\Delta>0$, we define the distribution $\calV$ supported over $V$, where
  each $v\in V$ is sampled with probability $||A_v|-|B_v||/\Delta$.
We also define 
$$p_{u,v,S}:=\Pr_{\cc_u\sim C_u,\cc_v \sim C_v} \big[\chi_S(\cc_u)\ne \chi_S(\cc_v)\big]\in [0,1]$$
for each $v\in V$ and $u=\pi(v)$. 
Lemma \ref{EMDlemma1} follows directly from the following lemma, 
which gives a linear sketch for approximating the expectation of $p(\pi(\vv),\vv,S)$
  with $\vv\sim\calV$ for a given $S\subseteq [d]$.

\begin{lemma}\label{MainEMDlemma2}Fix  $\sv:\{0,1\}^d\rightarrow V$, $S\subseteq [d]$
  and  $\eps\in (0,1)$. 
Then there~is~a~$O(1/\eps)\cdot 
  \poly(\log n,\log d )$ space linear sketch that, on any input $(A,B)$ of size $n$, outputs a number ${\bfeta}\in \mathbb{R}$
  with the following property: Let $\tau=1/\log^3 n$. 
Whenever $\Delta\ge \eps n/\log^3 n$, we have $$
 {\bfeta} =\Ex_{\vv\sim \calV}\big[p_{\pi(\vv),\vv,S}\big]\pm \tau 
$$ with probability at least $2/3$.
\end{lemma}

\begin{figure}[t]
	\begin{framed}
		\begin{flushleft}		
			\noindent {\bf Linear Sketch:}
			   
				\begin{enumerate}
				\item Run the $\ell_1$-sketch from Theorem \ref{lem:cauchy}
		on $(|A_v|-|B_v|:v\in V)$ (with $\eps=\delta=0.01$). 
				\item
				Draw independently $O(\log^6 n)$ many $\bS_{1},\ldots$ from $\calS_i$.  
For each $\bS_{j}$, use $\sv$ to run the linear sketch in Lemma \ref{MainEMDlemma2} 
  (with $\tau=1/\log^3 n$) independently $O(\log n)$ times.
			\end{enumerate}
\noindent {\bf Reporting Procedure:}
\begin{enumerate}
\item Run the reporting procedure of the $\ell_1$-sketch of Theorem \ref{lem:cauchy}
  to obtain $\bDelta$.
  
\item For each of the $O(\log^6 n)$ many $j$,
  run the reporting procedure to obtain a number for each
    of the $O(\log n)$ independent runs and let $\bfeta_j$ denote
  their median. 
  
\item 
If $\bDelta<2\eps n/\log^3 n$, output $0$; otherwise
  output \begin{equation}\label{finaloutput}
  \bfeta=\frac{3\bDelta}{\alpha_i}\cdot \left(\avg_j\{\bfeta_j\}+\frac{1}{8\log^2 n}\right).\end{equation}
\end{enumerate}
		\end{flushleft}\vskip -0.14in
	\end{framed}\vspace{-0.2cm}\caption{Linear sketch for Lemma \ref{EMDlemma1}}\label{mainlinearsketch1}
\end{figure}

\begin{proofof}{Lemma \ref{EMDlemma1} assuming Lemma \ref{MainEMDlemma2}}
Given a quadtree $T$, $i\in [h]$ and $\eps\in (0,1)$,
we start by sampling uniformly at random $h_{i-1}$ and $h_i$ and let $\sv$ denote 
  the map they induce.
We may assume that there is no collision with respect to $(A,B)$:
  Every two nodes $u,u'\in V_{i-1}$ with nonempty $C_u,C_{u'}$ in $T$
    satisfy $h_{i-1}(u)\ne h_{i-1}(u')$ and 
    every $v,v'\in V_i$ with nonempty $C_v,C_{v'}$ 
    satisfy $h_i(v)\ne h_i(v')$,
    and note the this event is violated with  probability at most $2n^2/m$.
When this happens, $\Delta=\Delta_i$ and $\calI_i$ can be defined equivalently using
  $\sv$ and $(A,B)$ as follows: $\calI_i=0$ if $\Delta=0$; otherwise ($\Delta>0$), we have 
$$
\calI_i=\frac{2\Delta}{\alpha_i}\cdot \Ex_{\vv\sim \calV,\bS\sim \calS_i}\big[p_{\pi(\vv),\vv,\bS}\big].
$$
To estimate $\calI_i$,  we use the linear sketch  described~in Figure \ref{mainlinearsketch1}, 
  which uses
  $O(1/\eps)\cdot \poly(\log n,\log d)$ space.
For its correctness, we note that the following events happen with probability at least $0.9$:
\begin{flushleft}\begin{enumerate}
\item $\bDelta=(1\pm 0.01) \Delta$, which happens with probability at least $0.99$;
\item If $\Delta\ge \eps n/\log^3 n$, then for every set $\bS_j$, we have
$$
{\bfeta}_{ j}=\Ex_{\vv\sim \calV} 
  \big[p_{\pi(\vv),\vv,\bS_{j}}\big]\pm \frac{1}{\log^3n},
$$
which by a Chernoff bound and a union bound on $j$ happens with probability $1-o_n(1)$;
\item When $\Delta>0$, we have that the 
$$
\avg_{j} \Big[ \bE_{\vv\sim \calV} 
  \big[p_{\pi(\vv),\vv,\bS_{j}}\big]\Big]= 
 \bE_{\vv\sim \calV,\bS\sim \calS_i} \big[p_{\pi(\vv),\vv,\bS}\big]\pm \frac{1}{\log^3 n},$$  which happens with probability at least $0.99$.
 When the last two items hold, we have
 \begin{equation}\label{blublublu2}
 \avg_j \{\bfeta_j\}=\bE_{\vv\sim \calV_i,\bS\sim \calS_i} 
  \big[p_{\pi(\vv),\vv,\bS }\big]\pm \frac{1}{8\log^2n}.
 \end{equation}
\end{enumerate}\end{flushleft}
We finish the proof by showing that, whenever events above occur,
  $\bfeta$ in (\ref{finaloutput}) satisfies
  \begin{equation}\label{blublublu}
\calI_i - \frac{\eps nd}{2^i} \le \bfeta \le 2\calI_i
 +\frac{\Delta d}{2^i}.
\end{equation}
When $\bDelta<2\eps n/\log^3 n$ 
  the inequality above is trivial because $\bfeta=0$ and the LHS above is negative
  (using $\Delta\le 2\bDelta$).
So we focus on the case when $\bDelta\ge 2\eps n/\log^3 n$. In this case we have
  $\Delta\ge \eps n/\log^3 n$ and thus, (\ref{blublublu2}) holds.  
The lower bound of (\ref{blublublu}) is trivial given (\ref{blublublu2}) and $3\bDelta>\Delta$.
For the upper bound,
$$
\bfeta\le \frac{4\Delta}{\alpha_i}\cdot \left(\bE_{\vv\sim \calV_i,\bS\sim \calS_i} 
  \big[p_{\pi(\vv),\vv,\bS }\big]+ \frac{1}{4\log^2n}\right)
  \le \frac{4\Delta}{\alpha_i}\cdot \bE_{\vv\sim \calV_i,\bS\sim \calS_i} 
  \big[p_{\pi(\vv),\vv,\bS }\big]+ \frac{\Delta d}{2^i}=2\calI_i+\frac{\Delta d}{2^i}.
$$ 
This finishes the proof of the lemma.
\end{proofof}

For the two-round linear sketch, Lemma \ref{EMDtwopasslemma1} follows from the following lemma:

\begin{lemma}\label{MainEMDlemma22222}Fix  $\sv:\{0,1\}^d\rightarrow V$ and $S\subseteq [d]$. Let $\tau=1/\log^3 n$.
Then there~is~a~$\polylog(n)$-space two-round linear sketch that, on any input $(A,B)$ of size $n$, outputs a number ${\bfeta}\in \mathbb{R}$ such that
 $$
 {\bfeta} =\bE_{\vv\sim \calV}\big[p_{\pi(\vv),\vv,S}\big]\pm \tau 
$$ 
with probability at least $2/3$.
\end{lemma}
\begin{proofof}{Lemma \ref{EMDtwopasslemma1} assuming Lemma \ref{MainEMDlemma22222}}
The first part of sampling $h_{i-1}$ and $h_i$ is the same as the proof of Lemma \ref{EMDlemma1}.
Similar to the linear sketch of Lemma \ref{EMDlemma1}, we use the $\ell_1$-sketch of Theorem  \ref{lem:cauchy} to obtain an estimate $\bDelta$ of $\Delta$.
We also sample $O(\log^6 n)$ many $\bS_1,\ldots$ from $\calS_i$ and for each $\bS_j$,
  repeat the two-round linear sketch of Lemma \ref{MainEMDlemma22222} $O(\log n)$ times to obtain
  the average as an estimate $\bfeta_j$ of $\bE_{\vv\sim \calV} [p_{\pi(\vv),\vv,\bS_j} ]$.
The total space used is $\polylog(n)$ and similar to the analysis of Lemma \ref{EMDlemma1},
  we have that with probability at least $0.9$, 
  $\bDelta=(1\pm 0.01)\Delta$ and $$
  \avg_j\{\bfeta_j\}=\bE_{\bv\sim\calV,\bS\sim \calS_i}\left[p_{\pi(\bv),\bv,\bS}\right]\pm 
   \frac{1}{8\log^2 n} .$$
Finally, returning $$\frac{3\bDelta}{\alpha_i}\cdot \left(\avg_j\{\bfeta_j\}+\frac{1}{8\log^2 n}\right)$$ satisfies the condition of Lemma \ref{EMDtwopasslemma1}.  \end{proofof}

\subsection{Two-Round Linear Sketch of Lemma \ref{MainEMDlemma22222}}

We first give the easier two-round linear sketch of Lemma \ref{MainEMDlemma22222} as a warmup.
In the first round, we apply the $\ell_1$-sampling linear sketch of Theorem \ref{thm:ell-1-sa} on the vector indexed by
  $v\in V$, where the entry indexed by $v$ is $|A_v|-|B_v|$.
The space used by this linear sketch is $O(\log^2 n)$.
Let $\vv$ be the output of the linear sketch. Then either $\vv$ is the failure symbol $\perp$, which happens with
  probability at most $1/3$, or the distribution of $\vv\in V$ has distance $1/n$ from $\calV$ in total variation.

Assuming $\vv=v\in V$ from the first round and letting $u=\pi(v)$, the second round computes
$$
p_{u,v,S}=
\frac{|C_{u,S}|}{|C_u|}+\frac{|C_{v,S}|}{|C_v|}-\frac{|C_{u,S}|}{|C_u|}\cdot \frac{|C_{v,S}|}{|C_v|},
$$
where we write $C_{u,S}$ to denote the set of points $a\in C_u$ with $\chi_S(a)=1$; we define $C_{v,S}$ similarly.
This can be done trivially by asking for these four numbers directly, with $O(\log n)$ space.

In summary, this two-round linear sketch uses $O(\log^2 n)$-space and achieves the following:
It either fails, with probability no more than $1/3$, or returns $p_{\pi(\bv),\bv,S}$ with
  $\bv$ drawn from a distribution that is $(1/n)$-close to $\calV$ in total variation distance.
Lemma \ref{MainEMDlemma22222} follows by repeating the two-round linear sketch independently and taking the average
  at the end.

\subsection{Overview of the Linear Sketch of Lemma \ref{MainEMDlemma2}}

We start with an overview of the linear sketch of Lemma \ref{MainEMDlemma2}.
Fix $\sv:\{0,1\}^d\rightarrow V$, a set $S\subseteq [d]$, and parameters $\eps,\tau\in (0,1)$. Let $(A,B)$ be the input pair of $n$ points each and we assume without loss of generality that 
  $\Delta\ge \eps n/\log^3 n$
  (for the case when $\Delta<\eps n\log^3 n$ we only need to make sure that the reporting
  procedure returns a number and the same space upper bound applies, which will be trivial given
  the description of the linear sketch later). 
  
We start with some notation.
Let $Q$ be the $m^2$-dimensional nonnegative vector 
  indexed by $v\in V$ with $Q_v=||A_v|-|B_v||$ (so $\Delta=\sum_{v\in V} Q_v$). Let  $P$ be~the $m$-dimensional nonnegative vector indexed by $u\in U$
  with $P_u=\sum_{v:\pi(v)=u} Q_v$ (so $\Delta=\sum_{u\in U} P_u$ as well).
We use $\calU$ to denote the distribution over $U$ where each $u\in U$ is sampled
  with probability $P_u/\Delta$; recall that $\calV$ is the distribution where each
  $v\in V$ is sampled with probability $Q_v/\Delta$.
Ideally, $ p_{\pi(\vv),\vv,S}$ with $\vv\sim \calV$  
  can be sampled using the following distribution $\calD$ supported on edges of $H$:
$(\uu,\vv)\sim \calD$ is sampled by first drawing 
  $\uu\sim \calU$ and then drawing a $\vv\sim \calV$ conditioning on $\pi(\vv)=\uu$ (i.e. 
  each node  $v\in V$ with $\pi(v)=u$ is sampled with probability $Q_v/P_u$), and
  finally setting $\bp=p_{\bu,\bv,S}$.

Our linear sketch for Lemma \ref{MainEMDlemma2} starts 
  by sampling independently two sequences of numbers $\bt_u\sim \Exp(1)$ and $\bt_v\sim\Exp(1)$   for each node $u\in U$ and $v\in V$.
The presentation of our linear sketch and its analysis will proceed in the following 
  three steps:
\begin{flushleft}\begin{enumerate}
\item Given $(A,B)$ and a tuple of positive numbers $(t_u,t_v:u\in U,v\in V)$,
  we define an edge $(u^*,v^*)$ of $H$. 
%
 We prove in Lemma \ref{firststeplemma} that,   
 $(\bu^*,\bv^*)$ defined using $(\bt_u,\bt_v:u\in U,v\in V)$ 
 when $\bt_u,\bt_v\sim \Exp(1)$ independently
 has the same distribution as $\calD$ described above.
\item Next we define an event $\calE$ on a tuple of positive numbers $(t_u,t_v:u\in U,v\in V)$
  with respect to 
  $(A,B)$,
  and show that $(\bt_u,\bt_v:u\in U,v\in V)$ when $\bt_u,\bt_v\sim \Exp(1)$ independently
   satisfies the event $\calE$ with respect to $(A,B)$ with high probability (Lemma \ref{EMDproblemma}).
\item Finally, given $\sv$, $S$, a tuple of positive numbers $(t_u,t_v:u\in U,v\in V)$ and $\eps,\tau\in (0,1)$, we give a
  low-space linear sketch  
  (Lemma \ref{MainEMDlemma3}) that, given any input pair $(A,B)$ of size $n$ each, 
  returns a number $ \bfeta \in [0,1]$ and satisfies the following property:
Whenever $\Delta\ge \eps n/\log^3 n$ and $(t_u,t_v:u\in U,v\in V)$ satisfies $\calE$ with respect to $(A,B)$,
  we have $ {\bfeta}=p_{u^*,v^*,S}\pm \tau$  with probability at least $2/3$.
\end{enumerate}\end{flushleft}
Lemma \ref{MainEMDlemma2} follows (by repeating $O(1/\tau^2)$ many rounds: in
  each round we sample fresh $(\bt_u,\bt_v)$ and then run the linear sketch to 
  obtain a number, and output the average at the end);
  see the proof of Lemma \ref{MainEMDlemma2} at the end of this subsection.

For the first step, given a tuple of positive numbers $(t_u,t_v:u\in U,v\in V)$,
  we define $u^*\in U$ as~the node  $u\in U$ that maximizes $P_u/t_u$ (breaking ties 
  by taking the smallest such $u$).
Given $\Delta >0$,
  we must have $P_{u^*}>0$. 
Next we define $v^*\in V$ as the $v\in V$ among those with $\pi(v)=u$ that maximizes
  $Q_v/t_v$ (again breaking ties by taking the smallest such $v$).
Similarly we have $Q_{v^*}>0$.

The following lemma shows that $(\uu^*,\vv^*)$ obtained from $\bt_u,\bt_v\sim \Exp(1)$ is
  distributed as $\calD$.

\begin{figure}[t!]
	\begin{framed}
		\begin{flushleft}		
			Let $\gamma$ and $\beta$ be the following two parameters:
			\begin{equation}\label{choicesgamma}
			\gamma\eqdef \frac{\tau}{\log n}=\frac{1}{\log^4 n}\quad \text{and}\quad \beta\eqdef \left\lceil\frac{\log^5 n}{\eps\tau\gamma^3}\right\rceil=\frac{\poly(\log n)}{\eps}.
    		\end{equation}
			We say a tuple of positive numbers 
			$(t_u,t_v:u\in U,v\in V)$ satisfy event $\calE$ with respect to $(A,B)$
			if the following four events $\calE_1,\calE_2,\calE_3$ and $\calE_4$ hold: 
			 
	 \begin{itemize}
		 \item Event $\calE_1$: The numbers $(t_u:u\in U)$ satisfy 
		   $$
		   \sum_{u\in U} \frac{P_u}{t_u} \le \frac{4\log (n/\gamma)}{\gamma}\cdot 
		     \Delta.
		   $$
		 Moreover, $u^*$ satisfies 
		 $$
		 \frac{P_{u^*}}{t_{u^*}}\ge \gamma\Delta\quad \text{and}\quad
		 \frac{P_{u^*}}{t_{u^*}}\ge (1+\gamma)\cdot \max_{\substack{ u\ne u^*}} \left\{\frac{P_u}{t_u}\right\}.
		 $$
		 \item Event $\calE_2$: The numbers $(t_u,t_v:u\in U,v\in V)$ satisfy 
		   $$
		   \sum_{v\in V} \frac{Q_v}{t_{\pi(v)}t_v}\le \frac{4\log (n/\gamma)}{\gamma}
		   \cdot \sum_{u\in U} \frac{P_u}{t_u}.
		   $$
		  Moreover, $v^*$ satisfies  		    $$
		    \frac{Q_{v^*}}{t_{v^*}}\ge \gamma P_{u^*}\quad\text{and}\quad
		    \frac{Q_{v^*}}{t_{v^*}}\ge (1+\gamma)\cdot \max_{v:\pi(v)=u^*} \left\{\frac{Q_v}{t_v}\right\}.
		    $$
		 \item Event $\calE_3$:  The numbers $(t_u:u\in U)$ satisfy that
		 $$
		 \sum_{u\in U} \frac{|C_u|}{t_u}\le \frac{4\log (n/\gamma)}{\gamma}\cdot n
		 $$
		 and the
		 $\ell_2$-norm of $(|C_u|/t_u: u\in U)_{-\beta}$
		   is at most $12n/\sqrt{\beta}$.		   
		 \item Event $\calE_4$: 
		   The $\ell_2$-norm of $(|C_v|/(t_{\pi(v)}t_v):v\in V)_{-\beta}$ is at most
		   $$
		  \frac{12}{\sqrt{\beta}}\cdot \sum_{u\in U}\frac{|C_u|}{t_u}.
		   $$
		  
		 \end{itemize} \end{flushleft}
	\end{framed}\vspace{-0.2cm}\caption{Event $\calE$ on $(t_u,t_v:u\in U,v\in V)$
	with respect to $(A,B)$}\label{mainlinearsketch22}
\end{figure}

\begin{lemma}\label{firststeplemma}
Let $(A,B)$ be a pair of $n$ points. 
With $\bt_u\sim \Exp(1)$ and $\bt_v\sim \Exp(1)$ independently for each $u\in U$ and $v\in V$,
  we have that the distribution of $(\bu^*,\bv^*)$ is the same as $\calD$.
\end{lemma}
\begin{proof}
This follows from Fact \ref{fact:scale} and Fact \ref{fact:order}.
\end{proof}

Next we describe the event $\calE$ in Figure \ref{mainlinearsketch22}, where we also introduce two
  new parameters $\gamma$ and $\beta$.

\begin{lemma}\label{EMDproblemma}
 Let $(A,B)$ be a pair of $n$ points.
With $\bt_u\sim \Exp(1)$ and $\bt_v\sim \Exp(1)$ independently~for each $u\in U$ and $v\in V$,
  $\calE$ holds with respect to $(A,B)$  with probability at least $1-O(\gamma)-e^{-\Omega(\beta)}$.
\end{lemma}
\begin{proof} The first part of event $\calE_1$ follows from Lemma~\ref{lem:sum_exponentials}, and the second part follows from Lemma~\ref{lem:gap_exponentials}. Similarly, the first part of event $\calE_2$ follows from Lemma~\ref{lem:sum_exponentials} on the randomness of $\bt_v$. 
The first part of $\calE_3$ follows again by Lemma~\ref{lem:sum_exponentials}, and the second part of event $\calE_3$ uses Lemma \ref{lem:sizebound} (with $c=1$).
For event $\calE_4$, one can first uses Lemma \ref{lem:sizebound} (with $c=1$) and randomness of $\bt_v$
  to upperbound the $\ell_2$-norm of $(|C_v|/(\bt_{\pi(v)}\bt_v:v\in V)_{-\beta}$ by
$$
\frac{12}{\sqrt{\beta}}\cdot \sum_{v\in V} \frac{|C_v|}{\bt_{\pi(v)}}
=\frac{12}{\sqrt{\beta}}\cdot \sum_{u\in U} \frac{|C_u|}{\bt_u}. 
$$
Applying a union bound  finishes the proof of the lemma.
\end{proof}

Finally we state the performance guarantee of the main linear sketch:
  
\begin{lemma}\label{MainEMDlemma3}
Fix $\sv:\{0,1\}^d\rightarrow V$, $S\subseteq [d]$, a tuple of positive numbers $(t_u,t_v:u\in U,v\in V)$~and~two parameters $\eps,\tau\in (0,1)$. 
There is a $O(1/\eps)\cdot 
  \poly(\log n,\log d, 1/\tau )$-space linear sketch that, on input $(A,B)$ of size $n$
  each, outputs a number ${\bfeta}\in [0,1]$
  with the following property:
When $\Delta\ge \eps n/\log^3 n$ and 
  $(t_u,t_v)$ satisfies  $\calE$ with respect to $(A,B)$, we have 
  $ {\bfeta}=p_{u^*,v^*,S}\pm \tau$ with probability at least $2/3$.
\end{lemma}

\begin{proofof}{Lemma \ref{MainEMDlemma2} assuming Lemma \ref{MainEMDlemma3}}
The linear sketch repeats the following $O(1/\tau^2)$ rounds:
\begin{flushleft}\begin{itemize}
\item Draw fresh $\bt_u,\bt_v\sim \Exp(1)$ and repeat the linear sketch of
  Lemma \ref{MainEMDlemma3} independently \\ (setting $\tau$ in Lemma \ref{MainEMDlemma3}
    to be a quarter of the $\tau$ in Lemma \ref{MainEMDlemma2}) for $O(\log (1/\tau))$ times.
\end{itemize}\end{flushleft}
For each round of the linear sketch, the reporting procedure uses the reporting procedure of Lemma \ref{MainEMDlemma3} to get $O(\log (1/\tau))$ numbers in $[0,1]$ and compute their median (in $[0,1]$).
It finally outputs the average of these $O(1/\tau^2)$ many medians.
  
To prove the correctness, it suffices to show that the expectation of $\bfeta$ (the median)
  reported from each round 
  is within $\bE_{\vv\sim \calV}[p_{\pi(\vv),\vv,S}]\pm (\tau/2)$. 
To see this is the case, we compare the expectation of $\bfeta$ with the ideal process of sampling first $(\bu^*,\bv^*)\sim \calD$
  and returning $p_{\bu^*,\bv^*,S}$, whose expectation is exactly 
   $\Ex_{\vv\sim \calV}[p_{\pi(\vv),\vv,S}]$.
First, the linear sketch may fail because $(\bt_u,\bt_v)$ sampled in this round does not satisfy $\calE$,
  which happens with probability $o(\tau)$ and may shift the expectation by no more 
  than $o(\tau)$.
Second, the median may fail to fall inside the $\pm \tau/4$ interval centered at $p_{u^*,v^*,S}$,
  which happens with probability $\tau/8$ (because we repeat $O(\log (1/\tau))$ times in each round), which can shift the expectation by no more than $\tau/8$.
Finally, even if the median is inside the interval $p_{u^*,v^*,S}\pm (\tau/4)$,~the~$\pm \tau/4$ may
  shift the expectation by $\tau/4$.
Summing up these three cases, we conclude that the expectation can be shifted by no more than $\tau/2$.
\end{proofof} 

\subsection{Three Linear Sketches and their Performance Gurantees}

Our linear sketch for Lemma \ref{MainEMDlemma3} consists of three
  linear sketches
  whose goals are to  recover node
   $u^*$, recover node $v^*$, and estimate $p_{u^*,v^*,S}$, respectively.
For each $i\in [3]$, we write $\LS_i$ to denote~the encoding algorithm of the $i$th linear sketch 
  which uses its own randomness $\aa_i$ to generate the linear sketch
  $\LS_i((A,B),\aa_i)$.
Each $\LS_i$ is paired with a decoding algorithm $\ALG_i$ (note that  
  $\LS_i$ and $\ALG_i$ are deterministic) which has
  the following input and output:
\begin{enumerate}
\item $\ALG_1$ takes $\aa_1$ and $\LS_1((A,B),\aa_1)$ as input and outputs a vertex in $U$;
\item $\ALG_2$ takes a node in $U$, $\aa_2$ and $\LS_2((A,B),\aa_2)$ as input and outputs a vertex in $V$;
\item $\ALG_3$ takes an edge in $H$, $\aa_3$ and $\LS_3((A,B),\aa_3)$ as input and outputs a number in $[0,1]$.
\end{enumerate}
We now state their performance guarantees  (where we fix 
  $\sv,S,(t_u,t_v),\eps$ and $\tau$ as in Lemma \ref{MainEMDlemma3}).

\begin{lemma}\label{FinalEMDlemma1} 
There is a $\poly(\log n,1/\gamma)$-space linear sketch $(\LS_1,\ALG_1)$ which outputs a node in~$ U$
  and satisfies the following property:
Whenever $\Delta>0$ and $(t_u,t_v)$ satisfies $\calE_1$ with respect~to $(A,B)$,  $\ALG_1$
  returns $u^*$ with probability at least $1-1/\poly(n)$.
\end{lemma}

\begin{lemma}\label{FinalEMDlemma2}
There is a $\poly(\log n,1/\gamma)$-space linear sketch $(\LS_2,\ALG_2)$ which outputs a node in $V$
  and satisfies the following properties: (1) If $u$ is node given to $\ALG_2$, then the latter
  always returns a node $\vv$ with $\pi(\vv)=u$; (2)
Whenever $\Delta>0$ and $(t_u,t_v)$ satisfies $\calE_1$ and $\calE_2$ with respect to $(A,B)$,  $\ALG_2$
  running on $u^*$, $\aa_2$ and $\LS_2((A,B),\aa_2)$ 
 returns $v^*$ with probability at least $1-\poly(n)$.
\end{lemma}

\begin{lemma}\label{FinalEMDlemma3}
There is a $O( \beta\log n)$-space linear sketch $(\LS_3,\ALG_3)$ which always outputs a number in $[0,1]$
  and satisfies the following property:
Whenever $\Delta\ge \eps n/\log^3n$ and $(t_u,t_v)$ satisfies  $\calE_3$ and $\calE_4$ 
  with respect to $(A,B)$,  $\ALG_3$
  running on $(u^*,v^*)$, $\aa_3$ and $\LS_3((A,B),\aa_3)$ 
  returns a number $\bfeta $ that satisfies 
$ {\bfeta}=p_{u^*,v^*,S}\pm \tau$ with probability at least $1-1/\poly(n)$.
\end{lemma}

\begin{proofof}{Lemma \ref{MainEMDlemma3} assuming Lemmas \ref{FinalEMDlemma1},
  \ref{FinalEMDlemma2} and \ref{FinalEMDlemma3}}
The main linear sketch of Lemma \ref{MainEMDlemma3} flips coins $\aa_1,\aa_2,\aa_3$
  and then encodes $(A,B)$ as $(\LS_i((A,B),\aa_i):i\in [3])$.
The reporting procedure runs $\ALG_1$ on $\aa_1$ and $\LS_1((A,B),\aa_1)$ to obtain
  a vertex $\uu\in U$,
  runs $\ALG_2$ on $\uu$, $\aa_2$ and $\LS_2((A,B),\aa_2)$ to obtain
  a vertex $\vv\in V$,
  and runs $\ALG_3$ on $(\uu,\vv)$, $\aa_3$ and $\LS_3((A,B),\aa_3)$ to obtain
  $ {\bfeta}$.
By a union bound on errors of the three lemmas,
  we have $ {\bfeta}=p_{u^*,v^*,S}\pm \tau$ with probability at least $2/3$.
\end{proofof}

\subsection{Linear Sketch of Lemma \ref{FinalEMDlemma1}} 

\begin{figure}[t]
	\begin{framed}
		\begin{flushleft}		
			\textbf{Linear Sketch $\LS_1$:} We repeat the following procedure 
			  for $j=1,\ldots,O(\log n/\gamma^2)$ times:
			\begin{enumerate}
			\item For each $v\in V$, draw an independent Cauchy random variable $\balpha_v$.
			\item Define the following vector $\bS^{(j)}$ indexed by $u\in U$: For each $u\in U$, let $$\bS^{(j)}_u:=\sum_{v:\pi(v)=u} \balpha_v \cdot \frac{|A_v|-|B_v|}{t_u}
.$$
			\item Run the Count-Sketch of Theorem \ref{thm:countsketch} on $\bS^{(j)}$ (with  
			  $\eps=\eta$ and $n=m$).
			\end{enumerate}
			
			\textbf{Reporting Procedure $\ALG_1$:} For each $j$, use the reporting procedure
			  of the Count-Sketch of Theorem \ref{thm:countsketch} to compute a vector $\widehat{\bS}^{(j)}$.
			Let $\widehat{\bS}$ be the vector with 
			$$\widehat{\bS}_u:=\median\left\{\widehat{\bS}^{(j)}_u:j=1,\ldots,O(\log n/\gamma^2)\right\},\quad\text{for each $u\in U$.}$$
			Output $\argmax_{u\in U} \{\widehat{\bS}_u\}$.
		\end{flushleft}
	\end{framed}\vspace{-0.2cm}\caption{The linear sketch $(\LS_1,\ALG_1)$ of Lemma \ref{FinalEMDlemma1}}\label{LSEMDfigure1}
\end{figure}

Let $\eta=O(\gamma^2/\log n)=1/\poly(\log n)$.
We describe  $(\LS_1,\ALG_1)$ in Figure \ref{LSEMDfigure1}, which uses space
  $$O\left(\frac{\log n}{\gamma^2}\right)\cdot O\left(\frac{\log n}{\eta^2}\right)
  =\poly(\log n).$$
We start by showing that $\widehat{\bS}^{(j)}$ is close to $\bS^{(j)}$ in $\ell_\infty$ distance
  for each $j$.
For this purpose, it suffices to upperbound $\smash{\|\bS^{(j)}_{-1/\eta^2}\|_2}$.
Note that each entry of $\bS^{(j)}$ is a sum with Cauchy random variables
  as coefficients.
Using the $1$-stability of the Cauchy distribution, each entry 
  $\smash{\bS^{(j)}_u}$  can be equivalently generated by drawing a
  Cauchy variable $\balpha_u$ and setting
$$
\bS^{(j)}_u= \balpha_u\cdot \sum_{v:\pi(v)=u} \frac{Q_v}{t_u}=\balpha_u\cdot \frac{
P_u}{t_u}.
$$
Applying Lemma \ref{lem:sizebound} and using the fact that $\Pr[\balpha\ge t]\le 1/t$,
  we have from $\calE_1$ that $$\left\|\bS^{(j)}_{-1/\eta^2} \right\|_2\le O(\eta)\cdot  \sum_{u\in U} \frac{P_u}{t_u}\le O\left(\frac{\eta\log n}{\gamma}\right)\cdot \Delta$$
with probability $1-1/\poly(n)$.
Assuming this holds for every $j$, it follows from Theorem \ref{thm:countsketch} that 
$$
\left\|\widehat{\bS}^{(j)}-\bS^{(j)}\right\|\le O\left(\frac{\eta^2 \log n}{\gamma}\right)\cdot \Delta 
$$
with probability $1-1/\poly(n)$.
As a result (and given that we only apply a union bound on $O(\log n)$ many events), 
  with probability at least $1-1/\poly(n)$, we have for every $u\in U$:
$$
\widehat{\bS}_u=\median \left\{\bS_u^{(j)}:j=1,\ldots,O(\log n/\gamma^2)\right\}\pm 
O\left(\frac{\eta^2 \log  n}{\gamma}\right)\cdot \Delta.
$$
On the other hand, for each $u\in U$, we 
  have from Theorem \ref{lem:cauchy} that 
$$
\median \left\{\bS_u^{(j)}:j=1,\ldots,O(\log n/\gamma^2)\right\}=\median(|\calC|)\cdot 
  \left(1\pm \frac{\gamma}{4}\right)\frac{P_u}{t_u}
$$
with probability at least $1-\poly(n)$.
Assuming all these events occur and using $\calE_1$, every $u\ne u^*$ satisfies
  (letting $c_0$ be the constant $\median(|\calC|)$ below)
\begin{align*}
\widehat{\bS}_u&\le c_0\cdot \left(1+\frac{\gamma}{4}\right)\cdot \frac{1}{1+\gamma}\cdot \frac{P_{u^*}}{t_{u^*}}+O\left(
\frac{\eta^2\log n}{\gamma}\right)\cdot \Delta\le 
  c_0 \cdot \left(1-\frac{ \gamma}{2}\right)\cdot \frac{P_{u^*}}{t_{u^*}}
  +O\left(
\frac{\eta^2\log n}{\gamma}\right)\cdot \Delta.\end{align*}
and $u^*$ satisfies
$$
\widehat{\bS}_{u^*} \ge c_0\cdot \left(1-\frac{\gamma}{4}\right)\frac{P_{u^*}}{t_{u^*}}-O\left(
\frac{\eta^2\log n}{\gamma}\right)\cdot \Delta.
$$
We have 
$$
\widehat{\bS}_{u^*}-\widehat{\bS}_{u}\ge c_0\cdot \frac{\gamma}{4}\cdot 
  \frac{P_{u^*}}{t_{u^*}}-O\left(
\frac{\eta^2\log  n}{\gamma}\right)\cdot \Delta>0
$$
using $\calE_1$ and our choice of $\eta=O(\gamma^2/\log n)$.

\subsection{Linear Sketch of Lemma \ref{FinalEMDlemma2}} 

\begin{figure}[t!]
	\begin{framed}
		\begin{flushleft}		
			\textbf{Linear Sketch $\LS_2$:} Let $S$ be the vector with
			$$
			S_v:=\frac{|A_v|-|B_v|}{t_{\pi(u)}t_v}.
			$$
			Run the Count-Sketch of Theorem \ref{thm:countsketch} on $S$ with $\eps=\eta$ and $n=m^2$.\medskip
			
			\textbf{Reporting Procedure $\ALG_2$ (given $u\in U$):} Use the reporting procedure
			  of the Count-Sketch of Theorem \ref{thm:countsketch} to compute a vector $\widehat{\bS}$, and output $\argmax_{v:\pi(v)=u} \{|\widehat{\bS}_v|\}.$
		\end{flushleft}
	\end{framed}\vspace{-0.2cm}\caption{The linear sketch $(\LS_2,\ALG_2)$ of Lemma \ref{FinalEMDlemma2}}\label{LSEMDfigure2}
\end{figure}

Let $\eta=O(\gamma^5/\log^2 n)$.  The pair $(\LS_2,\ALG_2)$ is given in Figure \ref{LSEMDfigure2}.
Assume both $\calE_1$ and $\calE_2$ hold, and assume that the node 
  given to $\ALG_2$ is $u^*$. Note that by $\calE_1$ and $\calE_2$ we have
$$
\|S\|_1=\sum_{v\in V} \frac{Q_v}{t_{\pi(v)}t_v}
\le O\left(\frac{\log n}{\gamma}\right)\cdot \sum_{u\in U}\frac{P_u}{t_u}
\le O\left(\frac{\log^2 n}{\gamma^2}\right)\cdot \Delta.
$$
As a result, we have $$\|\widehat{\bS}-S\|_\infty\le O\left(\frac{\eta\log^2 n}{\gamma^2}\right)\cdot \Delta. $$
Every $v$ with $\pi(v)=u^*$ and $v\ne v^*$ satisfies
\begin{align*}
|\widehat{\bS}_v|\le \frac{Q_v}{t_{u^*}t_v}+ O\left(\frac{\eta\log^2 n}{\gamma^2}\right)\cdot \Delta\quad\text{and}\quad 
|\widehat{\bS}_{v^*}|\ge\frac{Q_{v^*}}{t_{u^*}t_{v^*}}- O\left(\frac{\eta\log^2 n}{\gamma^2}\right)\cdot \Delta.
\end{align*} 
As a result, using $\calE_1$ and $\calE_2$ we have 
$$
|\widehat{\bS}_{v^*}|-|\widehat{\bS}_v|\ge 
  \frac{1}{t_{u^*}}\cdot \left(1-\frac{1}{1+\gamma}\right)\cdot \frac{Q_{v^*}}{t_{v^*}}
  -O\left(\frac{\eta\log^2 n}{\gamma^2}\right)\cdot \Delta
  \ge \Omega(\gamma^3)\cdot \Delta - O\left(\frac{\eta\log^2 n}{\gamma^2}\right)\cdot \Delta>0
$$
given our choice of $\eta=O(\gamma^5/\log^2 n)$.

\subsection{Linear Sketch of Lemma \ref{FinalEMDlemma3}}

\begin{figure}[t!]
	\begin{framed}
		\begin{flushleft}		
			\textbf{Linear Sketch $\LS_3$:} Let $S^{(1)}$ and $S^{(2)}$ the two vectors indexed
			  by $u\in U$ with
			$$
			S_u^{(1)}:=\frac{|C_u|}{t_u}\quad\text{and}\quad S_u^{(2)}=\frac{|C_{u,S}|}{t_u},\quad\text{for each $u\in U$}.
			$$
			Let $S^{(3)}$ and $S^{(4)}$ be the two vectors indexed by $v\in V$ with
			$$
			S_u^{(3)}:=\frac{|C_v|}{t_{\pi(v)}t_v}\quad\text{and}\quad S_u^{(4)}=\frac{|C_{v,S}|}{t_{\pi(v)}t_v},\quad\text{for each $v\in V$}.
			$$
			Run the Count-Sketch of Theorem \ref{thm:countsketch} on each $S^{(i)}$ (with $\eps=1/\sqrt{\beta}$ and $n=m$ or $m^2$).\medskip
			
			\textbf{Reporting Procedure $\ALG_3$ (given an edge $(u,v)$ of $H$):} Use the reporting procedure\\
			  of the Count-Sketch of Theorem \ref{thm:countsketch} to compute vectors $\widehat{\bS}^{(i)}$, $i\in [4]$.
			Output 
			\begin{equation}\label{alala}
			\frac{\widehat{\bS}^{(2)}_{u}}{\widehat{\bS}^{(1)}_{u}}
			+\frac{\widehat{\bS}^{(4)}_{v}}{\widehat{\bS}^{(3)}_{v}}
			-\frac{\widehat{\bS}^{(2)}_{u}}{\widehat{\bS}^{(1)}_{u}}\cdot 
			\frac{\widehat{\bS}^{(4)}_{v}}{\widehat{\bS}^{(3)}_{v}}
			\end{equation}
			after truncating it to $[0,1]$.
		\end{flushleft}
	\end{framed}\vspace{-0.2cm}\caption{The linear sketch $(\LS_3,\ALG_3)$ of Lemma \ref{FinalEMDlemma3}}\label{LSEMDfigure3}
\end{figure}

Given $S$ and $u\in U$, recall that $C_{u,S}$ denotes the set of points $a\in A\cup B$ with $\chi_S(a)=1$, and $C_{v,S}$ is defined similarly.
Then $p_{u,v,S}$ can be expressed as
$$
\frac{|C_{u,S}|}{|C_u|}+\frac{|C_{v,S}|}{|C_v|}-\frac{|C_{u,S}|}{|C_u|}\cdot \frac{|C_{v,S}|}{|C_v|}.
$$
We present the pair $(\LS_3,\ALG_3)$ in Figure \ref{LSEMDfigure3},
  which approximates the four numbers $|C_{u,S}|,|C_u|$~and $|C_{v,S}|,|C_v|$ 
  and uses them to obtain an estimate of $p_{u,v,S}$.
For correctness, it suffices to upperbound
$$
\left|\frac{\widehat{\bS}^{(2)}_{u^*}}{\widehat{\bS}^{(1)}_{u^*}}-
\frac{|C_{u^*,S}|}{|C_{u^*}|}\right|\quad\text{and}\quad
\left|\frac{\widehat{\bS}^{(4)}_{v^*}}{\widehat{\bS}^{(3)}_{v^*}}-
\frac{|C_{v^*,S}|}{|C_{v^*}|}\right|.
$$
We start with the first one. 
Note that by $\calE_3$ we have
$$
\|S^{(2)}_{-\beta}\|_2 \le \|S^{(1)}_{-\beta}\|_2 = O\left(\frac{n}{\sqrt{\beta}}\right).
$$
By Theorem \ref{thm:countsketch}, with probability at least $1-\poly(n)$, we have
$$
\left\|\bS^{(2)}-S^{(2)}\right\|_\infty, \left\|\bS^{(1)}-S^{(1)}\right\|_\infty\le O\left(\frac{  n }{\beta}\right) .
$$
Focusing on $\widehat{\bS}^{(1)}_{u^*}$ first, we claim that
$$
\left|\widehat{\bS}^{(1)}_{u^*}-S^{(1)}_{u^*}\right|\le \frac{\tau}{16}\cdot  S^{(1)}_{u^*}
$$
This follows from our choice of $\beta=\log^5 n/(\tau\eps\gamma^3)$ and that
$$
S^{(1)}_{u^*}\ge  \frac{P_{u^*}}{t_{u^*}}\ge  \gamma \Delta
  \ge \frac{\eps\gamma}{\log^3 n}\cdot n.
$$
Therefore, we have 
$$
\frac{\bS^{(2)}_{u^*}}{\bS^{(1)}_{u^*}}=\frac{S^{(2)}_{u^*}}{\bS^{(1)}_{u^*}}
\pm O\left(\frac{n}{\beta}\cdot \frac{1}{S^{(1)}_{u^*}}\right) 
=\left(1\pm \frac{\tau}{8}\right)\cdot \frac{ S^{(2)}_{u^*}}{ S^{(1)}_{u^*}}
\pm O\left(\frac{n}{\beta}\cdot \frac{\log^3 n}{\eps \gamma n}\right)
= \frac{ S^{(2)}_{u^*}}{ S^{(1)}_{u^*}}\pm \frac{\tau}{6}.
$$
The analysis for the other ratio is similar. First we have from $\calE_3$ and $\calE_4$ that
$$
\|S^{(4)}_{-\beta}\|_2\le \|S^{(3)}_{-\beta}\|_2
\le O\left(\frac{n\log n}{\sqrt{\beta}\gamma}\right).
$$
By Theorem \ref{thm:countsketch}, with probability at least $1-\poly(n)$, we have
$$
\left\|\bS^{(3)}-S^{(3)}\right\|_\infty, \left\|\bS^{(4)}-S^{(4)}\right\|_\infty\le O\left(\frac{  n\log n}{\beta \gamma}\right) .
$$
Focusing on $\widehat{\bS}^{(3)}_{v^*}$ first, we have similarly that
$$
\left|\widehat{\bS}^{(3)}_{v^*}-S^{(3)}_{v^*}\right|\le \frac{\tau}{16}\cdot  S^{(3)}_{u^*}
$$
using our choice of $\beta=\log^5 n/(\tau\eps\gamma^3)$ and that 
$$
S^{(3)}_{v^*}=\frac{ Q_{v^*} }{t_{u^*}t_{v^*}}\ge \gamma^2\Delta\ge \frac{\eps \gamma^2}{\log^3 n}\cdot n  .
$$
Therefore, we have
$$
\frac{\bS^{(4)}_{v^*}}{\bS^{(3)}_{v^*}}=\frac{S^{(4)}_{v^*}}{\bS^{(3)}_{v^*}}
\pm O\left(\frac{n\log n}{\beta\gamma}\cdot \frac{1}{S^{(3)}_{v^*}}\right) 
=\left(1\pm \frac{\tau}{8}\right)\cdot \frac{ S^{(4)}_{v^*}}{ S^{(3)}_{v^*}}
\pm O\left(\frac{n\log n}{\beta\gamma}\cdot \frac{\log^3 n}{\eps \gamma^2 n}\right)
= \frac{ S^{(4)}_{v^*}}{ S^{(3)}_{v^*}}\pm \frac{\tau}{6}.
$$
As a result we have (\ref{alala}) is
$$
\left(\frac{ S^{(2)}_{u^*}}{ S^{(1)}_{u^*}}\pm \frac{\tau}{6}\right)+
\left(\frac{ S^{(4)}_{v^*}}{ S^{(3)}_{v^*}}\pm \frac{\tau}{6}\right)
-\left(\frac{ S^{(2)}_{u^*}}{ S^{(1)}_{u^*}}\pm \frac{\tau}{6}\right)\cdot 
\left(\frac{ S^{(4)}_{v^*}}{ S^{(3)}_{v^*}}\pm \frac{\tau}{6}\right)
=\frac{ S^{(2)}_{u^*}}{ S^{(1)}_{u^*}}+\frac{ S^{(4)}_{v^*}}{ S^{(3)}_{v^*}}-
\frac{ S^{(2)}_{u^*}}{ S^{(1)}_{u^*}}\cdot \frac{ S^{(4)}_{v^*}}{ S^{(3)}_{v^*}}\pm {\tau}.
$$

%% file: sketch-plan.tex

\def\MST{\mathsf{MST}}

\section{Linear Sketches for MST}\label{sec:MSTSketch}

\newcommand{\sfv}{\mathsf{v}}

Fix $n, d \in \N$, and let $X = \{ x_1, \dots, x_n \} \subset \{0,1\}^d$. To settle on some notation, we let $h \eqdef \lceil \log_2 d \rceil$, and we consider a quadtree $T$ of depth $h+1$, whose levels are indexed by $i \in \{0, \dots, h\}$, starting at the root at depth $i=0$, to the leaves at depth $h$. 
For $i \in \{0, \dots, h\}$, we let $V_i$ be the set of nodes of $T$ at depth $i$. For each $i \in \{0,\dots, h\}$, $T$ induces a map $\sfv_i \colon \{0,1\}^d \to V_i$, such that $\sfv_0(a), \sfv_1(a), \dots, \sfv_h(a)$ is the path of a point $a \in \{0,1\}^d$ down the quadtree $T$. We say $a$ maps to $v \in V_i$ if $\sfv_i(a) = v$. 

For a node $v \in V_i$ for $i > 0$, we let $\pi(v)$ be the parent of $v$ in $T$ (which lies at depth $i-1$). The subset $L_i \subset V_i$ is the set of nodes $v \in V_i$ which are \emph{non-empty with respect to $X$}, i.e., some point $x \in X$ mapped to node $v$; for $v \in L_i$, the set $X_v \subset X$ consists of the set of points of $X$ which map to $v$. For a node $u \in V_{i-1}$, we let $C(u) \subset L_i$ be the set of child nodes of $u$. Our starting point is the fact that with probability at least $0.99$ over the draw of $\bT \sim \calT$, two things occur: the first is that within any node $v \in L_i$, and any two $x, y \in X_{v}$, we have $\|x - y \|_1 \leq O(d\log n / 2^i)$. The second is that we can approximate the cost of the minimum spanning tree by the following sampling experiment. 
\begin{enumerate}
\item For each $u \in L_{i-1}$, we sample a \emph{parent representative} $\bc_u$ from $X_u$ by sampling $\bv' \sim C(u)$, and then sampling $\bc_u \sim X_{\bv'}$.
\item For each $v \in L_i$, we sample a \emph{child representative} $\boldr_v\sim X_v$. 
\end{enumerate}
Then, we have that for every choice of $\bc_u$ and $\boldr_v$, we have
\begin{align}
\MST(X) \leq \sum_{i=1}^h \ind\{ |L_i| > 1\} \sum_{v \in L_i} \| \boldr_v - \bc_{\pi(v)}\|_1 = \sum_{i=1}^h \ind\{ |L_i| > 1 \} \cdot |L_i| \Ex_{\bv \sim L_i} \left[ \| \boldr_{\bv} - \bc_{\pi(\bv)}\|_1 \right] \label{eq:mst-lb-est}
\end{align}
In order to see why the first inequality in (\ref{eq:mst-lb-est}) holds, consider the following way to recursively build a spanning tree from a quadtree $T$: inductively assume that each node $v$ contains a spanning tree of $X_v$, then if a node $v$ is a child of $u$, place an edge between $\boldr_{v}$ and $\bc_{u}$, thereby connecting all spanning trees of children of $u$. On the other hand, the expectation of the above quantity (over the random choices of $\bc_u$ and $\boldr_v$) upper bounds $\MST(X)$, up to a factor of $\tilde{O}(\log n)$. Specifically, by triangle inequality, as well as the distribution of $\bc_u$, we have
\begin{align*}
&\sum_{i=1}^h \ind\{ |L_i| > 1\} \cdot |L_i| \Ex_{\bv \sim L_i}\left[ \Ex_{\boldr_{\bv}, \bc_{\pi(\bv)}}\left[ \| \boldr_{\bv} - \bc_{\pi(\bv)} \|_1 \right] \right] \\
&\qquad \qquad \leq \sum_{i=1}^h \ind\{ |L_i| > 1\} |L_i| \Ex_{\bv \sim L_i} \left[ \Ex_{\boldr_{\pi(\bv)} \sim X_{\pi(\bv)}} \left[ \Ex_{\boldr_{\bv}, \bc_{\pi(\bv)}}\left[ \| \boldr_{\bv} - \boldr_{\pi(\bv)}\|_1 + \| \boldr_{\pi(\bv)} - \bc_{\pi(\bv)}\right]\right]\right] \\
&\qquad \qquad =  2\sum_{i=1}^h \ind\{ |L_i| > 1 \} |L_i| \Ex_{\bv \sim L_i}\left[ \Ex_{\substack{\boldr_{\bv} \sim X_v \\ \boldr_{\pi(\bv)} \sim X_{\pi(\bv)}}}\left[ \| \boldr_{\bv} - \boldr_{\pi(\bv)}\|_1\right]\right].
\end{align*}
Furthermore, we know that this latter quantity satisfies
\[ \sum_{i=1}^h \ind\{ |L_i| > 1 \}\cdot |L_i| \Ex_{\bv \sim L_i} \left[ \Ex_{\substack{\boldr_{\bv} \sim X_{\bv} \\ \boldr_{\pi(\bv)} \sim X_{\pi(\bv)}}}\left[ \| \boldr_{\bv} - \boldr_{\pi(\bv)} \|_1\right] + \frac{d}{2^i} \right]\leq \tilde{O}(\log n) \cdot \MST(X). \]
Our goal is to design a linear sketch which uses $\poly(\log d, \log n)$ space, and for every $i \in \{1, \dots, h \}$, outputs two values $\hat{\ell}_i, \hat{\mu}_i \in \R$, satisfying
\begin{align}
| L_i | &\leq \hat{\ell}_i \leq 1.5 \cdot |L_i | \qquad\text{and}\qquad
\hat{\mu}_i \leq \Ex_{\substack{\bv \sim L_i \\ \boldr_{\bv}, \bc_{\pi(\bv)}}}\left[ \| \boldr_{\bv} - \bc_{\pi(\bv)} \|_1 \right] \leq \hat{\mu}_i + \frac{d}{2 \cdot 2^i}.  \label{eq:values}
\end{align}
If, for every $i \in \{1, \dots, h\}$, the estimates $\hat{\ell}_i$ and $\hat{\mu}_i$ all satisfy (\ref{eq:values}), then its not hard to see that
\[ \sum_{i=1}^h \ind\{ \hat{\ell_i} > 1.5 \} \cdot \hat{\ell}_i \cdot \left( \hat{\mu}_i + \frac{d}{2^i} \right) \]
is an $\tilde{O}(\log n)$-approximation to $\MST(X)$. We note that estimating $|L_i|$ is by now standard in the literature, via an $\ell_0$-sketch, and hence the main difficulty is producing estimates $\hat{\mu}_i$. Furthermore, if indeed every $x, y \in X_u$ at depth $i-1 \in \{ 0, \dots, h-1\}$ satisfies $\|x - y\|_1 \leq O(d \log n / 2^i)$, then it suffices to design a linear sketch using $\poly(\log n, \log d)$ space which can generate (approximate) samples from $\| \boldr_{\bv} - \bc_{\pi(\bv)} \|_1$ for $\bv \sim L_i$ and $\boldr_{\bv} \sim X_{\bv}$ and $\bc_{\pi(\bv)} \sim X_{\bv'}$ for $\bv' \sim C(\pi(\bv))$ (we will formalize this notion soon in Lemma~\ref{lem:sample}). This is true because each sample will be bounded by $O(d \log n / 2^i)$, so that the empirical mean of $\polylog(n)$ such samples will be a desired additive $d / (2 \cdot 2^i)$ approximation. 

We now specify how to generate samples for estimating $\| \boldr_v - \bc_{\pi(\bv)}\|_1$. One may think of sampling $\pi(\bv)$ and $\bc_{\pi(\bv)}$ first, and then drawing $\bv$ and $\boldr_{\bv}$ after. In order to do this, we let $\calP_i$ be the distribution supported on $L_{i-1}$ given by $\bu \sim \calP_i$ being set to $u \in L_{i-1}$ with probability proportional to $|C(u)|$. Furthermore, we consider 
\[ \alpha_i \eqdef \frac{2^i}{d \log^3 n} \in [0, 1],\]
and we denote $\calS_{i}$ the distribution over subsets of $[d]$ which include each coordinate $k \in [d]$ i.i.d with probability $\alpha_i$. For a subset $S \subset [d]$, the character $\chi_S \colon \{0, 1\}^d \to \{-1,1\}$ is the function  $\chi_S(x) = (-1)^{\sum_{k \in S} x_k}$. Then, we will estimate the distance $\| \boldr_{\bv} - \bc_{\bu}\|_1$ by considering the values of $\chi_{\bS}(\boldr_{\bv})$ and $\chi_{\bS}(\boldr_{\bv'})$ for random choice of $\bS \sim \calS_i$. Formally, by setting of $\alpha_i$ and the fact $\| \boldr_{\bv} - \bc_{\bu}\|_1 \leq O(d\log n / 2^i)$,
\begin{align*} 
\Ex_{\substack{\bv \sim L_i \\ \boldr_{\bv} , \bc_{\pi(\bv)}}}\left[ \| \boldr_{\bv} - \bc_{\pi(\bv)}\|_1 \right] &= \Ex_{\bu \sim \calP_i}\left[ \Ex_{\substack{ \bv, \bv' \sim C(\bu) \\ \boldr_{\bv} \sim X_{\bv} \\ \boldr_{\bv'} \sim X_{\bv'}}}\left[ \| \boldr_{\bv} - \boldr_{\bv'}\|_1 \right]\right] \\
&= \Ex_{\bu \sim \calP_i}\left[ \Ex_{\substack{ \bv, \bv' \sim C(\bu) \\ \boldr_{\bv} \sim X_{\bv} \\ \boldr_{\bv'} \sim X_{\bv'}}}\left[ \Prx_{\bS \sim \calS_{i}}\left[ \chi_{\bS}(\boldr_{\bv}) \neq \chi_{\bS}(\boldr_{\bv'}) \right] \right]\right] \cdot \frac{d\log^3 n}{2^i} \pm o\left( \frac{d}{2^i}\right). 
\end{align*}

Specifically, the sketch follows from approximately drawing samples according to the following lemma (Lemma~\ref{lem:sample}). Even though the lemma statement is long, it simply states the fact that we can approximately sample from the distribution needed to estimate the quantity
\begin{align}
\Ex_{\bu \sim \calP_i}\left[ \Ex_{\substack{ \bv,\bv' \sim C(\bu) \\ \boldr_{\bv} \sim X_{\bv} \\ \boldr_{\bv'} \sim X_{\bv'}}}\left[ \Prx_{\bS \sim \calS_{i}}\left[ \chi_{\bS}(\boldr_{\bv}) \neq \chi_{\bS}(\boldr_{\bv'}) \right] \right]\right]  \label{eq:estimation-quantity}
\end{align}
up to additive $o(1/\log^3 n)$ error with $\polylog(n)$ samples. One (minor) issue is that a $\poly(\log n, \log d)$-space sketch cannot specify the vertices $\bv$ or points $\boldr_{\bv}$ because these would require $O(d)$ bits to specify. In order to specify points and vertices we consider a universe reduction step. This is simply a random function $c \colon V_{i-1} \cup V_i \cup \{0,1\}^d \to \{0,1\}^{\polylog(n,d)}$ which is used to rename elements of $L_{i-1}, L_i$ and $X$. In particular, since $|L_{i-1} \cup L_i \cup X|$ is small (at most $\poly(n, \log d)$), the function will be injective on $L_{i-1} \cup L_i \cup X$ with high probability, and thus allows the sketch to refer to nodes implicitly. This is a very minor point, and we will abuse notation by assuming the universe reduction has been applied. 

\begin{lemma}\label{lem:sample}
Consider a fixed quadtree $T$, an index $i \in [h]$, as well as a small parameters $\eps, \delta > 0$. There exists a linear sketch which uses $\poly(\log n, \log d, 1/\eps, \log(1/\delta))$ space and outputs a sample from a distribution $\calD$ supported on $L_{i-1} \times L_i^2 \times X^2 \times \{-1,1\}^2$. The distribution $\calD$ is $\eps$-close in total variation from one where a sample $(\bxi_u, \bxi_{v}, \bxi_{v'}, \bxi_{r_v}, \bxi_{r_{v'}}, \bxi_{\chi_{v}}, \bxi_{\chi_{v'}}) \sim \calD$ satisfies:
\begin{enumerate}
\item The value $\bxi_u$ represents a sample of $\bu$, i.e., $\Prx\left[  \bxi_u = u \right] = |C(u)| / |L_i|$.
\item The value $\bxi_{v}$ represents a sample of $\bv$, i.e., $\Prx\left[ \bxi_v = v \mid \bxi_u = u \right] = 1 / |C(u)| $.
\item The value $\bxi_{r_v}$ represents a sample of $\boldr_{\bv}$, i.e., $\Prx\left[ \bxi_{r_v} = r_v \mid \bxi_{v} = v \right] = 1 / |X_v| $.
\item The value $\bxi_{v'}$ is independent of $\bxi_v$, and represents a sample of $\bv'$, i.e., $\Prx\left[ \bxi_{v'} = v \mid \bxi_u = u \right] = 1 / |C(u)| $.
\item The value $\bxi_{r_{v'}}$ represents a sample of $\boldr_{\bv'}$, i.e., $\Prx\left[ \bxi_{r_{v'}} = r_{v'} \mid \bxi_{v'} = v' \right] = 1 / |X_{v'}| $.
\item The value $\bxi_{\chi_{v}}, \bxi_{\chi_{v'}} \in \{-1,1\}^2$ represent samples of $\chi_{\bS}(\boldr_{\bv})$ and $\chi_{\bS}(\boldr_{\bv'})$, i.e., for any $(b_u, b_v) \in \{-1,1\}^2$, 
\[ \Prx\left[ (\bxi_{\chi_{v}}, \bxi_{\chi_{v'}}) = (b_u,b_v)\mid (\bxi_{r_v}, \bxi_{r_{v'}}) = (r_v, r_{v'}) \right] = \Prx_{\bS \sim \calS_{i}}\left[ (\chi_{\bS}(r_{v}), \chi_{\bS}(r_{v'})) = (b_v, b_{v'})\right] . \]
\end{enumerate}
\end{lemma}

Given Lemma~\ref{lem:sample}, the sketch is straight-forward. We independently consider $t= \polylog(n)$ draws from the distribution $\calD$ with $\eps = 1/\polylog(n)$ and $\delta = 1/\poly(n)$, and we count the number of times a sample $(\bxi_u, \bxi_{v}, \bxi_{v'}, \bxi_{r_v}, \bxi_{r_{v'}}, \bxi_{\chi_v}, \bxi_{\chi_{v'}}) \sim \calD$ satisfies $\bxi_{\chi_v} \neq \bxi_{\chi_{v'}}$. This count, divided by $t$, will be an appropriate estimate to (\ref{eq:estimation-quantity}), and completes the high-level plan of the sketch. 

%% file: main-lemma.tex

\subsection{Proof of Lemma~\ref{lem:sample}}

We will prove Lemma~\ref{lem:sample} in steps, corresponding to the six itemized elements needed for the lemma. Hence, a sample $(\bxi_u, \bxi_v, \bxi_{v'}, \bxi_{r_v}, \bxi_{r_{v'}}, \bxi_{\chi_v}, \bxi_{\chi_{v'}}) \sim \calD$ will proceed in six steps (even though all are correlated). In order to handle the correlation, the first step of generating a sample is to perform the following procedure.

\begin{figure}[h!]
	\begin{framed}
		\begin{flushleft}
			\noindent {\bf Sample:} Consider an independent sample $\bt_{u} \sim \Exp(1)$ for each $u \in V_{i-1}$. Furthermore, for each $\kappa \in \{ 0, \dots, \lceil \log_2 n \rceil \}$ and $j = O(\log(n))$, we sample two random subsets $\bD_{\kappa, j}, \bD'_{\kappa, j} \subset V_i $ by including each node $v \in V_i$ i.i.d with probability $1/2^{\kappa}$. We will consider the following three events, where we specify them with respect to a parameter $\gamma > 0$ which will be set to $1/\polylog(n)$.
						
			\begin{itemize}
				\item Event $\calbE_1$: we have 
				\[ \sum_{u \in L_{i-1}} \frac{|C(u)|}{\bt_{u}} \leq \frac{4\log\left(n/\gamma\right)}{\gamma} \cdot |L_i|. \]
				\item Event $\calbE_2$: let $\bu^* \in L_{i-1}$ be the non-empty node of $T$ which maximizes $|C(\bu^*)| / \bt_{c(\bu^*)}$, and let $\bu^{**}$ be the non-empty node of $T$ which has the second largest $|C(\bu^{**})|/\bt_{c(\bu^{**})}$. Then, we have
				\begin{align}
				\frac{|C(\bu^*)|}{\bt_{\bu^*}} \geq \gamma \cdot |L_{i}| \qquad\text{and}\qquad \frac{|C(\bu^*)|}{\bt_{\bu^*}} \geq \left(1 + \gamma\right) \cdot \frac{|C(\bu^{**})|}{\bt_{\bu^{**}}}. \label{eq: E_2_cond}
				\end{align} 
				We note that we will aim to set $\bxi_{u} = \bu^*$.
				\item Event $\calbE_3$: For every $\kappa \in \{ 0, \dots, \lceil \log_2 n \rceil \}$ and every $j$, we have for $\bR \in \{ \bD_{\kappa, j}, \bD'_{\kappa, j}\}$,
				\[ \sum_{\substack{u \in L_{i-1}}} \frac{|\bR \cap C(u)|}{\bt_{u}} \leq \frac{\log^3 n}{2^{\kappa} \cdot \gamma} \sum_{u \in L_{i-1}} \frac{|C(u)|}{\bt_{u}}.  \]
				Furthermore, there exists a maximum $\kappa^* \in \{0, \dots, \lceil \log_2 n \rceil \}$ and $j \geq j'$ (with respect to lexicographic ordering) such that
				\[ \left| \bD_{\kappa^*, j}  \cap \left\{ v : v \in C(\bu^*) \right\} \right| = \left| \bD_{\kappa^*, j'}' \cap \left\{ v' : v' \in C(\bu^*) \right\} \right|= 1,\]
				and this maximum satisfies $2^{\kappa^*} \geq |C(\bu^*)|$. Let $\bv^* \in C(\bu^*)$ be that unique child of $\bu^*$ specified by $\bD_{\kappa^*, j}$ and $\bv^{**} \in C(\bu^*)$ be the unique child of $\bu^*$ specified by $\bD_{\kappa^*, j'}'$. We will aim to set $\bxi_v = \bv^*$ and $\bxi_{v'} = \bv^{**}$.
			\end{itemize}
		We note that the success of the entire sampling procedure is dependent on the collection $(\bt_u)_{u \in V_{i-1}}$, $(\bD_{\kappa,j})_{\kappa,j}$, and $(\bD_{\kappa,j}')_{\kappa,j}$, and the events $\calbE_1,\calbE_2,$ and $\calbE_3$ occurring. We will continually refer to $\bu^*, \bv^*$, and $\bv^{**}$ and $\kappa^*$ as specified.
		\end{flushleft}\vskip -0.14in
	\end{framed}\vspace{-0.2cm}
	\caption{Sampling for the Parent and Children.}\label{fig:sample-exp}
\end{figure}

With the above sampling procedure and specified events, we will first show that the distribution of $\bu^*$, $\bv^*$, and $\bv^{**}$ generated above are appropriate, and that events $\calbE_1$, $\calbE_2$ and $\calbE_3$ occur with high enough probability. Then, we show that there exists a linear sketch which can recover $\bu^*$, $\bv^*$, and $\bv^{**}$; and finally, we show how to sample the representative elements.

\begin{lemma}[Parent and Child Distribution]
Let $u \in L_{i-1}$ and $v, v' \in C(u)$. Then, as per Figure~\ref{fig:sample-exp}, we have
\begin{align*}
\Prx_{(\bt_{u})_u}\left[ \bu^* = u \right] = \frac{|C(u)|}{|L_i|}, \qquad\text{and}\qquad \Prx_{\substack{(\bD_{\kappa, j})_{\kappa,j}\\ (\bD_{\kappa, j}')_{\kappa,j}}}\left[ \bv^* = v, \bv^{**} = v' \mid \bu^* = u\right] = \frac{1}{|C(u)|^2}.
\end{align*}
\end{lemma}

\begin{proof} The first item follows from a simple calculation using the CDF of exponential random variables
	\begin{align*}
		\Prx_{(\bt_m)_m}\left[ \bu ^*=u\right]&=	\Prx_{(\bt_m)_m}\left[ \frac{|C(u )|}{\bt_{u}}\ge \frac{|C(u' )|}{\bt_{u' }}\;\; \forall u'\in L_{i-1}\right]=\Ex _{\bt_{u}}\left[ \prod_{u'\in L_{i-1}}\Prx_{\bt_{u'}}\left[ \frac {|C(u)|}{\bt_{u}}\ge \frac {|C(u')|}{\bt_{u'}}\right]\right]\\
		&=\Ex_{\bt_{u}}\left[ \exp\left( -\sum_{u'\neq u}\frac{|C(u')|}{|C(u)|}\cdot \bt_{u}  \right) \right]= \int_0^\infty \exp (-t)\exp\left(-\frac{t}{|C(u)|} \sum_{u'\neq u}|C(u')|\right)dt\\
		&=\frac{|C(u)|}{|L_i|}.  
	\end{align*}
The second item follows from the observation that conditioned on the event that some unique vertex $v\in  C(\bu^*)$ was been picked to $\bD_{\kappa^*, j}$ (or equivalently, $\bD_{\kappa^*, j'}'$), its distribution is uniform over $\{v : v\in C(\bu^*)\}$, due to the fact that each $v\in V_{i}$ is picked independently with the same probability.
\end{proof}

\begin{lemma}\label{lem:GoodEventsLemma}
The events $\calbE_1, \calbE_2$, and $\calbE_3$ occur with probability at least $1 - 8\gamma - O(1/n)$ over the draw of $(\bt_u)_{u \in V_{i-1}}$, $(\bD_{\kappa, j})_{\kappa, j}$, and $(\bD_{\kappa, j}')_{\kappa, j}$.
\end{lemma}

We henceforth consider a fixed setting of $(t_{u})_{u \in V_{i-1}}$, as well as sets $(D_{\kappa, j})_{\kappa, j}$ and $(D_{\kappa, j}')_{\kappa,j}$ which satisfy events $\calbE_1, \calbE_2$ and $\calbE_3$; this also specifies the nodes $u^*$, $v^*$ and $v^{**}$. We remove the boldfaced notation to indicate these are no longer random variables. The next lemmas show how to algorithmically recover the values of $u^*, v^*$, and $v^{**}$ assuming that $\calbE_1, \calbE_2$ and $\calbE_3$ are satisfied, and how to sample their representative points. 

\begin{lemma}[Parent Recovery Lemma]\label{lem:parent-recovery}
There exists a linear sketch which uses $\polylog(n)$ space which, given access to the collection $(t_{u})_{u \in V_{i-1}}$ satisfying $\calbE_1 \wedge \calbE_2$, outputs an element $\bxi_{u} \in V_{i-1}$ which is equal to $u^*$ with high probability over internal randomness of the sketch.
\end{lemma}

\begin{lemma}[Child Recovery Lemma]\label{lem:child-recovery}
There exists a linear sketch which uses $\polylog(n)$ space, which given access to the collection $(t_u)_{u \in V_{i-1}}$ and $(D_{\kappa,j})_{\kappa,j}$ satisfying $\calbE_1 \wedge \calbE_2 \wedge \calbE_3$, outputs an element $\bxi_{v} \in V_{i}$ which is equal to $v^*$ with high probability over internal randomness of the sketch. 
\end{lemma}

\begin{lemma}[Child Representative Sampling Lemma]\label{lem:child-sampling}
There exists a linear sketch which uses $\polylog(n)$ space, which given access to the collection $(t_u)_{u \in V_{i-1}}$ and $(D_{\kappa, j})_{\kappa, j}$, outputs an element $\bxi_{r_v}$ or ``fail''. Whenever events $\calbE_1 \wedge \calbE_2 \wedge \calbE_3$ are satisfied, the sketch does not output ``fail'' with high probability, and the output $\bxi_{r_v} \in X$ is uniformly distributed among $X_{v^*}$. 
\end{lemma}

Lemma~\ref{lem:child-recovery} and Lemma~\ref{lem:child-sampling} may be equivalently stated for $v^{**}$ and $\bxi_{r_{v}'}$. It will be clear from the proof of Lemma~\ref{lem:child-recovery} and Lemma~\ref{lem:child-sampling} that we may apply these same linear sketches to recover $\bxi_{\chi_v}$ and $\bxi_{\chi_{v'}}$ which are distributed as $\chi_{S}$ applied to $\bxi_{r_v}$ and $\bxi_{r_v'}$ for any $S \subset [d]$. 

\input{ProofGoodEvent.tex}

\input{ProofLemma2.3}

\input{ProofLemma2.4}

%% file: ProofGoodEvent.tex

\subsection{Proof of Lemma~\ref{lem:GoodEventsLemma}}
The proof follows from computing the probabilities of events $\calbE_1,\calbE_2,\calbE_3$ and using a union bound. In order to compute the probability of event $\calbE_1$ we will use  lemma~\ref{lem:sum_exponentials}, and use Lemma~\ref{lem:gap_exponentials} to compute the probability of event $\calbE_2$.

\newcommand{\boldE}{\boldsymbol{E}}

%

\begin{lemma} 
For any $\gamma>0$, event $\calbE_3$ holds with probability at least $1-O(\gamma/\log n) - 1/ \poly(n)$.
\end{lemma}

\begin{proof} 
Fix the set of exponential random variables $(t_u)_{u\in V_{i-1}}$ and let $u^*$ be the node that maximize $|C(u)|/t_{u}$ over all $u\in L_{i-1}$. 
For the first part of event $\calbE_3$, consider any $\kappa \in \{ 0 , \dots, \lceil \log_2 n \rceil \}$, and let $\bR \subset V_i$ be a set which includes each $v \in V_i$ i.i.d with probability $1/2^{\kappa}$. Notice that 
\[ \Ex_{\bR}\left[ \sum_{u \in L_{i-1}} \frac{|C(u) \cap \bR|}{t_u}\right] = \frac{1}{2^{\kappa}} \sum_{u \in L_{i-1}} \frac{|C(u)|}{t_u}. \]
Hence, the probability that this sum exceeds $\log^3 n / \gamma$ times the expectation is at most $\gamma / \log^3 n$, and we may therefore union bound over all $O(\log^2 n)$ possible draws of sets $\bR$ for all $\kappa$ and $j$. 

For the second part of event $\calbE_3$, notice that for $\kappa$ satisfying $2^{\kappa} \geq |C(u^*)| \geq 2^{\kappa - 1}$, with constant probability over the draw of $\bR \subset V_i$ generated by including each $v \in V_i$ i.i.d with probability $1/2^{\kappa}$, $|\bR \cap C(u^*)| = 1$. Hence, by repeating for $j = O(\log n)$ times, $\bR$ intersects a unique element of $C(u^*)$ at least twice (thereby setting $j$ and $j'$), and therefore the maximum $\kappa^*$ for which this property holds furthermore satisfies $2^{\kappa^*} \geq |C(u^*)|$. 
\end{proof}

%% file: ProofLemma2.3.tex

\newcommand{\cD}{\calD}
\newcommand{\exx}[2]{\mathop{{\bf E}}_{#1}\left[ #2 \right]}
\newcommand{\prb}[2]{\mathop{{\bf Pr}}_{#1}\left[ #2 \right]}

\subsection{Proof of Lemma~\ref{lem:parent-recovery}}
\newcommand{\vals}{\Gamma}

We refer to $\vals \in \Z_{\geq 0}^{|L_i|}$ as the vector such that $\vals_v = |X_v|$ for each $v \in L_i$. We will refer to Figure~\ref{fig:parent-recovery}. Note that $\vals$ is a linear function of the input, and that the sketch is indeed a linear function of $\vals$ which stores $k \cdot l = \polylog(n)$ many $\ell_p$-sketches of accuracy $(1\pm O(1/\log n))$ for $p = \Omega(1/\log^2 n)$ which succeeds with high probability, and thus the entire sketch uses $\polylog(n)$ space. 


\begin{figure}[h!]
	\begin{framed}
		\begin{flushleft}
			\noindent {\bf Linear Sketch of Lemma~\ref{lem:parent-recovery}:} We receive as input a sequence of numbers $(t_u)_{u \in V_{i-1}}$. We set parameters $\eps, \eps_0 \in (0, 1)$ satisfying $\eps_0 < \eps < \gamma / 10$, as well as $p = \eps / (2\log n)$. (Recall the setting of $\gamma = 1/\polylog(n)$). 
			
			\begin{itemize}
				\item We will instantiate a Count-Sketch data structure with $l = O(\log n)$ independent hash functions into $k = O(\log(n/\gamma) / (\eps_0 \gamma^3))$ buckets. In other words, for every $j \in [l]$, we independently sample a hash function $\bh_j \colon L_{i-1} \to [k]$, and we will maintain some memory corresponding to each hash function $j \in [l]$ and each bucket $q \in [k]$.
				\item For each $j \in [l]$ and $q \in [k]$, we will maintain a $(1\pm \eps_0)$-accuracy $\ell_p$-sketch of the vector $\smash{\vals \in \Z_{\geq 0}^{|L_i|}}$, where we rescale the entry $\vals_v$ by $\smash{1/t_{\pi(v)}^{1/p}}$ if $\bh_j(\pi(v)) = q$, and $0$ otherwise. Specifically, the $\ell_p$-sketch proceeds by maintaining $t_0 = O(\log n / (\eps_0 p)^4)$ copies of the following linear function of $\vals$:
				\begin{itemize}
				 \item For each $t \in [t_0]$, we generate a sequence of independent $p$-stable random variables $(\balpha^{(t,j,q)}_v)_{v \in L_i}$ and we maintain 
				 \[ \bA^{(t)}_{j,q} = \sum_{\substack{u \in L_{i-1} \\ \bh_j(u) = q}} \sum_{\substack{v \in L_i \\ \pi(v) = u}} \frac{\balpha_v^{(t,j,q)} \cdot \vals_v}{t_{u}^{1/p}}. \]
				\end{itemize}
			\end{itemize}
			
\ignore{			We fix parameters $\eps,\eps_0 \in (0,1)$ satisfying $\eps_0 < \eps < \gamma/10$ (recall the definitions of $\calbE_1$ and $\calbE_2$), as well as $p =\Theta(\eps/\log n)$ and . We receive as input a set of exponen $t = \Theta(\frac{\log n}{\eps_0^4 p^4})$, $k = \Theta(\frac{\eps_0}{\log (n/\gamma)\gamma^3})$, and $\ell = O(\log n)$ \\
			\begin{itemize}
				\item	Draw hash functions $h_{i}: L_{i-1} \to [k]$ for $i=1,2,\dots,\ell$.
				\item 	Generate $p$-stable random variables $\{\alpha_v^{(i,j)}\}_{v \in L_{i}}$ for each $i \in [\ell],j \in [t]$. Generate exponentials $\{\bt_u\}_{u \in L_{i-1}}$. 
				\item Initalized tables $\bA^{(1)},\bA^{(2)},\dots,\bA^{(t)}\in \R^{\ell \times k}$ with zero in all entries.
				\item For each $i \in [\ell],j \in [t], q \in [k]$, compute the linear function of $\vals$ given by: 
				\[	\bA_{i,q}^{(j)} = \sum_{\substack{u \in L_{i-1} \\ h_i(u) = q}} \sum_{\substack{v \in L_{i} \\ \pi(v) = u}} \frac{1}{\bt_u^{1/p}} \alpha_v^{(i,j)} \vals_v	\] 	
					
			\end{itemize} }
		\noindent {\bf Reporting Procedure:} Given the $k \cdot \ell$ many $\ell_p$-sketches, corresponding to each hash function and bucket pair, we produce an output $u \in L_{i-1}$ (which we will show will be $u^*$ with high probability when events $\calbE_1$ and $\calbE_2$ are satisfied).
				\begin{itemize}
				\item For each $j \in [l]$ and $q \in [k]$, we run the reporting procedure for the $\ell_p$-sketch corresponding to the $j$-th hash function and $q$-th bucket. Let $(\bB_{j, q})$ be the reported $\ell_p$-sketches.
				\item For each $u \in L_{i-1}$, we consider the collection of $l$ values $(\bB_{j, \bh_j(u)} : j \in [l])$ corresponding to reported $\ell_p$-sketches of where $u$ hashed into. We set $\tilde{\bB}_u$ be the median of the values $(\bB_{j, \bh_j(u)})$, and we report $\argmax_{u \in L_{i-1}} \tilde{\bB}_u$.
				
			\end{itemize}
		\end{flushleft}\vskip -0.14in
	\end{framed}\vspace{-0.2cm}
\caption{Parent Recovery Linear Sketch for Lemma~\ref{lem:parent-recovery}.}\label{fig:parent-recovery}
\end{figure}

\begin{proposition}\label{prop:bucketbound}
Consider any fixed $j \in [l]$ and a fixed hash function $h_j \colon L_{i-1} \to [k]$, and any value of $q \in [k]$. Then, with probability at least $1 - 1/\poly(n)$ over the randomness in the $\ell_p$-sketches, 
\[ \bB_{j, q} = (1 \pm \eps_0) \Big( (1\pm \eps) \sum_{\substack{u \in L_{i-1} \\ h_j(u) = q}} \frac{|C(u)|}{t_u}\Big)^{1/p}. \]
%
	\end{proposition}
\begin{proof}
For $j \in [l]$ and $q \in [k]$, the vector whose $\ell_p$ norm we are estimating is $\tilde{\Gamma}_{j,q} \in \R^{|L_i|}$ where
\[ (\tilde{\Gamma}_{j, q})_v \eqdef \left\{\begin{array}{cc} \Gamma_{v}/t_{\pi(v)}^{1/p} & h_j(\pi(v)) = q \\
		 0 & \text{o.w} \end{array} \right.   \]
Therefore, by the correctness of the $(1\pm \eps_0)$-accuracy $\ell_p$-sketches, we have
\begin{align*} 
\bB_{j,q} &= (1\pm \eps_0) \| \tilde{\Gamma}_{j,q}\|_p = (1 \pm \eps_0) \Big( \sum_{\substack{u \in L_{i-1} \\ h_j(u) = q}} \frac{1}{t_u}\sum_{v \in C(u)} \Gamma_v^p \Big)^{1/p}
\end{align*}
with probability at least $1 - 1/\poly(n)$. The claim then follows from the fact that $\Gamma_v \in \{ 0 \} \cup [1, \dots, n]$, which means $\Gamma_v^p \in \{0 \} \cup [1, 1 + \eps]$, since $p = \eps / (2\log n)$.
%
	
\end{proof}

\begin{lemma}\label{lem:bucketbound2}
	Consider any setting of $(t_u)_{u \in V_{i-1}}$ where $\calbE_1$ and $\calbE_2$ hold, and let $K = |C(u^*)| / t_{u^*}$. 
For any $j \in [l]$, with probability at least $1-1/\poly(n)$ over the randomness of the hash functions $\bh_j$, 
\[ \bB_{j, \bh_j(u^*)}  \geq (1 - \eps_0)(1- \eps)^{1/p} K^{1/p}.\]
Moreover, for any $u \neq u^*$ and $j \in [l]$, with probability at least $9/10 - 1/\poly(n)$ over the draw of $\bh_j$, 
\[ \bB_{j, \bh_{j}(u)} < (1 - \eps_0) (1-\eps)^{1/p} K^{1/p}.  \]
	\end{lemma}

\begin{proof}
By Proposition~\ref{prop:bucketbound}, for any $u \in L_{i-1}$, we have with probability $1-1/\poly(n)$,
\begin{align}
	\bB_{j, \bh_{j}(u)} &= (1 \pm \eps_0) \Big((1 \pm \eps) \cdot \frac{|C(u)|}{t_{u}} + (1\pm\eps) \sum_{\substack{u \in L_{i-1} \setminus \{u\} \\ \bh_j(u) = q}} \frac{|C(u)|}{t_u} \Big)^{1/p}. \label{eq:haha}
\end{align}
When $u = u^*$, we obtain the first claim. 
For the second claim, notice that for any $u \neq u^*$, using $\calbE_2$, we have that $|C(u)|/t_u < (1-\gamma) K$, so plugging into (\ref{eq:haha}) we have the upper bound
\begin{align} 
\bB_{j, \bh_j(u)} \leq (1 + \eps_0) (1+\eps)^{1/p}\Big((1-\gamma) K + \sum_{\substack{u' \in L_{i-1}\setminus \{ u\} \\ h_j(u') = q}} \frac{|C(u')|}{t_{u'}} \Big)^{1/p}.\label{eq:hahah}
\end{align}
%
Now note that by event $\calbE_1$ and $\calbE_2$, and the fact that $\bh_j(u)$ is a uniform hash function on to a universe of size $k$, we have
\[ \Ex_{\bh_j}\Big[\sum_{\substack{u' \in L_{i-1} \setminus \{u\} \\ h_i(u') = q}}   \frac{|C(u')|}{t_{u'}}\Big] \leq \frac{1}{k} \sum_{u' \in L_{i-1} }   \frac{|C(u')|}{t_{u'}} \leq \frac{4 \log(n/\gamma)}{k \gamma^2} \cdot K < \frac{1}{10}\cdot \eps_0 \gamma K.\]
By applying Markov's inequality, with probability $9/10$, $\sum_{u' \in L_{i-1} \setminus \{ u\}} |C(u')| / t_{u'} < \eps_0 \gamma K$. Hence, plugging back into (\ref{eq:hahah}),
\begin{align*}
\bB_{j, \bh_j(u)} < (1 + \eps_0)(1+\eps)^{1/p}( (1-\gamma) K + \eps_0 \gamma K)^{1/p} < (1+\eps_0)(1+\eps)^{1/p}(1 - \gamma + \eps_0 \gamma)^{1/p} K^{1/p},
\end{align*}
which implies our desired bound once $\gamma/10 > \eps > \eps_0$. 
%
%
\end{proof}

We are now ready to prove Lemma \ref{lem:parent-recovery}

\begin{proof}[Proof of Lemma \ref{lem:parent-recovery}] 
By \ref{lem:bucketbound} and a Chernoff bound, for any $u \neq u^*$ we have $\tilde{\bB}_u  <(1-\eps_0)  (1 - \eps)^{1/p}K^{1/p}$ with probability at least $1-\exp(-\Omega(\ell)) > 1-1/\poly(n)$, and we can then union bound so that this holds for all $u \neq u^*$. Secondly, again by  \ref{lem:bucketbound} and a Chernoff bound, we have that with probability at least $1-1/\poly(n)$, $\tilde{\bB}_{u^*}  >(1-\eps_0)  (1 - \eps)^{1/p}K^{1/p}$. It follows that $u^* = \argmax_{u} \tilde{\bB}_u$ with high probability, which completes the proof.
\end{proof}

%% file: ProofLemma2.4.tex

\subsection{Proofs of Lemma~\ref{lem:child-recovery} and Lemma~\ref{lem:child-sampling}}

We will prove Lemma~\ref{lem:child-recovery} and Lemma~\ref{lem:child-sampling} as corollaries of the following sketching lemma, which we prove next.

\begin{lemma}\label{lem:recovery}
There exists a linear sketch using $\poly(\log n)$ space which takes as input a sequence $(t_u)_{u \in V_{i-1}}$ satisfying events $\calbE_1$ and $\calbE_2$, a set $D = D_{\kappa,j} \subset V_{i}$ for some $\kappa, j$ satisfying $\calbE_3$, and a set $P \subset \{0,1\}^d$. Given $u^*$, the sketch recovers a vector $\smash{\bz \in \R^{|P|}_{\geq 0}}$ satisfying the following conditions with high probability:
\begin{itemize}
\item If $v \in C(u^*)\cap D$, every $a \in X_{v} \cap P$ has $\bz_{a} \geq 1/(3t_{u^*})^{1/p}$. 
\item If $v \notin C(u^*) \cap D$, then every $a \in X_{v} \cap P$ has $\bz_a \leq 1/(3 \cdot 2^{\kappa} / |C(u^*)| \cdot t_{u^*})^{1/p}$. 
\end{itemize} 
\end{lemma}

From Lemma~\ref{lem:recovery}, we argue that Lemma~\ref{lem:child-recovery} and Lemma~\ref{lem:child-sampling} are straight-forward. To see why Lemma~\ref{lem:child-recovery} follows, instantiate the above lemma with $P = \{0, 1\}^d$, and use the sketch to recover a vector \smash{$z^{(\kappa, j)} \in \R_{\geq 0}^{|X_{u^*}|}$} for every $\kappa, j$. Whenever $2^\kappa \geq |C(u^*)|$, there is a simple test to determine whether any $v \in C(u^*)$ happens to lie in $D_{\kappa, j}$: check whether \smash{$z^{(\kappa, j)}_{a} \geq 1/(3\bt_{u^*})^{1/p}$} for some $a \in \{0,1\}^d$ whose $i$-th node in $T$ is $v$. Even though the algorithm does not know $|C(u^*)|$ (and hence cannot check that $2^{\kappa} \geq |C(u^*)|$), it may start with the largest $\kappa = \lceil \log_2 n \rceil$ and work its way down, while recovering the sets $C(u^*) \cap D_{\kappa, j}$. By $\calbE_3$, the maximum setting of $\kappa$ and $j,j'$ where $C(u^*) \cap D_{\kappa,j}$ and $C(u^*) \cap D_{\kappa, j'}$ contain a unique element also satisfies $2^{\kappa} \geq |C(u^*)|$. Hence, the maximum value of $\kappa$ for which we find $C(u^*) \cap D_{\kappa, j}$ and $C(u^*) \cap D_{\kappa, j'}$ contain unique elements will recover $v^{*}$ and $v^{**}$.

The proof of Lemma~\ref{lem:child-sampling} follows similarly, except we change the setting of $P$. Namely, for every $\eta \in \{0,\dots, \lceil \log_2 n \rceil \}$, we instantiate the linear sketch of Lemma~\ref{lem:recovery} with $\bP_{\eta} \subset \{0,1\}^d$ being a random set generated by including each point i.i.d with probability $1/2^{\eta}$. Then, we recover the vector \smash{$\bz^{(\eta)} \in \R^{|\bP_{\eta}|}_{\geq 0}$} (obtained from Lemma~\ref{lem:recovery} with $P = \bP_{\eta}$), and we check the first time that $v \in C(u^*) \cap D$ contains a unique $a \in X_{v} \cap \bP_{\eta}$ where \smash{$\bz^{(\eta)}_a \geq 1/(3t_{u^*})^{1/p}$}. When this occurs, the point $a$ is uniformly distributed among all $X_{v}$, and hence we use $a$ as the child representative point $\boldr_v$. Furthermore, it is simple to recover the evaluation of any Boolean function $f \colon \{0,1\}^{d} \to \{-1,1\}$ (in particular, $\chi_{\bS}$) at $\boldr_v$ by utilizing Lemma~\ref{lem:recovery}: we instantiate the lemma again while letting $P = \bP_{\eta} \cap f^{-1}(1)$, and we check whether we can recover the same point $a \in \bP_{\eta}$. If so, then $f(a) = 1$, otherwise, $f(a) = -1$.  

\ignore{
The proof of Lemma~\ref{lem:child-recovery} follows similarly as Lemma~\ref{lem:parent-recovery}, albeit with some minor changes. Recall that the goal is two-fold: 1) identify the maximum $\kappa^*$ and $j$ where $|\bD_{\kappa^*, j} \cap C(\bu^*)| = 1$, and 2) recover $\bD_{\kappa^*, j} \cap C(\bu^*)$. We will do this by describing a linear sketch which takes as input the sequence $(\bt_{u})_{u \in L_{i-1}}$, as well as a set $\bD \subset L_i$ which was generated with sampling probability $\kappa$. The sketch stores a linear function of $\vals$. Given $\bu^*$, the recovery procedure outputs a set $\bS \subset C(\bu^*)$. We will show that $\bS \supset \bD \cap C(\bu^*)$ when $\calbE_1, \calbE_2,$ and $\calbE_3$ are satisfied. Furthermore, if $|C(\bu^*)| \leq 2^{\kappa}$, $\bS = \bD \cap C(\bu^*)$. Hence, in order to find the maximum $\kappa^*$ and $j$ where $|\bD_{\kappa, j} \cap C(\bu^*)| = 1$, we start with the maximum $\kappa$ and $j$ and work our way down, stopping the first time we encounter a set $|\bS| = 1$. By $\calbE_3$, the first time we have $|\bD_{\kappa, j} \cap C(\bu^*)| = 1$, we will have $|C(\bu^*)| \leq 2^{\kappa}$ and we exactly decode $\bD_{\kappa, j} \cap C(\bu^*)$ which gives the desired $\bv$. We similarly repeat the procedure for $\bv'$.
}

We now prove Lemma~\ref{lem:recovery}, where we specify the sketch in Figure~\ref{fig:child-recovery}.




\begin{figure}[h!]
	\begin{framed}
		\begin{flushleft}
			\noindent {\bf Linear Sketch for Lemma~\ref{lem:recovery}:} We receive as input a sequence of numbers $(t_u)_{u \in V_{i-1}}$, a set $D \subset V_i$, and a set $P \subset \{0,1\}^d$. We set parameters $\eps = \eps_0  = 1/2$ as well as $p = \eps / (2\log n)$. 
			
			\begin{itemize}
				\item We instantiate a Count-Sketch data structure with $l = O(\log n)$ independent hash functions into $k = O(\log^3 n\log(n/\gamma) / (\gamma^3))$ buckets. 
				\item For each $j \in [l]$ and $q \in [k]$, we maintain a $(1\pm \eps_0)$-accuracy $\ell_p$-sketch of the vector $\smash{\tilde{\vals}_{j,q} \in \Z_{\geq 0}^{|L_i|}}$, where for every $v \in L_i$
				\[ (\tilde{\vals}_{j,q})_v = \left\{ \begin{array}{cc} \ind\{ v \in D\} \cdot \frac{|X_v \cap P|}{t_{u}^{1/p}} & \bh_j(v) = q  \vspace{0.2cm} \\
				0 & \text{o.w.} \end{array} \right. .\]
			\end{itemize}
			
\ignore{			We fix parameters $\eps,\eps_0 \in (0,1)$ satisfying $\eps_0 < \eps < \gamma/10$ (recall the definitions of $\calbE_1$ and $\calbE_2$), as well as $p =\Theta(\eps/\log n)$ and . We receive as input a set of exponen $t = \Theta(\frac{\log n}{\eps_0^4 p^4})$, $k = \Theta(\frac{\eps_0}{\log (n/\gamma)\gamma^3})$, and $\ell = O(\log n)$ \\
			\begin{itemize}
				\item	Draw hash functions $h_{i}: L_{i-1} \to [k]$ for $i=1,2,\dots,\ell$.
				\item 	Generate $p$-stable random variables $\{\alpha_v^{(i,j)}\}_{v \in L_{i}}$ for each $i \in [\ell],j \in [t]$. Generate exponentials $\{\bt_u\}_{u \in L_{i-1}}$. 
				\item Initalized tables $\bA^{(1)},\bA^{(2)},\dots,\bA^{(t)}\in \R^{\ell \times k}$ with zero in all entries.
				\item For each $i \in [\ell],j \in [t], q \in [k]$, compute the linear function of $\vals$ given by: 
				\[	\bA_{i,q}^{(j)} = \sum_{\substack{u \in L_{i-1} \\ h_i(u) = q}} \sum_{\substack{v \in L_{i} \\ \pi(v) = u}} \frac{1}{\bt_u^{1/p}} \alpha_v^{(i,j)} \vals_v	\] 	
					
			\end{itemize} }
		\noindent {\bf Reporting Procedure:} Given the $k \cdot \ell$ many $\ell_p$-sketches, corresponding to each hash function and bucket pair, as well as the identity of $u^*$, we produce a vector \smash{$\bz \in \R^{|P|}_{\geq 0}$}. 
				\begin{itemize}
				\item For each $j \in [l]$ and $q \in [k]$, we run the reporting procedure for the $\ell_p$-sketch corresponding to the $j$-th hash function and $q$-th bucket. Let $(\bB_{j, q})$ be the reported $\ell_p$-sketches.
				\item For each $v \in C(u^*)$ and $a \in X_v \cap P$, we let $\bz_a$ be set to the minimum value of $(\bB_{j, \bh_j(v)} : j \in [l])$.
				
			\end{itemize}
		\end{flushleft}\vskip -0.14in
	\end{framed}\vspace{-0.2cm}
\caption{Child Recovery Linear Sketch for Lemma~\ref{lem:child-recovery}.}\label{fig:child-recovery}
\end{figure}

\begin{proposition}\label{prop:bucketbound-2}
Consider any fixed $j \in [l]$ and a fixed hash function $h_j \colon L_{i} \to [k]$, any value of $q \in [k]$. Then, with probability at least $1 - 1/\poly(n)$ over the randomness in the $\ell_p$-sketches, 
\[ \bB_{j, q} = (1 \pm \eps_0) \Big( (1\pm \eps) \sum_{u \in L_{i-1}} \frac{|D \cap \{ v \in C(u) : h_j(v) = q, X_v \cap P \neq \emptyset\} |}{t_u}\Big)^{1/p}. \]
%
	\end{proposition}
\begin{proof}
The proof mirrors that of Proposition~\ref{prop:bucketbound}. By the correctness of the $(1\pm \eps_0)$-accuracy $\ell_p$-sketches, we have
\begin{align*} 
\bB_{j,q} &= (1\pm \eps_0) \| \tilde{\Gamma}_{j,q}\|_p = (1 \pm \eps_0) \Big( \sum_{\substack{u \in L_{i-1}}} \frac{1}{t_u}\sum_{\substack{v \in C(u) \\ h_j(v) = q}} \ind\{ v \in D \} \cdot |X_v \cap P|^p \Big)^{1/p}
\end{align*}
with probability at least $1 - 1/\poly(n)$. The claim then follows from the fact that $|X_v \cap P| \in \{ 0 \} \cup [1, \dots, n]$, which means $|X_v \cap P|^p \in \{0 \} \cup [1, 1 + \eps]$, since $p = \eps / (2\log n)$.
%
	
\end{proof}

\begin{lemma}\label{lem:bucketbound}
	Consider any setting of $(t_u)_{u \in L_{i-1}}$ where $\calbE_1, \calbE_2$ are satisfied, and let $D$ be some $D_{\kappa, j}$ where $\calbE_3$ is satisfied. 
For any $j \in [l]$ and any $v \in C(u^*) \cap D$, with probability at least $1-1/\poly(n)$ over the randomness of the hash functions $\bh_j$ and $\ell_p$-sketches, 
\[ \bB_{j, \bh_j(v)}  \geq 1/(3t_{u^*})^{1/p}.\]
Moreover, for any $v \in C(u^*) \setminus D$ and $j \in [l]$, with probability at least $1/2$ over the draw of $\bh_j$ and $\ell_p$-sketches, 
\[ \bB_{j, \bh_{j}(v)} < 1/(3 \cdot 2^{\kappa} / |C(u^*)| \cdot t_{u^*})^{1/p}. \] 
	\end{lemma}

\begin{proof}
By Proposition~\ref{prop:bucketbound}, for any $v \in C(u^*) \cap D$, we have with probability $1-1/\poly(n)$,
\begin{align*}
	\bB_{j, \bh_{j}(v)} &\geq (1 \pm \eps_0) \Big((1 \pm \eps) \cdot \frac{1}{t_{u^*}} + (1\pm\eps) \sum_{u \in L_{i-1} \setminus \{u^*\}} \frac{|D \cap \{ v' \in C(u) : \bh_j(v') = q, X_{v'} \cap P \neq \emptyset\}}{t_u} \Big)^{1/p} \\
	&\geq \frac{(1 - \eps_0)(1 - \eps)^{1/p}}{t_{u^*}^{1/p}} \geq \frac{1}{2} \cdot \frac{1}{(2t_{u^*})^{1/p}} \geq \frac{1}{(3t_{u^*})^{1/p}}. 
\end{align*}
For the second claim, notice that for any $v \in C(u^*) \setminus D$, we have
\begin{align*} 
\bB_{j, \bh_j(v)} \leq (1 + \eps_0) (1+\eps)^{1/p}\Big(\sum_{\substack{u \in L_{i-1}}} \frac{|D \cap \{ v' \in C(u) : \bh_j(v') = \bh_j(v), X_{v'} \cap P \neq \emptyset\}}{t_{u}} \Big)^{1/p}.
\end{align*}
%
We now show that the right-hand side of the above inequality is small with probability at least $1/2$. Using the fact that $\calbE_1, \calbE_2$, and $\calbE_3$ are satisfied,
\begin{align*}
&\Ex_{\bh_j}\Big[\sum_{\substack{u \in L_{i-1}}}\frac{|D \cap \{ v' \in C(u) : \bh_j(v') = \bh_j(v), X_{v'} \cap P \neq\emptyset\}}{\bt_{u}}\Big] \\
&\qquad\qquad\qquad\leq \sum_{u \in L_{i-1}}   \frac{|C(u) \cap D|}{k\cdot t_{u}} \leq \frac{\log^3 n}{k \cdot 2^{\kappa} \gamma} \sum_{u \in L_{i-1}} \frac{|C(u)|}{t_u} \\
&\qquad\qquad\qquad\leq \frac{4\log^3 n \log(n/\gamma)}{k \cdot 2^{\kappa} \cdot \gamma^2} \cdot |L_i| \leq \frac{4 \log^3 n \log(n/\gamma)}{k \cdot 2^{\kappa} \cdot \gamma^3} \cdot \frac{|C(u^*)|}{t_{u^*}} \\
&\qquad\qquad\qquad\leq \frac{1}{30 \cdot 2^{\kappa} / |C(u^*)| \cdot t_{u^*}}.
\end{align*}
Hence, with probability at least $2/3 - 1/\poly(n) \geq 1/2$ over the choice of hash functions and $\ell_p$-sketch, $\bB_{j, \bh_j(v)} \leq (3/2) (3/2)^{1/p} / (10 \cdot 2^{\kappa}/ |C(u^*)| \cdot t_{u^*})^{1/p} \leq 1/(3 \cdot 2^{\kappa} / |C(u^*)| \cdot t_{u^*})^{1/p}$. 
\end{proof}



%% file: lower-bound.tex

\section{Lower Bound for Streaming MST}

Our main theorem is the following:

\begin{theorem} \label{thm:main_LB}
Let $n, d \in \N$ satisfy $n \leq 2^{3d/4}$  and $\ell \in \N$. A randomized $\ell$-bit streaming algorithm which estimates the MST of $n$ points in $\{0,1\}^d$ up to multiplicative factor $\alpha > 1$ with probability at least $2/3$ must satisfy
\[ \ell \geq \frac{d}{\alpha} - \log \alpha. \]
In particular, if $\ell$ is a constant independent of $n$ and $d$, then $\alpha = \tilde{\Omega}(\log n)$.
\end{theorem}

We first state a theorem of \cite{AIK08} giving a lower bound for communication protocols for $\EMD$ over $\{0,1\}^d$. Fix a $d\ge 1$ and $1\le \alpha \le d$.
Let $C \subset \F_2^d$ be a linear code of dimension at least $d / 4$ and weight at least $c d$, for a constant $c \geq 0$. Furthermore, the set $C^{\perp}$ is the orthogonal subspace (i.e., the set of vectors $z \in \{0,1\}^d$ satisfying $\langle z, x \rangle = 0$ for all $x \in C$).
Let $\calT_\eps$ denote the distribution over $\{0,1\}^d$ 
  such that each bit of $\bz\sim \calT_\eps$ is set to $1$ with probability $\eps$ independently.
We define a probability distribution $\calD$  supported on
  $(Z,A,B)$ where $Z\in \{0,1\}$ and $A,B$ are subsets 
   of size $n:=|C^\perp|\leq 2^{3d/4}$. 

Let $\eps:=1/(200\alpha)$.
To draw $(\bZ,\bA,\bB)\sim \calD $, 
\begin{enumerate}
\item We first draw $\bZ\sim \{0,1\}$ uniformly at random.
\item If $\bZ= 1$, 
we sample $\bx \sim \{0,1\}^d$ and $\bz\sim \calT_\eps$, and set $\by=\bx+\bz$,
\[
\bA = \left\{ \bx + c \in \{0,1\}^d : c \in C^{\perp} \right\} 
\quad\text{and}\quad 
\bB = \left\{ \by + c \in \{0,1\}^d : c \in C^{\perp} \right\}.
\]
\item  If $\bZ=0$, we sample $\bx,\by \sim \{0,1\}^d$ independently and set
\[
\bA = \left\{ \bx + c \in \{0,1\}^d : c \in C^{\perp} \right\} \quad\text{and}\quad \bB = \left\{ \by + c \in \{0,1\}^d: c \in C^{\perp} \right\} .
\]
\end{enumerate}

The main lower bound of \cite{AIK08} for $\EMD$ states the following.
\begin{theorem}[Proposition 4.1 of \cite{AIK08}]\label{thm:aik}
Let $\Pi$ be any deterministic 1-bit protocol between two parties which receive a subset of $\{0,1\}^d$ as input.
%
 Then, \[
\Prx_{(\bZ,\bA,\bB)\sim \calD }\big[ \Pi(\bA,\bB)=\bZ\big]\le \frac{1}{2} + 2^{-\Omega(d/\alpha)}.
\]
\end{theorem}

Given the above theorem, we now show that a too-good-to-be-true protocol for Euclidean MST on $\{0,1\}^d$ would imply a protocol contradicting Theorem~\ref{thm:aik} above.  
Given a $k\ge 2$, we define the following distribution $\calD^{(k)}$
  supported on $(Z,A_1,\ldots,A_k)$ where $Z\in \{0,1\}$ and $A_1,\ldots,A_k$ are subsets
  of $\{0,1\}^d$ of size $n$:
\begin{enumerate}
\item We first draw $\bZ\sim \{0,1\}$ uniformly at random.
\item If $\bZ=1$, we draw $\bx\sim \{0,1\}^d$ and $\bz_2,\ldots,\bz_k\sim \calT_\eps$
  independently.
Set $\bx_1=\bx$ and $\bx_{\ell}=\bx_{\ell-1}+\bz_{\ell}$ for each $\ell=2,\ldots,k$, and
\[ \bA_{\ell} = \left\{ \bx_{\ell} + c \in \{0,1\}^d : c \in C^{\perp}  \right\}. \]
\item 
If $\bZ=0$, we draw $\bx_1, \dots, \bx_k \sim \{0,1\}^d$ independently and set
\[ \bA_{\ell} = \left\{ \bx_{\ell} + c \in \{0,1\}^d : c \in C^{\perp} \right\}.\]
\end{enumerate}

\begin{observation}\label{obs:dist-same}
When $k=2$, $\calD^{(k)}$ is exactly $\calD$.
For any $\ell<r\in [k]$, the marginal distribution of $(\bZ,\bA_\ell,\ldots,\bA_{r})$
  as $(\bZ,\bA_1,\ldots,\bA_k)\sim \calD^{(k)}$ is the same as $\calD^{(r-\ell+1)}$.
\end{observation}

\begin{observation}\label{obs:gap-instances}
For any $k \geq 2$, the following holds with probability at least $1-o(1)$ over the draw of a sample $(\bZ, \bA_1, \dots, \bA_k) \sim \calD^{(k)}$: the set $\bA = \bA_1 \cup \dots \cup \bA_k$ has a minimum spanning tree of cost $O(\eps d n k + nd)$ if $\bZ = 1$, and minimum spanning tree of cost at least $\Omega(dnk)$ if $\bZ = 0$.
\end{observation}

In our main technical lemma which follows, we consider $k$-party, one-way protocols. There are $k$ players (labeled $1,\dots, k$), and player $\ell$ is allowed to send a single message to player $\ell+1$. In particular, the protocol proceeds by the first player communicating to the second player, the second to third, and so on, until the $k$-th player, who produces an output. 

\begin{lemma}[Player Reduction Lemma]\label{lem:player-reduction}
Let $\ell$ be a positive integer and 
  $k=2^b+1$ for some $b\ge 0$.
Let $\Pi$ be any deterministic $\ell$-bit, $k$-party one-way protocol, in which each
  party receives a subset of $\{0,1\}^d$.
If
\[
\Prx_{(\bZ,\bA_1,\ldots,\bA_k)\sim \calD^{(k)} }\big[ \Pi(\bA_1,\ldots,\bA_k)=\bZ\big]\ge \frac{1}{2} + \delta
\]
for some $\delta>0$, 
then there is a deterministic $\ell$-bit, $2$-party one-way protocol $\Pi'$ such that
$$
\Prx_{(\bZ,\bA_1, \bA_2)\sim \calD }\big[ \Pi(\bA_1,\bA_2 )=\bZ\big]\ge \frac{1}{2} + \frac{2\delta}{2^b}>\frac{1}{2}+\frac{2\delta}{k}.
$$
\end{lemma}

We use this lemma to prove the main lower bound:

\begin{proofof}{Theorem~\ref{thm:main_LB}}
Let $k= \Theta(\alpha)$ such that $k=2^b+1$ for some $b$ and $(\bZ,\bA_1,\ldots,\bA_k)\sim \calD^{(k)}$. As per Observation~\ref{obs:gap-instances}, if $\Pi^*$ is a randomized $k$-party, $\ell$-bit one-way protocol for $\alpha$-approximating MST, then 
 one can derive from $\Pi^*$ a deterministic $k$-party, $\ell$-bit one-way protocol $\Pi$ 
  such that 
$$
\Prx_{(\bZ,\bA_1,\ldots, \bA_k)\sim \calD^{(k)} }\big[ \Pi(\bA_1,\ldots,\bA_k )=\bZ\big]\ge \frac{2}{3}-o(1).
$$
Applying Lemma~\ref{lem:player-reduction}, we obtain a deterministic two-party, $\ell$-bit one-way protocol $\Pi'$
  such that 
$$
\Prx_{(\bZ,\bA_1, \bA_2)\sim \calD }\big[ \Pi'(\bA_1, \bA_2 )=\bZ\big]\ge \frac{1}{2}+\Omega\left(\frac{1}{k}\right).
$$
We now reduce the communication to 1 bit, at the cost of lowering the advantage by a factor of $O(2^{\ell})$: the players guess a transcript of $\ell$ bits, and player $1$ sends one bit to verify that their part of the guessed transcript agrees with $\Pi'$ (also in Lemma 4.1 of \cite{AIK08}). In particular, we get a deterministic two-party, one-bit one way protocol $\Pi''$ such that 
$$
\Prx_{(\bZ,\bA_1, \bA_2)\sim \calD }\big[ \Pi''(\bA_1, \bA_2 )=\bZ\big]\ge \frac{1}{2}+\Omega\left(\frac{1}{k2^\ell}\right)
 \ge \frac{1}{2}+\Omega\left(\frac{1}{\alpha 2^\ell}\right).
$$
The lower bound $\ell\ge (d/\alpha)-\log \alpha $ follows from Theorem~\ref{thm:aik} and Observation~\ref{obs:dist-same}, since $\Omega(1/(\alpha 2^{\ell})) \leq 2^{-\Omega(d/\alpha)}$.
 In particular, this implies that $\ell=\omega(1)$ when $\alpha=c\log n/\log\log n$ for some sufficiently small constant $c>0$. 
\end{proofof}

Lemma~\ref{lem:player-reduction} is the corollary of the following lemma, which we prove next.

\begin{lemma}
Let $\ell$ be a positive integer and 
  $k=2k'+1$ for some positive integer $k'$.
Let $\Pi$ be any deterministic $\ell$-bit, $k$-party one-way protocol, in which each
  party receives a subset of $\{0,1\}^d$.
If
\[
\Prx_{(\bZ,\bA_1,\ldots,\bA_k)\sim \calD^{(k)} }\big[ \Pi(\bA_1,\ldots,\bA_k)=\bZ\big]\ge \frac{1}{2} + \delta
\]
for some $\delta>0$, 
then there is a deterministic $\ell$-bit, $(k'+1)$-party one-way protocol $\Pi'$ such that
$$
\Prx_{(\bZ,\bA_1, \ldots,\bA_{k'+1})\sim \calD^{(k'+1)} }\big[ \Pi'(\bA_1,\ldots,\bA_{k'+1} )=\bZ\big]\ge \frac{1}{2} + \frac{ \delta}{2}.
$$
\end{lemma}

\def\frakA{\frak{A}}

\begin{proof}
Let $\frakA$ denote the collection of sets $\{ C^\perp(x) : x\in \{0,1\}^d \}$.
It is easy to verify that $(\bZ,\bA_{k'+1})$ for $(\bZ,\bA_1,\ldots,\bA_k)\sim \calD^{(k)}$
  is  distributed equivalently as $\bZ\sim \{0,1\}$ and $\bA_{k'+1}\sim \frakA$
  uniformly and independently.
We introduce some notation. Given $A\in \frakA$, $z\in \{0,1\}$ and $M\in \{0,1\}^\ell$,
  we write $p(A,z,M)$ to denote the probability of (below is the message of the 
  $k'$th player)
$$
\Pi_{k'}(\bA_1,\ldots,\bA_{k'})=M
$$
when $(\bZ,\bA_1,\ldots,\bA_k)\sim \calD^{(k)}$ conditioning on $\bZ=z$ and $\bA_{k'+1}=A$.
Note that for any $A \in \frakA$ and $z \in \{0,1\}$, $\sum_{M} p(A,z,M)=1$.
We break up the analysis into two cases, corresponding to the value of the quantity
$$
\xi:= \frac{1}{ 4|\frakA|}\sum_{A\in \frakA}\sum_{M\in \{0,1\}^\ell}
\big| p(A,0,M)-p(A,1,M)\big|.
$$
The first case occurs when $\xi\ge \delta/2$; the deterministic protocol $\Pi'$ is obtained
  from $\Pi$ as follows:
\begin{flushleft}\begin{enumerate}
\item Given inputs $A_1,\ldots,A_{k'},A_{k'+1}$, the first 
  $k'$ players simulate $\Pi$; let $M_{k'}$ denote the $\ell$-bit message 
  sent by player $k'$ to player $k'+1$.
\item After receiving $M_{k'}$, player $k'+1$ returns $1$ if $p(A_{k'+1},1,M_{k'})\ge 
  p(A_{k'+1},0,M_{k'})$ and returns $0$ otherwise.
\end{enumerate}\end{flushleft}
The probability that $\Pi'$ succeeds can be expressed as
\begin{align*}
&\frac{1}{ 2|\frakA|} \sum_{A\in \frakA} \max\big\{p(A,0,M),p(A,1,M)\big\}\\
&\hspace{1cm}=\frac{1}{2 |\frakA|} \sum_{A\in \frakA} \frac{p(A,0,M)+p(A,1,M) +\big|p(A,0,M)-p(A,1,M)\big|}{2}\\&\hspace{1cm}=
\frac{1}{2}+ {\xi} \ge \frac{1}{2}+\frac{\delta}{2}.
\end{align*}

For the second case when $\xi<\delta/2$, we use $\Pi$ to obtain the following 
  randomized protocol $\Pi^*$:
\begin{flushleft}\begin{enumerate}
\item Let $A_{k'+1},\ldots,A_{k}$ be the inputs of the $k'+1$ players.
The first player draws a message $\bM$ from the distribution $p(A_{k'+1},0,\cdot)$, i.e., it draws $\bA_{1}, \bA_{2}, \dots, \bA_{k'} \sim \frakA$, and simulates the protocol to generate the message $\bM$ sent by the $k'$-th player.
\item The first player simulates $\Pi$ as player $k'+1$, using $\bM$ as the message from
  player $k'$; the rest of the $k'$ players simulate $\Pi$ as players $k'+2,\ldots,k$.
\end{enumerate}\end{flushleft}
Then one can bound the probability that $\Pi^*$ succeeds by 
  the probability that $\Pi$ succeeds minus 
$$
\frac{1}{2|\frakA|} \sum_{A\in \frakA} \dtv\big(p(A,0,\cdot)-p(A,1,\cdot)\big)
= {\xi} \le \frac{\delta}{2}.
$$
As a result, there exists a deterministic protocol $\Pi'$ that succeeds with
  probability at least $(1/2)+\delta-(\delta/2)\ge (1/2)+(\delta/2)$.
\end{proof}

%% file: l-p-norms.tex
\section{Sketching $\ell_p$ norms for $p$ near $0$}

In this section, we compute the dependence on $p$ for Indyk's $\ell_p$-sketch using $p$-stable random variables. The reason for this section is that we will later utilize this sketch with $p \to 0$ as $n \to \infty$. We first define stable random variables.
\begin{definition}\label{def:stable}
	A distribution $\mathcal{D}$ is said to be $p$-stable if whenever $\bx_1,\dots,\bx_n \sim \mathcal{D}$ are drawn independently and $a \in \R^n$ is a fixed vector, we have
	\[	\sum_{i=1}^n a_i \bx_i \mathop{=}^{d} \|a\|_p \bx	\]
	where $\bx \sim \mathcal{D}$ is again drawn from the same distribution, and the ``$\displaystyle\mathop{=}^d$'' symbol above denotes distributional equality. 
\end{definition}
We refer to \cite{nolan2009stable} for a thorough discussion on $p$-stable random variables, and to \cite{I06} for their use in streaming/sketching computation. Our proof of correctness will utilize the following standard methods for generating $p$-stable random variables.

\newcommand{\btheta}{\boldsymbol{\theta}}

\begin{proposition}[\cite{chambers1976method}]\label{prop:genstable}
	Fix any $p \in (0,1)$. Then a draw from a $p$-stable distribution $\bx \sim \cD_p$ can be generated as follows: (i) generate $\btheta \sim [-\frac{\pi}{2},\frac{\pi}{2}]$, (ii) generate $\boldr \sim [0,1]$, and set 
		\begin{align}
		\bx = \frac{\sin(p \btheta)}{\cos^{1/p}(\btheta)} \cdot \left(\frac{\cos(\btheta(1-p))}{\ln(1/\boldr)}\right)^{\frac{1-p}{p}}. \label{eq:gen-p-stable}
		\end{align}
\end{proposition}

\begin{remark}
By Proposition \ref{prop:genstable}, in order to generate a draw of $|\bx|$, for $\bx \sim \cD_p$, we may generate $\theta \sim [-\frac{\pi}{2},\frac{\pi}{2}]$ and $r \sim [0,1]$, and consider the magnitude of (\ref{eq:gen-p-stable}). 
%
Notice, furthermore, that $\cos(\theta(1-p)) = \cos(-\theta(1-p))$, so (\ref{eq:gen-p-stable}) is symmetric around the origin in $\theta$. This means we may generate $|\bx|$ by drawing $\boldr \sim [0,1]$ and $\btheta \sim [0, \pi/2]$, and outputting $g(\boldr, \btheta)$,
for $g \colon (0, 1) \times (0,\pi/2) \to \R$ be
\begin{align}
g(r, \theta) \eqdef \frac{\sin (p\theta)}{\cos^{1/p}(\theta)} \left(\frac{\cos(\theta(1-p))}{\ln(1/r)} \right)^{\frac{1-p}{p}}. \label{eq:g-def}
\end{align}
One can check, via simple calculus, that the function $g \colon (0, 1) \times (0, \pi/2) \to \R$ is continuous and monotone increasing in both parameters. 
\end{remark}

\ignore{We now prove some facts about $g(r,\theta)$

\begin{proposition}\label{prop:monotoneg}
	Define the function $g(r,\theta)$ as above, with the domain $(0,1) \times (0,\frac{\pi}{2})$. Then $g$ is monotone increasing in both parameters.
\end{proposition}
\begin{proof}
	First, fix any $\theta \in (0,\pi/2)$. Then since $\ln(1/r)$ is decreasing as $r$ increases in $(0,1)$ the value $(\frac{1}{\ln(1/r)})^{(1-p)/p}$ is monotone decreasing in $r$, as needed. For the second part, notice that $\sin(p \theta)$ is monotone increasing in $\theta$. Thus, it suffices to show that the function
	
	\[	\frac{\cos(\theta(1-p))^{1-p}}{\cos(\theta)}	\]
	is monotone increasing.  We have
	
	\[\frac{\partial}{\partial \theta} \frac{\cos ^{1-p}(\theta (1-p))}{\cos (\theta )}=\tan (\theta ) \sec (\theta ) \cos ^{1-p}(\theta (1-p))\]
	
	Note that all three of $\tan (\theta ), \sec (\theta ) ,$ and $\cos(\theta(1-p))$ are positive for $\theta \in (0,\pi/2)$, which completes the claim and the proof of the Proposition
\end{proof}}

\begin{claim}\label{cl:median}
There is some setting $t^* \in [1/10, 9/10]$ satisfying
	\[g\left(t^*,\frac{\pi t^*}{2} \right) \eqdef \median(|\cD_p|)| = \sup\left\{ z \in \R : \Prx_{\bx \sim \cD_p}\left[ |\bx| \leq z\right]\leq 1/2 \right\}. 	\]
\end{claim}
\begin{proof}
	Consider the CDF $F(z) = \prb{\bx \sim \cD_p}{|\bx| \leq z}$ as well as the function $F(g(t, \pi t /2))$. Since $g(t, \pi t/2)$ is continuous and monotone increasing in both parameters and $F(z)$ is continuous and monotone increasing, the function $F(g(t, \pi t /2))$ is continuous and monotone increasing. Furthermore, since $g(\boldr, \btheta)$ generates a draw from $|\bx|$, where $\bx \sim \cD_p$ when $\boldr \sim [0,1]$ and $\btheta \sim [0, \pi /2]$, for any $t\in[0,1]$ we have that $t^2 \leq F(g(t, \pi t/2)) \leq 1-(1-t)^2$. In other words, $F(g(1/10, \pi/20)) < 1/2 < F(g(9/10, 9\pi/20))$, so the result follows from the intermediate value theorem.  
\end{proof}

\begin{lemma}\label{lemma:stablemedianbound}
	Fix any $p \in (0,1)$ and $\eps \in(0, 1)$, and let $R =\median(|\cD_p|)|$. Then there exists a small constant $c \in (0,1)$ independent of $p,\eps$ such that
	
	\[	\prb{X \sim \cD_p}{R \leq |X| \leq (1+  \eps)R} \geq c p^2 \eps^2  \qquad \text{and}\qquad \prb{X \sim \cD_p}{(1-\eps)R \leq |X| \leq R} \geq c p^2 \eps^2  \]
\end{lemma}

\begin{proof}
Let $t^* \in (1/10, 9/10)$ be the parameter satisfying $R = g(t^*, \pi t^*/2)$ as per Claim~\ref{cl:median}. Then, by the definition of $g(r, \theta)$ for generating a draw from $|\calD_p|$, we have
\begin{align*}
\Prx_{\bx\sim \calD_p}\left[ R \leq |\bx| \leq (1+\eps) R \right] &\geq \Prx_{\substack{\btheta \sim [0, \pi/2] \\ \boldr \sim [0,1]}}\left[ \ln R \leq \ln g(\boldr, \btheta) \leq \ln R + \frac{\eps}{2} \right] \\
	&\geq \Prx_{ \boldr \sim [0,1]}\left[ t^* \leq \boldr \leq t^* + \frac{\eps}{2 D} \right] \Prx_{\btheta \sim [0, \pi/2]}\left[ \pi t / 2 \leq  \btheta \leq \pi t/2 + \frac{\pi\eps}{2D}\right] \gsim \frac{\eps^2}{D^2}.
\end{align*}
where $D = \max\left\{ \frac{d}{dt}\left[ \ln g(t, \pi t/ 2) \right](\ell) : \ell \in [1/100, 99/100] \right\}$, and we used the fact $g(r, \theta)$ is continuous and increasing in both variables. Hence, it remains to upper bound $D$, which we do next.
\begin{align*}
\frac{d}{dt} \left[ \ln g(t, \pi t /2) \right] (\ell) &= \frac{p \pi\cdot \cos(p\pi \ell/2)}{2\cdot \sin(p \pi \ell/2)} + \frac{\pi \cdot \sin(\pi \ell / 2)}{2p \cdot \cos(\pi\ell/2)} - \frac{(1-p)^2 \pi \cdot \sin(\pi \ell / 2 (1-p))}{2p \cdot \cos(\pi\ell/2(1-p))} + \frac{1-p}{p \ell \ln(1/\ell)},
\end{align*}
which is $O(1/p)$ for any $\ell \in [1/100, 99/100]$. The bound for the probability $|\bx|$ lies between $(1-\eps)R$ and $R$ follows similarly. 
\end{proof}

We are now ready to prove the correctness of the Indyk $p$-stable sketch. 
\begin{theorem}[Indyk's $p$-Stable Sketch]\label{thm:indykPrelims}
	Consider any $p,\eps,\delta \in (0,1)$, and set $t = O(\log(1/\delta) / (p\eps)^4)$ Let $\bx_1,\dots,\bx_t \sim \cD_p$ independently. Then, with probability at least $1-\delta$, 
	\[1-\eps\leq	\median_{i \in [t]} \left\{	\frac{|\bx_i|}{\median(|\cD_p|)|}	\right\}  \leq 1+\eps.		\]
\end{theorem}
\begin{proof}
	Set $R = \median(|\cD_p|)$. By Lemma \ref{lemma:stablemedianbound}, 
there exists a constant $c>0$ such that for each trial $\bx_i \sim \cD_p$, the probability that $|\bx_i| > (1+\eps )R$ is at most $1/2 - c p^2 \eps^2$, and the probability that $|\bx_i| < (1-\eps)R$ is at most $1/2 - c p^2 \eps^2$.	
	By Chernoff bounds, the number of trials $i \in [t]$ satisfying $|\bx_i| > (1+\eps) R$ is at most $t/2$ with probability at least $1 - \exp(-(\eps p)^4 t /3)$, and similarly, the number of trials $i \in [t]$ satisfying $|\bx_i| < (1-\eps) R$ is at most $t/2$ with probability at least $1 - \exp(-(\eps p)^4 t / 3)$. 
%
When both of these events occur simultaneously, with probability at least $1- 2 e^{- \eps^4 p^4  t/3} > 1-\delta$, the median falls within the range $(1-\eps)R$ and $(1+\eps)R$, and we obtain our desired bound. 
%
%
\end{proof}

%% file: appendix.tex
\section{Embedding $\ell_p^{d}$ into $\{0,1\}^{d'}$}\label{app:ell_p}

\begin{lemma}\label{lem:stability}
Let $p \in (1, 2]$. There exists a distribution $\calD_p$ supported on $\R$ which exhibits $p$-stability: for any $d \in \N$ and any vector $x \in \R^d$, the random variables $\bY_1$ and $\bY_2$ given by
\begin{align*}
\bY_1 &= \sum_{i=1}^d x_i \bz_i \qquad \text{ for $\bz_1, \dots, \bz_d \sim \calD_p$ independently,} \\
\bY_2 &= \bz \cdot \|x\|_p \qquad\text{ for $\bz \sim \calD_p$,}
\end{align*}
are equal in distribution, and $\Ex_{\bz \sim \calD_p}[|\bz|]$ is at most a constant $C_p$.
\end{lemma}

Consider fixed $R \in \R^{\geq 0}$ which is a power of $2$, and for $p \in [1, 2]$, let $\boldf \colon \R^d \to \{0,1\}$ be the randomized function given by sampling the following random variables
\begin{itemize}
\item $\bz = (\bz_1, \dots, \bz_d)$ where $\bz_1, \dots, \bz_d \sim \calD_p$. In the case $p = 1$, we let $\calD_1$ be the distribution where $\bz_i \sim \calD_1$ is set to $1$ with probability $1/d$ and $0$ otherwise. 
\item $\bh \sim [0, R]$, and
\item a uniformly random function $\bg \colon \Z \to \{0,1\}$.
\end{itemize}
and we evaluate
\begin{align*}
\boldf(x) &= \bg \left( \left\lfloor \frac{\sum_{i=1}^d x_i \bz_i - \bh}{R} \right\rfloor \right) \in \{0,1\}.
\end{align*}
If $x, y \in \R^d$ with $\|x - y\|_p \leq R / t_p$ (for a parameter $t_p$ which we specify later). When $p = 1$, we consider $\|x - y\|_{\infty} \leq R/t_1$. Then
\begin{align*}
\Prx_{\bz, \bh, \bg}\left[ \boldf(x) \neq \boldf(y)\right] &=  \frac{1}{2} \cdot \Prx_{\bz, \bh}\left[ \left\lfloor \frac{\sum_{i=1}^d x_i \bz_i - \bh}{R} \right\rfloor \neq \left\lfloor \frac{\sum_{i=1}^d y_i \bz_i - \bh}{R} \right\rfloor \right].
\end{align*}
Consider a fixed $\bz = (\bz_1, \dots, \bz_d)$ and let
\[ \bw = \sum_{i=1}^{d} (x_i - y_i) \bz_i.\] 
Then, we have
\begin{align*}
\Prx_{\bh \sim [0, R]}\left[  \left\lfloor \frac{\sum_{i=1}^d x_i \bz_i - \bh}{R} \right\rfloor \neq \left\lfloor \frac{\sum_{i=1}^d y_i \bz_i - \bh}{R} \right\rfloor \right] &= \max\left\{ \frac{|\bw|}{R}, 1 \right\}, 
\end{align*}
so that
\begin{align*}
\Prx_{\bz,\bh,\bg}\left[ \boldf(x) \neq \boldf(y)\right] &= \frac{1}{2} \Ex_{\bz}\left[ \max\left\{\frac{|\bw|}{R}, 1 \right\} \right].
\end{align*}
When $p \in (1, 2]$, we notice that by Jensen's inequality and Lemma~\ref{lem:stability},
\begin{align*}
\Ex_{\bz}\left[\max\left\{ \frac{|\bw|}{R}, 1\right\} \right] &\leq \max\left\{ \frac{C_p \|x-y\|_p}{R}, 1 \right\} = \frac{C_p\|x-y\|_p}{R},
\end{align*}
and similarly, $\Ex_{\bz}\left[ \max\{ |\bw| / R, 1 \} \right] \leq \|x-y\|_1 / (dR)$ when $p = 1$. Furthermore, for $p \in (1, 2]$, we lower bound the above quantity by 
\begin{align*}
\Ex_{\bz}\left[\max\left\{ \frac{|\bw|}{R}, 1 \right\} \right] &\geq \frac{\|x-y\|_p}{R} \cdot \Ex_{\bz \sim \calD_p}\left[ \max\left\{ |\bz|, t_p \right\} \right] \geq \frac{\|x-y\|_p}{R} \cdot \frac{C_p}{2},
\end{align*}
since $\Ex_{\bz \sim \calD_p}\left[ \max\{ |\bz|, t_p\}\right] \to C_p$ as $t_p \to \infty$, which means that for some constant setting of high enough $t_p$, we obtain the above inequality. In particular, for $p \in (1, 2]$, we may choose $t_p$ to be a large enough constant depending only on $p$. When $p =1$, 
\[ \Ex_{\bz}\left[ \max\left\{ \frac{|\bw|}{R}, 1\right\}\right] \geq \sum_{i=1}^d \frac{1}{d} \left(1 - \frac{1}{d} \right)^{d-1} \max\left\{\frac{|x_i - y_i|}{R}, 1\right\} = \frac{1}{d} \left( 1- \frac{1}{d}\right)^{d-1} \frac{\|x-y\|_1}{R},\]
by setting $t_1 = 1$.
 
For $p \in (1,2]$, we consider concatenating $m$ independent executions of the above construction and consider that as an embedding of $\bphi \colon \R^d \to \{0,1\}^{m}$, we have that for any $x, y \in \R^d$ with $\|x -y\|_p \leq R/ t_p$,
\begin{align*}
\Ex_{\bphi} \left[ \| \bphi(x) - \bphi(y) \|_1\right] &= m \cdot \Prx_{\boldf}\left[ \boldf(x) \neq \boldf(y)\right]  \\
	&\asymp m \cdot \frac{\|x-y\|_p}{R}, 
\end{align*}
where the notation $\asymp$ suppresses constant factors. Suppose $A \cup B \subset \R^d$ is a subset of at most $2s$ points with aspect ratio 
\[ \Phi \eqdef \dfrac{\mathop{\max}_{\substack{x,y \in A \cup B}} \|x-y\|_p}{\mathop{\min}_{\substack{x,y \in A \cup B \\ x \neq y}} \|x-y\|_p}.  \]
Then, we let $R = t_p \cdot \max_{x, y \in A \cup B} \|x-y\|_p$, and
\[ m = O\left( t_p \Phi \log s\right). \] 
For any fixed setting of $x, y \in A \cup B$, the distance between $\bphi(x)$ and $\bphi(y)$ is the number of independent of executions of $\boldf$ where $\boldf(x) \neq \boldf(y)$. Since executions are independent, we apply a Chernoff bound to say that with probability $1 - 1/s^{100}$ every $x, y \in A \cup B$ has $\|\bphi(x) - \bphi(y)\|_1$ concentrating up to a constant factor around its expectation, and therefore, we apply a union bound over $4s^2$ many pairs of points in $A \cup B$ to conclude $\bphi$ is a constant distortion embedding. The case $p =1$ follows similarly, except that 
\[ \Ex_{\bphi}\left[ \|\bphi(x) - \bphi(y)\|_1\right] \asymp m \cdot \frac{\|x-y\|_1}{dR},\]
and $R = 2\max_{x,y \in A \cup B} \|x-y\|_{\infty}$. In summary, we have the following lemmas.

\begin{lemma}
Fix any $p \in (1, 2]$, $n,d \in \N$, and a parameter $\Phi \in \R^{\geq 0}$. There exists a distribution $\calE_{p}$ over embeddings $\bphi \colon \R^d \to \{0,1\}^{d'}$, where
\[ d' = O(\Phi \log n), \]
such that for any fixed set $X \subset \R^d$ of size at most $n$ and aspect ratio in $\ell_p$ at most $\Phi$, a draw $\bphi \sim \calE_p$ is a constant distortion embedding of $X$ into the hypercube with Hamming distance with probability at least $1 - 1/n^{10}$. 
\end{lemma}

\begin{lemma}
Fix any $n,d \in \N$, and a parameter $r, R \in \R^{\geq 0}$. There exists a distribution $\calE_{1}$ over embeddings $\bphi \colon \R^d \to \{0,1\}^{d'}$, where
\[ d' = O\left(\frac{d R \log n}{r}\right), \]
such that for any fixed set $X \subset \R^d$ of size at most $n$ and $\ell_{\infty}$ distance at most $R$ and $\ell_1$ distance at least $r$, a draw $\bphi \sim \calE_1$ is a constant distortion embedding of $X$ into the hypercube with Hamming distance with probability at least $1 - 1/n^{10}$. 
\end{lemma}

\section{Bit Complexity and Randomness}\label{app:formal-models}

In this section, we handle several details relating to the bit complexity and generation of randomness in the streaming model as raised in Section \ref{sec:linearsketching}.
In Sections \ref{sec:EMDSketch} and \ref{sec:MSTSketch}, we analyzed one-and two pass \textit{linear sketching} algorithms for $\EMD$ and $\MST$. However, our algorithm assumed that we could generate real-valued random variables (e.g. Cauchy or exponential random variables). Of course, a small-space streaming algorithm can only generate such variables to finitely many bits of precision. Similarly, we assumed that all random variables generated at the start of the algorithm could be stored, without affecting the space of the algorithm --- this is known as the random oracle model (Definition \ref{def:randoracle}). In this section, we demonstrate how both of these details can be addressed by generating random variables to finite precision, and using pseduo-random generators to avoid the random oracle assumption.

\begin{lemma}\label{lem:bits}
	The one and two-pass linear sketching algorithms in Sections \ref{sec:EMDSketch} and \ref{sec:MSTSketch} can be implemented by generating each random variable to $\polylog(n)$ bits of precision, so that the space required to store the linear sketches is within a $\polylog(n)$ factor of the dimension of the sketch. 
\end{lemma}
\begin{proof}
	We first handle the issue of bit-complexity. Specifically, we demonstrate that it suffices to generate each random variable used in the algorithms to $\delta = \polylog(n)$ bits of precision. Since every entry of the linear sketches is a sum of at most $n$ such variables, multiplied by integer coefficents of the vectors $f_X$ or $f_{A,B}$ which are bounded in magnitude by $n$, the result will follow. Firstly, we argue that every random variable generated by the algorithms are bounded by $2^{\polylog n}$ in magnitude with high probability. Note that the only unbounded distributions that we use in Sections \ref{sec:EMDSketch} and \ref{sec:MSTSketch} are the exponential distribution, and the $p$-stable distribution for $p \in \{1,\poly(\frac{1}{\log n}\}$. Using the cdf of an exponential, $\poly(n)$ exponentials generated in the sketches will be bounded by $\poly(n)$ with probability $1-\exp(-\poly(n))$. Using tail bounds for $p$-stable variables \cite{nolan2009stable}, the same is true for Cauchy ($1$-stable) random variables. Using the method for generating $p$-stable variables from Proposition \ref{prop:genstable}, we have that a $p$-stable variable is given by:
	
	\[		\bx = \sin(p \btheta) \cdot \left(\frac{\cos(\btheta(1-p))}{\cos^{\frac{1}{1-p}}(\btheta)\ln(1/\boldr)}\right)^{\frac{1-p}{p}}\]
where	$\btheta \sim [-\frac{\pi}{2},\frac{\pi}{2}]$ and $\boldr \sim [0,1]$. Note that $|\theta| - \pi/2 > 1/\poly(n)$ and $r > 1/\poly(n)$ imply that $\frac{\cos(\btheta(1-p))}{\cos^{\frac{1}{1-p}}(\btheta)\ln(1/\boldr)}< \poly(n)$ for $p < 1/2$. It follows that $|\bx| < 2^{\polylog n}$ with probability $1-1/\poly(n)$, and then one can union bound over all the $p$-stable variables generated. 

By the above discussion, it follows that generating each random variable to $\delta$-bits of precision results in an additive $1/2^{\polylog n}$ error in each coordinate of the linear sketch $\bS,\bS f$. Moreover, generating the $p$-stalbes to $\delta$ bits of precision can be accomplished by generating $\boldr, \btheta$ to $\poly(\delta)$-bits (see argument in Appendix $A6$ of \cite{kane2010exact}).
 It suffices to prove that this will not effect the output of the algorithm significantly. This argument is standard using Lipschitzness of the output of the linear sketching algorithm on the sketch (see, Appendix $A6$ of \cite{kane2010exact}). The only step of the algorithm which is not Lipschitz in the sketch $\bS f$ is the generation of a sampled edge in the Quadtree from the sketch in the one-pass algorithms (the two-pass algorithm from Section  \ref{sec:EMDSketch} uses the $L_p$ sampler of \cite{JW18} which already uses bounded bit complexity) --- we argue that the probability that the sample changes after truncating is small. To see this, note that our algorithms in Sections \ref{sec:EMDSketch} and \ref{sec:MSTSketch} are only required to recover whenever the events  $\calbE_1$ and $\calbE_2$ hold. By the proofs of
Lemma \ref{FinalEMDlemma2} (for $\EMD$) and Lemma \ref{lem:bucketbound2} (for $\MST$) there will be a $1/\poly(n)$ sized gap between the biggest and second largest estimated coordinate. As a result, whenever these two events hold, the additive $2^{-\polylog n}$ error will not effect the sample obtained by the algorithm. Moreover, when  $\calbE_1$ and $\calbE_2$ do not hold, recall that our algorithm is allowed to output any sampled edge (this error is absorbed into the variational distance of the algorithm). It follows that the remainder of the analysis of the sampling algorithm is unaffected by the additive $2^{-\polylog n}$ to the coordinates of the sketch, which completes the proof.


\end{proof}

\begin{lemma}
	The one and two-pass linear sketching algorithms in Sections \ref{sec:EMDSketch} and \ref{sec:MSTSketch} can be derandomized with an additive $\tilde{O}(d \cdot \polylog n)$ bits of space.
	\end{lemma}

\begin{proof}
 We first handle the randomness required generation of the Quadtree. By Remark \ref{rem:aik-bdirw}, in each depth $t\in [h]$ of the Quadtree we can sample the same set $(\bi_1,\bi_2,\dots,\bi_{2^t}) \sim [d]$ of coordinates for each vertex at that depth (instead of independently sampling coordinates in every vertex at the same depth). Since $h = \log 2d$, the total number of bits we must sample to define an entire quadtree is $2d \log d$. The algorithm can then generate this randomness and store it, resulting in an additive $O(d \log n)$ bits of space
	
	Now that the Quadtree is generated, since the sketches used in Section \ref{sec:usefulsketches} are already derandomized, it remains only to derandomized the exponential and $p$-stable random variables. Since these variables are independently drawn for separate edges of the tree, by the standard reordering argument of Indyk \cite{indyk2006stable} for linear sketches, we can apply Nisan's PRG \cite{nisan1992pseudorandom} to randomize these random variables with only a $\log(S)$ -blowup in the space complexity of the algorithm, where $S$ is the space complexity of storing the sketch $\bS f$. Since the latter is at most $\poly(d,\log n)$, the blow-up is at most a $\log nd$ factor, which completes the proof. 
\end{proof}